\documentclass[11pt]{article}

\usepackage[american]{babel}

\usepackage{amsmath, amsthm, amssymb}
\usepackage{amsfonts}
\usepackage{mathtools}
\usepackage{thmtools}
\usepackage[authoryear,round,comma]{natbib}
\usepackage{enumitem}

\usepackage{microtype}

\usepackage[dvipsnames]{xcolor}
\usepackage{hyperref}
\hypersetup{
  colorlinks=true,
  linkcolor=blue!50!black,
  citecolor=blue!50!black,
  urlcolor=blue!50!black
}

\usepackage{geometry}
\usepackage{booktabs}
\usepackage{needspace}
\usepackage{tikz}
\usetikzlibrary{calc,arrows.meta,decorations.pathreplacing}
\usepackage{etoolbox}
\newtoggle{journal}

\theoremstyle{plain}
\newtheorem{theorem}{Theorem}[section]
\newtheorem{lemma}[theorem]{Lemma}
\newtheorem{proposition}[theorem]{Proposition}
\newtheorem{corollary}[theorem]{Corollary}
\theoremstyle{definition}
\newtheorem{definition}[theorem]{Definition}

\newtheorem{remark}[theorem]{Remark}

\newcommand{\R}{\mathbb{R}}
\newcommand{\Rpp}{\mathbb{R}_{++}}
\newcommand{\Rp}{\mathbb{R}_{+}}
\newcommand{\Bid}{\mathcal{B}}
\newcommand{\St}{\mathcal{S}}
\newcommand{\E}{\mathbb{E}}

\title{Axiomatic Market Making}

\date{December 27, 2025}

\iftoggle{journal}{%
\author{Frank M.\ V.\ Feys \\ \small Independent researcher, Antwerp, Belgium \\ \small \texttt{frank@fmvfeys.com}}
}{\author{Frank M.\ V.\ Feys}}

\begin{document}

\maketitle

\begin{abstract}
\noindent
This paper axiomatizes the bid-ask market maker's quoting rule.
A quoting rule maps the maker's state, namely inventory, belief, variance, trade intensity, and informed-trader fraction, to a bid-ask pair.
Eight natural axioms, together with six environmental assumptions on the maker's inventory cost, force a unique three-parameter family: the mid-quote is linear in inventory, and the spread decomposes additively into inventory and adverse-selection components.
Each of the three parameters is identified from a distinct moment of the observable quoting rule, with the three identifications mutually decoupled.
The eight axioms partition into a four-axiom indispensable core, one structural choice, and three modularity extensions.
Two structural corollaries follow:  
 the latent inventory cost function is recoverable from the limit order book, and a sharp phase transition separates a functioning regime from a frozen one.
A closing meta-theorem identifies four features invariant across all admissible structural primitives within the axiom system.
To our knowledge, this is the first forced-uniqueness axiomatization of the quoting rule.
\end{abstract}

\medskip
\noindent\textbf{JEL classification.}
C72 (Noncooperative Games);
D47 (Market Design);
D82 (Asymmetric and Private Information; Mechanism Design);
G12 (Asset Pricing; Trading Volume; Bond Interest Rates);
G14 (Information and Market Efficiency; Event  Studies; Insider Trading).

\medskip

\noindent\textbf{Keywords.}
Market microstructure;
market making;
bid-ask spread;
adverse selection;
Bertrand competition;
axiomatic characterization;
forced uniqueness;
Legendre-Fenchel duality;
phase transition.

%-------------------------------------------------------------------
\section{Introduction}
\label{sec:introduction}
%-------------------------------------------------------------------

\subsection{The Problem}

\iftoggle{journal}
{
Market making is, almost without exception, an optimization theory.
 \citet{ho1981optimal} derive the spread from a dealer maximizing expected utility over a finite horizon.
\citet{avellaneda2008high} extend this framework to high-frequency  settings with exponential utility and Poisson trade arrivals.
\citet{cartea2015high} refine the inventory-control program by adding running and terminal penalties on accumulated inventory,  with the optimal spread derived from a Hamilton-Jacobi-Bellman equation.
\citet{glosten1985bid} obtain the spread as a Bayesian equilibrium under asymmetric information.
\citet{kyle1985continuous} derives the equilibrium price impact in continuous-time trading with strategic informed traders.
}
{ 
Market making is, almost without exception, an optimization theory.
\citet{ho1981optimal} derive the spread from a dealer maximizing expected utility over a finite horizon.
\citet{avellaneda2008high} extend this framework to high-frequency settings with exponential utility and Poisson trade arrivals.
\citet{cartea2015high} refine the inventory-control program by adding running and terminal penalties on accumulated inventory, with the optimal spread derived from a Hamilton-Jacobi-Bellman equation.
\citet{glosten1985bid} obtain the spread as a Bayesian equilibrium under asymmetric information, and \citet{kyle1985continuous} characterizes equilibrium price impact in a continuous-time model with a single strategic insider trading against a competitive market maker.
The resulting spread formulas are as varied as the primitives that generate them.
}

 Each of these models begins with a primitive specification (a particular utility function, a distributional assumption, or a specific information structure) and then computes the implied quoting rule.
The strength of the optimization approach is its derivation of behavior from explicit primitives; the corresponding weakness is that the output is only as robust as those primitives, since a change in the   utility function, the distribution of beliefs, or the equilibrium concept produces a change in the quoting rule.
A risk-averse dealer with logarithmic utility and exponential trade arrivals produces one spread; a Bayesian specialist facing informed traders produces another;  the connection between them is mediated  by their respective optimization problems rather than  by any structural property of the  quoting rule itself.
The literature on market making thus accumulates as a catalog of models rather than as a structural theory of the quoting rule.

The present paper takes a different approach.
We do not specify a utility function, distribution, or information structure.
Instead, we impose natural axioms on the quoting rule directly and derive  its functional form by forced uniqueness:
 the conjunction of the axioms admits a unique closed-form solution within a finite-dimensional parameter family.
The axioms are minimal in two senses.
First, each captures a clear economic or mathematical desideratum,
 such as translation equivariance, inventory aversion, or the indifference of the maker at the trading margin.
Second, the axioms are irredundant; removing any one enlarges the solution set.
Under the  eight axioms and six environmental assumptions on the cost function,
 the space of admissible quoting rules collapses to a three-dimensional manifold parametrized by three interpretable scalars, each identified from a distinct moment of the observable rule.

\subsection{The Main Characterization}

Let $ s = (q, \mu, \sigma^2, \lambda, \pi)$ denote the maker's state.
Here  $q$ is her current inventory; $\mu$ her conditional belief about the asset value given public information; $\sigma^2$ the variance of that belief; $\lambda$ the trade arrival intensity; and $\pi$ the fraction of informed traders in the trade flow.
A quoting rule is given by  a function $Q\colon \St \to \Bid$, where $\St$ is the state space of tuples 
 $$ s = (q, \mu, \sigma^2,  \lambda, \pi) \in \R \times \R \times \Rpp \times \Rpp \times [0,1],$$ 
 and $\Bid = \{(b,a) \in \R^2 \mid b < a\}$ is the space of valid bid-ask pairs.

Our central result (Theorem~\ref{thm:full}) shows that, under eight axioms together with six environmental assumptions on the maker's cost function, each quoting rule takes the form
\[
   b = \mu - \gamma_0\sigma^2 q -  \sigma^2\!\left( \frac{\kappa_{\mathrm{inv}}}{\lambda} + \kappa_{\mathrm{adv}}\pi\right),
  \qquad
  a = \mu - \gamma_0\sigma^2 q + \sigma^2\!\left(\frac{\kappa_{\mathrm{inv}}}{\lambda} + \kappa_{\mathrm{adv}}\pi\right),
\]
for some $(\gamma_0, \kappa_{\mathrm{inv}}, \kappa_{\mathrm{adv}})  \in \Rpp^3$.
The three parameters  are uniquely identified by the quoting rule.
The first parameter $\gamma_0$ is the maker's inventory risk aversion; the second parameter $\kappa_{\mathrm{inv}}$ is the marginal cost per unit of held inventory per unit of price variance;   the third parameter $\kappa_{\mathrm{adv}}$ is the expected loss per informed trade per unit of price variance.
Each is an economic primitive of natural interest to an applied researcher.
The three parameters live on distinct structural axes: $\gamma_0$ on the maker's preference axis, $\kappa_{\mathrm{inv}}$ on her cost axis, $\kappa_{\mathrm{adv}}$ on the information axis  of the environment.
Preference, cost, and information are the three classical primitives of any microstructure model, and the axioms make the quoting rule depend on them through exactly one scalar each.

\subsection{The Structural Theorems}

The characterization is the entry point to a richer  structural theory.
Five deeper theorems follow.

\paragraph{Forced uniqueness (Theorem~\ref{thm:bijection}).}
The eight-axiom system is in canonical bijection  with the parameter space $\Rpp^3$.
Each of the  eight axioms contributes one functional-form constraint to the space of measurable quoting rules; collectively they identify a three-dimensional manifold.
No axiom is redundant,  as established by an irredundancy proposition that exhibits a measurable counterexample for each of the eight.

\paragraph{Structural separability (Theorem~\ref{thm:separability}).}
The three parameters and the state variables play isolated, nonoverlapping roles in the quoting rule.
The parameter $\gamma_0$ enters only the skew; $\kappa_{\mathrm{inv}}$ only the inventory spread; $\kappa_{\mathrm{adv}}$ only the adverse-selection spread.
The three economic  primitives (preference, cost, information) are mathematically decoupled from one another.
A direct  corollary is that each parameter is identified from a distinct moment of the observable quoting rule, with no need for joint estimation or auxiliary modeling.

\paragraph{Recovery of the cost function  (Theorem~\ref{thm:recovery}).}
The entire latent  inventory cost function $\kappa$~is observable in the limit order book depth profile.
By relaxing the linearity of the cost to   general convexity and applying the indifference axiom at each level of the book, one obtains $\kappa(Q)$ by integration across levels.
The cost function is not  a hidden primitive that must be calibrated separately.
In contrast to Ho-Stoll or Avellaneda-Stoikov, where the cost coefficient is a free parameter to be fit against data, the axiomatic theory predicts that the cost function itself is observable.

\paragraph{Legendre-Fenchel duality (Theorem~\ref{thm:duality}).}
Market making is, at its mathematical core, an instance of convex conjugate  duality applied to the inventory cost.
The forward map $q \mapsto p$ from inventory to marginal price is the derivative  of $\kappa$; the inverse $p \mapsto q$ is the derivative of the  conjugate $\kappa^*$.
The order book reveals  both $\kappa$ and $\kappa^*$.
This places market making within the broader mathematics of convex analysis, alongside the Lagrangian-Hamiltonian duality of classical mechanics, the energy-free-energy duality of thermodynamics, and the utility-indirect-utility duality of consumer theory.

\paragraph{Universal structure (Theorem~\ref{thm:universal}).}
A closing meta-theorem identifies the structural invariants.
Holding fixed the axiom system (A1)--(A8), with the (A3s) variant admitted in place of (A3), and varying the admissible structural primitives $(\kappa, h, \beta)$ across all choices satisfying the cost-structure conditions and the equivariance generator, four features of the quoting rule remain invariant: the inventory half-spread at level $q$ is determined by $\kappa'_+(|q|)$; the spread decomposes additively into inventory and adverse-selection components;  the skew does not depend on $\pi$; and each component is homogeneous of some degree in $\sigma^2$.
These are the universal content of the axiomatic approach within the (A1)--(A8) family, and the standard microstructure models (Ho-Stoll, Avellaneda-Stoikov, Cartea-Jaimungal, Glosten-Milgrom, Kyle) sit within this family as specific structural-primitive choices.

\subsection{The Layered Analysis}

The eight-axiom system admits a  natural three-layer partition, organized by the kind  of denial each axiom admits, and the theory's conclusions partition correspondingly along these layers.
Four core axioms ((A1), (A4), (A7), (A8)) are indispensable   in practice, being either tautological in content or required by competitive operation; one structural-choice axiom ((A3)) selects the $\sigma^2$-scaling exponent of the spread; three modularity axioms ((A2), (A5), (A6)) decouple the inventory channel of the quoting rule from the trading environment.
The strongest empirical content of the theory, namely the quantitative spread formula, the phase transition, and the Bertrand equivalence for the spread, rests on the four core axioms together with the structural choice, leaving the modularity layer unused.
The closed-form skew $m -  \mu = -\gamma_0\sigma^2 q$ and the canonical bijection with $\Rpp^3$ require,  in addition, the three modularity extensions.
Section~\ref{sec:layered-reading} develops this layered reading and   shows that the most robust empirical predictions of the theory survive any weakening of the modularity layer.

\subsection{Empirical Content and the Phase Transition}

The axiomatic characterization carries strong empirical content.
The closed-form quoting  rule~\eqref{eq:characterization} predicts both the mid-quote shift and the spread  as explicit functions of the state.
In particular, the three parameters are recoverable directly from observable quoting behavior by three distinct moments.
The Recovery Theorem extends this: the full inventory cost function (not just its parameter) is observable in the order book.
Beyond identification, a sharp phase transition between functioning and frozen markets follows from the axioms, with explicit thresholds in volatility, informed-flow fraction, and liquidity that the theory predicts and the data can test.
The layered structure of the axioms further gives the theory a diagnostic dimension: empirical violations of its predictions localize to specific axiom layers, with the form of the deviation indicating which axiom is responsible.

A sharp phase transition, Theorem~\ref{thm:phase}, follows from combining the spread formula with a disagreement bound between maker and counterparties.
The market is functioning (generates positive trade volume) if and only if
\[
  \frac{\kappa_{\mathrm{inv}}}{\lambda}  + \kappa_{\mathrm{adv}}\pi < \frac{\delta_{\max}}{2\sigma^2},
\]
where $\delta_{\max}$ denotes the largest spread the market can sustain, twice the maximal divergence between the maker's mid-belief and  that of her counterparties (see Section~\ref{sec:setup}).
The functioning region partitions into an inventory-dominated regime and an adverse-selection-dominated regime, separated by the critical liquidity $\lambda_\times = \kappa_{\mathrm{inv}}/(\kappa_{\mathrm{adv}}\pi)$.
For $\lambda < \lambda_\times$, the spread scales as  $1/\lambda$, with falling trade arrivals widening it sharply.
For $\lambda > \lambda_\times$, the spread saturates at $2\kappa_{\mathrm{adv}}\pi\sigma^2$ in the high-frequency limit, set by the information environment rather than the inventory channel.
Indeed, volatility shocks, informed-flow surges, or liquidity collapses can drive a market across the phase boundary into the frozen regime.
This provides a structural account, at the level of the quoting rule, of market freezes of the type observed during the 2008 financial crisis and various flash-crash events, where spreads widened beyond the point at which willing counterparties could be found.

\subsection{The Dual Economic Foundation}

The axiomatic theory derives the quoting rule from indifference axioms applied to  a single maker.
A natural question is how the resulting formula relates to a competitive-equilibrium derivation.
Theorem~\ref{thm:bertrand} establishes a direct correspondence: the  quoting rule~\eqref{eq:characterization} is the unique symmetric Nash equilibrium of a Bertrand-competitive market with free entry of homogeneous makers, and the axiomatic and competitive theories yield asymptotically equivalent predictions in the continuous-quantity limit.
The correspondence is best read as two encodings of a common  competitive primitive, the per-trade zero-margin condition, rather than as two logically independent routes  to the same conclusion: the break-even content of (A4) and (A8) is itself a  translation of competitive logic into single-maker language.
The axioms can therefore be read as an as-if characterization of competitive behavior, with the  agreement reinforcing the structural content of the characterization.

\subsection{Relation to the Literature}

The setup is in the tradition of microstructure models that  combine inventory effects with adverse selection.
The inventory-cost lineage runs through \citet{garman1976market}, who introduced the structural microstructure framework, \citet{stoll1978supply} and \citet{amihud1980dealership}, who developed the inventory-management foundations of dealer pricing,  \citet{ho1981optimal}, who derived the spread from a utility-maximizing dealer, and \citet{avellaneda2008high} and \citet{gueant2013dealing}, who extended the framework to high-frequency settings.
The adverse-selection lineage runs through \citet{glosten1985bid}, who derived  the spread  as a Bayesian equilibrium under asymmetric information, \citet{kyle1985continuous} and \citet{back1992insider}, who extended the framework to strategic informed traders in continuous time, and \citet{easley1987price} and \citet{easley1996liquidity}, who characterized the empirical content of adverse-selection-based spread decompositions.
The closed-form quoting rule~\eqref{eq:characterization} contains the inventory and adverse-selection components separately and additively, in the structural-decomposition tradition initiated by \citet{madhavan1993analysis} and developed in \citet{huang1997components}.
Each of Ho-Stoll, Avellaneda-Stoikov, and Glosten-Milgrom can be recovered from~\eqref{eq:characterization} under additional structural commitments, such as a specific cost function, a specific variance scaling, or the absence of skew.
Our contribution is to derive~\eqref{eq:characterization}  from minimal axioms, identify the three parameters as the unique parameter space, and characterize the structural properties (separability, duality, phase transition) that follow.

Axiomatic methods have a long lineage in economics.
The  social-choice tradition runs from \citet{arrow1951social} to \citet{sen1970collective}, the bargaining tradition from \citet{nash1950bargaining} to \citet{kalai1975other} and \citet{roth1979axiomatic}, the mechanism-design tradition   from \citet{vickrey1961counterspeculation} and \citet{hurwicz1973design} through \citet{myerson1981optimal}, \citet{maskinriley1984optimal}, \citet{maskin1999nash}, and \citet{maskin2008mechanism}, the belief-aggregation tradition through \citet{genest1986combining}, and the risk-measurement tradition through \citet{artzner1999coherent}.
The decision-theoretic foundations connect via \citet{vonneumann1944theory}, \citet{savage1954foundations}, \citet{anscombe1963definition}, and \citet{gilboa1989maxmin}, each of which derives a functional form on preferences from axioms on  the choice structure.
The general-equilibrium axiomatic foundation traces to \citet{debreu1959theory}; the methodological case for axiomatization is stated in \citet{aumann1985game}.
None of these traditions has been applied,  until now, to  the derivation of the bid-ask spread.
The technique closest in spirit  is the axiomatic approach to coherent risk measures, where four axioms (translation invariance, subadditivity, positive homogeneity, monotonicity) yield a specific functional form on risk measures.
Our axiomatic system has the same structural flavor, with  a larger axiom count to accommodate the richer state space of market making.

The convex-risk-measure tradition extending the coherent framework is the most closely related body of work.
The convex relaxation of \citet{artzner1999coherent}, developed in \citet{follmer2002convex}, \citet{follmer2016stochastic}, and \citet{frittelli2002putting}, drops the positive-homogeneity axiom and admits the general convex case; dynamic extensions to time-varying risk measurement are developed in \citet{cheridito2006dynamic} and \citet{detlefsen2005conditional}.
The structural commitments of these  papers (functional forms  derived   from axioms on preferences over random payoffs) are methodologically parallel to our commitments (functional  form derived from axioms on the quoting rule), and our Legendre-Fenchel duality theorem (Section~\ref{sec:duality}) makes precise the formal connection through the  convex conjugacy of the inventory cost and the half-spread.
The microstructure literature itself, by contrast, has been overwhelmingly  optimization-based, with the canonical foundations in \citet{ohara1995market} and the modern syntheses in \citet{foucault2013market} and \citet{hasbrouck2007empirical} working from utility maximization, Bayesian equilibrium, or strategic learning.
The over-the-counter strand of \citet{duffie2005otc} and \citet{weill2020search} derives spreads from search-theoretic bargaining, with structural commitments different from ours and from the equity-market microstructure tradition.
Our axiomatic approach is, to our knowledge, the first to operate at the level of structural commitments on the quoting rule directly, rather than on the upstream preference functional, the equilibrium-derivation route, or the search-and-matching process.

A companion paper~\citep{feys2026inventory}  develops the dual axiomatization on the objective side.
The two papers form a coherent axiomatic program  for market making.
The companion paper takes the maker's preferences over inventory trajectories as the primitive and derives, from five axioms, the entropic certainty-equivalent objective with running coefficient $\phi = \gamma\sigma^2/2$, unifying the Avellaneda-Stoikov and Cartea-Jaimungal frameworks as two presentations of a single object.
The present paper takes the quoting rule itself as the primitive and derives,  from eight axioms, the closed-form bid-ask  structure of~\eqref{eq:characterization}.
The two axiomatizations operate at different structural levels and  are connected by the Legendre-Fenchel duality developed in Section~\ref{sec:duality}, which translates  between cost-side and price-side representations.
The companion paper supplies the objective from which any optimization-based maker would derive a quoting rule; the present paper supplies the quoting rule that any such maker would post under competitive entry, regardless of the specific optimization route taken to it.

\subsection{Outline}

Section~\ref{sec:setup} fixes notation and primitives, including the  environmental assumptions on the cost function.
Section~\ref{sec:axioms} presents the eight axioms on the quoting rule, with brief motivations for each.
Section~\ref{sec:characterization} states and proves the  Full Characterization Theorem in detail.
Section~\ref{sec:bijection} establishes the axiom-system--parameter-space bijection, together with the irredundancy of the axioms.
Section~\ref{sec:separability} states the Structural Separability Theorem and derives the empirical identification corollary.
Section~\ref{sec:order-book} extends the theory to the full limit  order book and proves the Recovery Theorem.
Section~\ref{sec:duality} establishes  the Legendre-Fenchel duality structure of the theory.
Section~\ref{sec:phase} characterizes the phase transition and the three liquidity regimes.
 Section~\ref{sec:comparative} collects the comparative statics, develops cross-partials and elasticities, compares the axiomatic predictions with existing optimization-based models, and gives an empirical-implementation roadmap.
Section~\ref{sec:further-structural} collects two further structural implications of the characterization, namely the skew-spread ratio and the inventory-martingale property under symmetric arrivals.
Section~\ref{sec:bertrand} establishes the Bertrand equivalence and its  reading as a competitive interpretation of the break-even content of (A4) and (A8).
Section~\ref{sec:layered-reading} identifies the three-layer partition of the axiom system and traces which conclusions follow from which subset of axioms, applying the Arrovian discipline of naming the core.
Section~\ref{sec:universal} states the Universal Structure Theorem, a closing meta-statement on the invariant features of the theory.
Section~\ref{sec:methodology} situates the present work within the axiomatic tradition in economic theory.
Section~\ref{sec:conclusion} concludes.

%-------------------------------------------------------------------
\section{Setup}
\label{sec:setup}
%-------------------------------------------------------------------

Throughout the paper, $\R$ denotes the real line,  $\Rp = [0, \infty)$ the nonnegative reals,  and $\Rpp = (0, \infty)$ the strictly positive reals.

\subsection{The State Space}

The maker quotes a single asset; her state has five variables.
The inventory $q \in \R$ is her  current position, signed so that positive means long.
The mean $\mu \in \mathbb{R}$ is her belief about the asset value, i.e., the conditional expectation of $v$ given her public information.
 The variance $\sigma^2 \in \Rpp$ is the conditional variance of her belief.
The arrival intensity $\lambda \in \Rpp$ is the trade arrival rate, with  $1/\lambda$ the mean inter-trade time.
The informed fraction  $\pi \in [0,1]$ is the fraction of informed traders in the flow.
The state space $\St = \R \times \R \times \Rpp \times \Rpp \times [0,1]$ consists of tuples
$  s = (q, \mu, \sigma^2, \lambda, \pi).$

\subsection{The Quoting Rule}

The space of valid bid-ask pairs is $\Bid = \{(b, a) \in \R^2 \mid b < a\}$.

\begin{definition}
\label{def:quoting-rule}
A \emph{quoting rule} is a Borel-measurable function
\[
  Q\colon \St \to \Bid, \qquad Q(s) = (b(s), a(s)),
\]
where $b(s)$ and $a(s)$ denote the bid and ask at  state $s$.
The \emph{mid-quote} is $m(s) = \tfrac{1}{2}(b(s) + a(s))$ and the \emph{spread} is $\delta(s) = a(s) - b(s)$.
\end{definition}

Note that the two  parametrizations $(b, a)$ and  $(m, \delta)$ carry the same information: 
 the map $(b, a) \mapsto ((b+a)/2,\ a-b)$ is a bijection from $\Bid$ onto $\R \times \Rpp$.
We use whichever is more convenient.

The measurability requirement is the standard regularity condition for functions of the state.
It rules out pathological constructions such as Hamel-basis  solutions to Cauchy's functional equation, which would admit noncanonical skew rules  satisfying the linearity-over-the-rationals content of (A2) without being globally linear.
Every quoting rule arising in microstructure  modeling, in optimization-based theories of market making, or in observed market data satisfies it by construction.
Measurability is therefore essentially without loss of generality, and is standard throughout mathematical economics on continuum state spaces.
Constructing a nonmeasurable function on $\R$ requires the axiom of choice; the assumption rules out only objects with no constructive description, which have never arisen in any applied model of market  making.

We equip the state space $\St$ with the product topology induced by the standard Euclidean topologies on $\R$, $\Rpp$, and $[0,1]$, and the codomain  $\Bid \subset \R^2$ with the subspace topology inherited  from the Euclidean topology of $\R^2$.
Borel measurability of $Q$ is taken with respect to the Borel  $\sigma$-algebras generated by these topologies.
All continuity statements in what follows refer to these   topologies.

\subsection{The Trading Environment}

 We adopt three primitives of the trading environment, in addition to the maker's state.

\paragraph{Informed and uninformed traders.}
Trades arrive at total rate $\lambda$.
A fraction $\pi$ of arriving traders is informed.
Informed traders observe the true asset value $v$ before trading; the remaining fraction $1-\pi$ is uninformed.
Both populations trade whenever profitable at the posted quotes, that is, a trader buys at the maker's ask whenever her belief or true value exceeds the ask, and sells at the maker's bid whenever her belief or true value falls below the bid.

\paragraph{Belief structure.}
The market maker's belief about the asset value is summarized by its first two moments, namely the conditional expectation $\mu$ and the conditional variance $\sigma^2$.
We impose no further parametric distributional commitment beyond the existence of these two moments.
The underlying belief distribution may be Gaussian, Student-$t$, or otherwise, and may even  be the posterior of a nonparametric Bayesian filter on a richer information set, provided only that the first two moments are well defined.
The  characterization theorems below operate on  $(\mu, \sigma^2)$ directly, and the axioms are formulated in terms of these moments.
This two-moment reduction is the substantive content of the setup, in the spirit of mean-variance microstructure: the higher moments of the maker's belief, whatever they may be, do not enter the quoting rule.

\paragraph{Disagreement bound.}
Each counterparty in the market holds her own belief about the underlying asset value. 
We write $\delta_{\max} > 0$ for the maximum admissible spread in the trading environment, defined  as
\[
  \delta_{\max} := 2\, \max_C\,|\mu_C - \mu|,
\]
twice the maximum belief disagreement between the maker and any counterparty $C$ present in the market.
The factor of two reflects that a trade at the ask requires the counterparty's belief to exceed the maker's posted ask, $\mu_C > a$, that is, the disagreement $\mu_C - \mu$ must exceed the half-spread $a - \mu = \delta/2$ (and symmetrically at the bid).
A spread $\delta$ is admissible (consistent with positive trading volume) if and only if $\delta < \delta_{\max}$, equivalently if and only if some counterparty's disagreement exceeds the half-spread.
The quantity $\delta_{\max}$ is a  primitive of the environment; it bounds how  wide a spread can be  while still attracting any counterparty willing to trade.
It plays no role   in the axioms or the characterization theorems, but it is the  key input to the phase-transition analysis of Section~\ref{sec:phase}.

\paragraph{Inventory cost.}
The maker bears a cost of holding inventory.
The cost takes the form $C(q, \sigma^2) = \kappa(q)\sigma^2$, where $\kappa\colon \R \to \Rp$ is convex, nondecreasing in $|q|$,  with $\kappa(0) = 0$, and symmetric  in the sign of inventory.
Multiplicative variance dependence reflects the dollar variance of the held position; the function  $\kappa$ encodes the maker's cost structure beyond pure variance.

\subsection{Environmental Assumptions on the Cost Function}

We collect six assumptions on the cost function as environmental primitives, distinct from the axioms on the quoting rule.

\begin{itemize}
\item[C1.] $\kappa(0) = 0$ (holding zero inventory has zero cost).
\item[C2.] $\kappa$ is symmetric:  $\kappa(-q) = \kappa(q)$ (long and short are symmetric).
\item[C3.] $\kappa$ is nondecreasing in $|q|$ (larger positions cost more).
\item[C4.] $\kappa$ does not depend on $\lambda$ or $\pi$ (cost depends on the maker's position and the price variance, not on the trading environment).
\item[C5.] $\kappa$ is linear in $|q|$: $\kappa(q) = \kappa_{\mathrm{inv}}|q|$ for  some $\kappa_{\mathrm{inv}} > 0$.
\item[C6.] $\kappa$ is convex on $\R$ (so that one-sided derivatives exist and are finite at every interior point).
\end{itemize}

 Conditions  C1, C2, and C3 are minimal regularity conditions.
Condition C6 (convexity) is implied by C5 (linearity) and is stated separately so that the order-book extension of 
 Section~\ref{sec:order-book}, which drops C5 in favor of a general convex cost, can rely on C6 explicitly.
Conditions C4 and C5 are substantive structural commitments and deserve explicit comment.

Condition C4 mirrors (A5) and (A6) on the cost side: neither the cost function nor the skew depends on $\lambda$ or $\pi$.
Each of (A5), (A6), and  C4 expresses a modularity claim that the inventory channel and the trading-environment channel of the market-making problem are structurally orthogonal.
Weakening C4 admits cost structures in which $\kappa$ depends on the speed of trade  arrivals or on the informed-trader fraction; the Avellaneda-Stoikov framework, discussed   in Section~\ref{sec:universal}, is the canonical example, with the inventory cost coefficient varying through $\lambda$.
We list C4 on the cost  side rather than as an axiom on the quoting rule because $\kappa$ is  a primitive of the environment in our setup, and the axioms are reserved for behavioral conditions on the maker's quoting rule.
This distinction follows the standard division  in the axiomatic tradition.
In the social-choice framework of \citet{arrow1951social}, the axioms constrain the  social welfare function, not the underlying individual preferences over which it aggregates.
In the coherent-risk-measure framework  of \citet{artzner1999coherent}, the axioms constrain the risk measure, not the underlying  probability  space on which the random payoffs are defined.
In the bargaining-solution framework of \citet{nash1950bargaining}, the axioms constrain the solution operator, not the bargaining problem (the disagreement point, the feasible set) itself.
In each case, the axioms constrain the chooser's instrument; the environment is taken as primitive.
In particular, the structural force of C4 (and similarly of C5) is no less than that of (A5) or (A6), and the environmental status reflects the kind of object constrained rather than the strength of the constraint.

Condition C5 is the linearity assumption.
In Section~\ref{sec:order-book} we relax C5 to general convexity and recover the entire function $\kappa$ from the limit order book; under C5, $\kappa$ collapses to the scalar $\kappa_{\mathrm{inv}}$.

A remark on the choice of a linear baseline for the inventory cost, in preference to a quadratic one. 
The optimization-based inventory-management literature commonly takes a quadratic cost (Avellaneda-Stoikov, Cartea-Jaimungal), corresponding to a per-unit carrying cost proportional to inventory.
The  axiomatic stance here selects the linear case as the baseline because (A4) equates the half-spread with the marginal carrying cost $\sigma^2\,\kappa'_+(|q|)/\lambda$; under linear cost this is the constant $\kappa_{\mathrm{inv}}\sigma^2/\lambda$, the simplest nontrivial baseline in which the inventory channel of the spread depends on $\lambda$ but not on the maker's current inventory.
The quadratic case is then covered by the convex-cost extension of Section~\ref{sec:order-book}: the recovery theorem identifies the full $\kappa$ from the order book, and quadratic $\kappa$ corresponds  to a particular shape  of the depth profile that is empirically distinguishable from the linear case.

\subsection{Comments on the Setup}

The five state variables $(q, \mu, \sigma^2, \lambda, \pi)$   form the canonical decomposition of the maker's state in market microstructure.
Inventory $q$ captures her current position.
The belief $\mu$ and the variance $\sigma^2$ summarize her view of the asset value.
The arrival intensity $\lambda$ measures market activity,  and the informed fraction $\pi$ measures the information environment.
Each variable corresponds to a primitive of the trading environment.
The   axioms in Section~\ref{sec:axioms} impose structural restrictions on how the quoting rule depends on each.

%-------------------------------------------------------------------
\section{The Eight Axioms}
\label{sec:axioms}
%-------------------------------------------------------------------

We now present the eight axioms on the  quoting rule, organized along two complementary lines.

The first line is conceptual.
The axioms partition into four conceptual groups by economic content. 
These are boundary and equivariance   properties that constrain the basic shape of the rule; monotonicity and linearity conditions that impose regularity on the skew; scale and indifference axioms that fix the variance scaling  and the inventory portion of the spread; and adverse-selection axioms that determine the information portion of the spread.
We  motivate each axiom carefully within these four groups, in the subsections below.

The second line is methodological.
The axioms partition into three nested layers by the kind of denial each  one admits.
The first layer is the \emph{indispensable core}, consisting  of axioms whose denial changes the object of study rather than selecting a different family of quoting rules within it.
These axioms encode a definition, the qualitative content of  adverse  selection, or the break-even condition of competitive operation.
The second layer is the \emph{structural choice}, a single axiom that represents   a conventional selection from a broader axiomatic universe; its denial corresponds to a different family, such as the $\sigma$-scaling Glosten-Milgrom neighbor.
The third layer is the \emph{modularity extensions}, consisting of  axioms that decouple the inventory channel of the quoting rule from the trading-environment channel; their denial yields qualitatively distinct quoting rules within the same underlying structural framework.
The four axioms (A1), (A4), (A7), and (A8) lie in the indispensable core;  the axiom (A3) is the structural choice; the three axioms (A2), (A5), and (A6) are the modularity extensions.

The two lines of organization serve different purposes.
The conceptual grouping aids exposition, since axioms of similar economic content are presented together.
The layered partition aids the methodological reading, since it identifies which subset of axioms suffices to yield each conclusion of the theory.
Conclusions that follow from the core alone are most robust, surviving any weakening of the modularity extensions; conclusions that require the full theory are stronger structural commitments.
The closed-form characterization of Theorem~\ref{thm:full} requires the full theory; the quantitative spread formula and the phase transition require the core together with the structural choice; the additive decomposition of the spread and the qualitative shape of the skew require only  the core.
We trace this layered structure in detail in Section~\ref{sec:layered-reading}.

A characterization theorem shows that a list of conditions forces a unique functional form. 
 Granted the naturalness of the conditions, the conclusions follow by logical necessity, regardless of their apparent strength. 
Our axioms are chosen to be difficult to deny, especially the four core axioms; the reader who accepts each axiom as natural in isolation thereby accepts the closed-form characterization and its structural consequences.
The irredundancy of the axioms is established formally in Section~\ref{sec:bijection}, where we verify that  no axiom is redundant.

\subsection{The Status of the Break-Even Axioms}
\label{subsec:break-even-status}

A careful reader will object that two of our core axioms, (A4) (inventory indifference) and (A8) (adverse-selection break-even), differ in character from the other two.
The inventory-aversion axiom (A1) is a qualitative monotonicity claim on the maker's skew, and the spread-monotonicity-in-$\pi$ axiom (A7) is a qualitative monotonicity claim on the spread.
Each of these is indispensable on definitional grounds: denial contradicts a directly observable feature of any quoting rule used in practice.

Axioms (A4) and (A8), by contrast, posit zero-profit conditions at the trading margin.
A skeptical reader might argue that these are not first-principle constraints on a quoting rule but rather equilibrium conditions on a competitive market.
A monopolistic maker could in principle violate (A4) by charging a strictly positive expected profit on each trade; the indifference condition imposes that she does not.
This concern is sharper for (A4) and (A8) than for the other axioms precisely because the others restrict the shape of the quoting rule, whereas these two fix its level.
The level fixed by (A4) and (A8) is also the only place in the axiom system where competitive content enters the derivation, and the rest of the closed-form characterization is independent of how that level is justified.

Our response is in two parts.
First,  the indifference axioms (A4) and (A8) admit an equilibrium foundation that does not rely on the rest of the axiom system.
Theorem~\ref{thm:bertrand} shows that the same closed-form quoting rule arises as the unique symmetric Nash equilibrium of a Bertrand-competitive market with free entry of homogeneous makers, where the zero-profit condition is a derived feature of competition rather than an imposed axiom.
In that derivation, the risk-aversion parameter $\gamma_0$ enters as an exogenous parameter of the makers' inventory-disutility function, not as an output of the axiomatic derivation; the Bertrand foundation for the break-even content of (A4) and (A8) therefore does not depend on (A1), (A2), (A3), (A5), (A6), and the two derivations agree fully in the continuous-quantity limit.
A reader who accepts the competitive-equilibrium foundation as natural thereby accepts the content of (A4) and (A8), even if she resists labeling them as axioms in the strict sense.

Second, the indifference axioms are natural even at the single-maker level under any of three weakly competitive conditions.
A market in which counterparties have access to alternative liquidity
sources at competitively determined prices  constrains the maker's spread
by the same zero-profit logic, since any pricing  above the competitive
cost would lose flow to the alternative makers.
A  market in which the maker is one of many small  contributors to a consolidated order book inherits a similar Bertrand-style constraint, even without explicit competition among makers.
A steady-state, repeated-interaction market in which the maker must remain solvent over time (the typical   market-making operational reality) imposes time-averaged zero profit by the law of large numbers, since persistent positive margins attract entry and persistent  negative margins force exit.

Axioms (A4) and (A8) are indispensable on competitive grounds: any quoting rule in a functioning competitive market satisfies them,  whether by direct optimization, by Bertrand pressure, or by long-run solvency.
This is a weaker sense than the strict tautological  status of the qualitative axioms (A1) and (A7), but it is strong enough to  support the layered analysis and the methodological program of the paper.
The Bertrand-equivalence theorem (Section~\ref{sec:bertrand}) provides the rigorous foundation; here in the axiom presentation we adopt the indifference axioms as natural and signal that their justification on competitive-equilibrium grounds is developed later.

\subsection{Monotonicity and Linearity}

A second group of axioms imposes regularity on the dependence of the quoting rule on inventory.
Inventory aversion is captured qualitatively by requiring the mid-quote to be strictly decreasing in $q$.
Skew linearity is the substantive structural commitment of this group: the maker's incremental quote-adjustment for one more unit of inventory does not depend on her current inventory level.

We have not included a separate continuity axiom in this group.
Axiom (A2) directly asserts $\mathbb{R}$-linearity of the skew in $q$, which delivers continuity in $q$ for free; 
 continuity of the quoting rule in each of the five state variables is established as a derived property in Theorem~\ref{thm:continuity}.
Continuity is therefore a derived property of the quoting rule, not a separate  commitment, and we shall not list it as an axiom.

\begin{description}[font=\normalfont\bfseries, listparindent=\parindent, itemsep=0.6em]

\item[(A1) Inventory aversion.]
The mid-quote is strictly decreasing in $q$. 
 That is, for every $s = (q, \mu, \sigma^2, \lambda, \pi)$ and every $q' > q$,
\[
  m(q', \mu, \sigma^2, \lambda, \pi) < m(q, \mu, \sigma^2, \lambda, \pi).
\]
\end{description}

The axiom captures the qualitative  content of inventory aversion.
A market maker is, by  definition, a trader who is paid to provide liquidity rather than to take directional positions.
She holds inventory only as a means of providing the liquidity service, and she prefers to keep her inventory  close to a neutral level (typically zero) rather than to accumulate large positions.
Indeed, when she finds herself holding a large positive inventory, she lowers her mid-quote in order to discourage further buys (at the ask) and encourage sells (at the bid), reducing her exposure.
A trader who preferred more inventory to less would not be playing the market-making role at all; she would be a directional trader making a speculative bet on the asset.
The inventory-aversion axiom is indispensable for any participant meaningfully called a market maker.

\begin{description}[font=\normalfont\bfseries, listparindent=\parindent, itemsep=0.6em]
\item[(A2) Linearity of the skew.]
The skew is linear in $q$, that is, there exists a function $\gamma\colon \Rpp  \times \Rpp \times [0,1] \to \R$ such that
\[
  m(s) - \mu = -\gamma(\sigma^2, \lambda, \pi) \cdot q.
\]
\end{description}

The axiom is a Cauchy-equation condition.
The substantive content is that the maker's incremental quote-adjustment for one additional unit of inventory does not depend on her current inventory level: indeed, if she shifts her mid-quote by $\Delta$ when her inventory moves from $0$ to $1$,  then she also shifts by $\Delta$ when her inventory moves from $7$ to $8$, or from $-3$ to $-2$.
This is the natural expression of inventory aversion at the margin; the maker's  local discomfort from holding one additional unit is the same regardless of her starting position.
The axiom is defensible on the grounds that any nonuniformity would  have to be  motivated by an additional structural feature of the maker's environment, and no such feature is present in the axiomatic setup we have specified.
Under measurability, the linearity content of the present axiom makes the skew linear in $q$ as a function, ruling out nonlinear dependencies such as quadratic skew.
We note that the choice of domain for $\gamma$, namely $(\sigma^2, \lambda, \pi)$ and not $\mu$, is itself a substantive component of (A2): it encodes the $\mu$-independence of the skew coefficient, so that the skew-side content of Theorem~\ref{thm:translation-equivariance} (translation   equivariance in $\mu$) is already partly built into (A2), with (A4) and (A8) supplying the remaining spread-side content.

\subsection{Scale and Indifference}

The  third group of axioms identifies the scale on which the maker's quoting behavior  operates and the indifference condition that determines the inventory portion of the spread.
The variance-homogeneity axiom fixes the variance-scaling exponent so that quotes  scale linearly with the price variance $\sigma^2$.
The inventory-indifference axiom is the maker's break-even condition for inventory: the half-spread at zero informed flow equals the marginal carrying cost  per unit time, so that the maker is indifferent at the trading margin between accepting a trade and remaining at her current inventory.
The two skew-decoupling axioms are modularity axioms that assert the skew does not depend on the trading-environment variables $\lambda$ and $\pi$, separating the structural role of inventory positioning from the structural role of the spread.

\begin{description}[font=\normalfont\bfseries, listparindent=\parindent, itemsep=0.6em]

\item[(A3) Variance homogeneity.]
The quoting rule $Q$ is homogeneous of degree one in $\sigma^2$, that is, there exist functions $f, g$ such that
\[
  m(s) - \mu = \sigma^2 \cdot f(q, \lambda, \pi),
  \qquad
  \delta(s) = \sigma^2 \cdot g(q, \lambda, \pi).
\]
\end{description}

The axiom fixes the variance-scaling exponent.
The economic motivation is a unit-of-measurement argument.
If one rescales the price unit, for example by reading prices in cents rather than in dollars, then $\mu$ and $\sigma$ rescale linearly, $\sigma^2$ rescales quadratically, and the quoting rule (a price difference) must rescale linearly in the price unit.
In particular, the natural way to achieve this is for the quoting rule to depend on $\sigma^2$ in a way that, combined with the rescaling, returns a linear-in-unit dependence.
The variance-homogeneity axiom is   the specific scaling that achieves this under the assumption that the quote shift from $\mu$ is proportional to $\sigma^2$ rather than to some other power of $\sigma$.
Alternative scaling exponents, e.g., $\sigma$ instead  of $\sigma^2$, correspond to alternative structural models with different unit-of-measurement properties.
We discuss the Glosten-Milgrom $\sigma$-scaling alternative and its structural interpretation in Section~\ref{sec:universal}, where we identify it as a different choice of structural primitive within the universal   axiomatic family.
The axiom is defensible as the natural baseline scaling for absolute price differences in our moment-based belief setting.

\paragraph{(A3s) Mixed-scaling alternative.}
The Glosten-Milgrom $\sigma$-scaling alternative replaces the uniform $\sigma^2$-scaling of (A3) by a two-component scaling: there exist functions $f, g_{\mathrm{inv}}, g_{\mathrm{adv}}$ and an exponent $\beta \in (0, 1]$ such that
\[
  m(s) - \mu = \sigma^2 \cdot f(q, \lambda, \pi),
  \qquad
  \delta(s) = \sigma^2 \cdot g_{\mathrm{inv}}(q, \lambda, \pi) + \sigma^{2\beta} \cdot  g_{\mathrm{adv}}(q, \lambda, \pi),
\]
with $g_{\mathrm{inv}}$ capturing the inventory contribution and $g_{\mathrm{adv}}$ the adverse-selection contribution.
The Glosten-Milgrom rule corresponds to $\beta = 1/2$,  the case in which the adverse-selection component  scales linearly in $\sigma$ rather than in $\sigma^2$.
We treat (A3s) as the alternative structural primitive discussed in Section~\ref{sec:universal}, where the Glosten-Milgrom model is identified as a structural neighbor of the present axiomatic family.

\begin{description}[font=\normalfont\bfseries, listparindent=\parindent, itemsep=0.6em]
\item[(A4) Inventory indifference.]
The half-spread at $\pi = 0$ (no informed flow) equals the marginal carrying cost per unit time:
\[
  \frac{\delta(s)|_{\pi=0}}{2} = \frac{\sigma^2\,\kappa'_+(|q|)}{\lambda},
\]
where $\kappa'_+$ denotes the right derivative of $\kappa$ at $|q|$, which exists and is finite at every $|q| \in \Rp$ by C6.
\end{description}

The axiom captures the market maker's break-even condition for holding inventory at the trading margin.
At the moment a trade is executed, the maker absorbs (or sheds) a unit of inventory; the corresponding marginal carrying cost per unit time is $\sigma^2\,\kappa'_+(|q|)$, with the trade-arrival timescale $1/\lambda$ converting the per-unit-time cost into a per-trade  cost $\sigma^2\,\kappa'_+(|q|)/\lambda$.
The half-spread is the maker's per-trade revenue, and the axiom asserts that revenue equals marginal cost at the trading margin.
The use of $|q|$ reflects that the marginal carrying cost depends on the  magnitude of inventory, not on its sign: by C2 the cost function $\kappa$ is symmetric   in $q$, so the right derivative $\kappa'_+(|q|)$ is the correct gradient regardless of whether the trade moves  the maker further from or closer to zero inventory.
Under C5 (linear cost, the active regime in Theorem~\ref{thm:full}), we have $\kappa'_+(|q|) = \kappa_{\mathrm{inv}}$ at every $q$, so the right side of (A4) equals $\kappa_{\mathrm{inv}}\sigma^2/\lambda$ exactly, which coincides with the discrete unit-trade cost increment $\sigma^2\,[\kappa(|q|+1) - \kappa(|q|)]/\lambda$.
The continuous-marginal form is  therefore the natural axiomatic  primitive, with the unit-trade reading recovered as its discretization.
The economic justification is twofold.
First, if the half-spread strictly exceeded the marginal carrying cost, then the maker would be earning strictly positive economic profit on each trade, which would attract  entry of competing makers and drive the spread down through Bertrand competition.
Second, if  the half-spread  were strictly less than the marginal carrying cost, then the maker would be losing money on each trade and would withdraw from the market.
In either case, the absence of (A4) is inconsistent with a steady-state market with positive trading volume.
The axiom is therefore indispensable in any competitive or steady-state setting; we make this explicit through the Bertrand  Equivalence Theorem in Section~\ref{sec:bertrand}, where (A4) emerges naturally as the equilibrium condition of competitive entry.
Under linear cost (C5), the right side reduces to $\kappa_{\mathrm{inv}}  \sigma^2/\lambda$.

\begin{description}[font=\normalfont\bfseries, listparindent=\parindent, itemsep=0.6em]
\item[(A5) Skew decoupling from liquidity.]
The skew $m(s) - \mu$ does not depend on $\lambda$,  that is, there exists a function $\tilde f$ such that $m(s) - \mu = \tilde f(q, \sigma^2, \pi)$.
\end{description}

The axiom is a modularity assertion.
The skew is the maker's adjustment to her mid-quote in response to her inventory position; the arrival intensity $\lambda$ is a property of the trading environment, governing how  frequently trades occur.
The axiom says that these two quantities are structurally orthogonal.
 Indeed, the maker's hedging response to her inventory does not depend on how busy  the market is.
 The intuition is that the skew, expressed in price units, encodes a per-unit-of-inventory adjustment reflecting the maker's risk attitude toward her exposure; the speed of trade arrivals does not change that attitude.
The arrival intensity does affect the spread, through the per-trade allocation of the carrying cost in (A4).
That role is separated from the skew.
The axiom is defensible but is a structural assertion rather than a definitional one.
Alternative structures could conceivably tie the skew  to liquidity, for example, a maker might adjust her inventory positioning more aggressively in illiquid markets, but such  structures would require additional modeling primitives not present in the baseline axiomatic setup.

\begin{description}[font=\normalfont\bfseries, listparindent=\parindent, itemsep=0.6em]
\item[(A6) Skew decoupling from information.]
The skew $m(s) - \mu$ does not depend on $\pi$, that is, there exists a function $\hat f$ such that $m(s) - \mu = \hat f(q, \sigma^2, \lambda)$.
\end{description}

The axiom is the analog of (A5) for the information environment.
The maker's inventory positioning, encoded in the skew, is structurally orthogonal to the fraction of informed traders in the trading flow.
The role  of $\pi$ is to widen the spread through adverse selection, not to alter the maker's inventory hedging.
Like (A5), the present axiom is a modularity assertion that separates the inventory channel of the quoting rule from the information channel.
The intuition is that the maker's risk attitude toward her own inventory does not depend on whether the traders she faces are informed or  uninformed.
Each informed trader generates expected loss for the maker through adverse selection; the maker compensates for this loss through a wider spread, not through a different inventory positioning.
The axiom is defensible on these structural grounds.

 \subsection{Adverse Selection}

The fourth and final group of axioms governs the information component of the spread.
The spread-monotonicity axiom captures the qualitative content of adverse selection: 
more informed order flow widens the spread.
The break-even axiom is the maker's break-even condition for adverse selection, together with a complementarity condition that ensures the contributions of distinct informed-trader subpopulations add up.
Together, the two axioms make the adverse-selection spread a linear function   of the informed fraction, with a coefficient that captures the expected information rent per informed trade.

\begin{description}[font=\normalfont\bfseries, listparindent=\parindent, itemsep=0.6em]

\item[(A7) Spread monotonicity in information.]
The spread $\delta(s)$ is strictly increasing in  $\pi$~for $\pi \in [0,1]$.
\end{description}

The axiom is the qualitative content of adverse selection.
Informed traders, by hypothesis, observe the true asset value before trading.
By trading, informed agents extract an expected information rent from the market maker, who sets her quotes on the basis of her publicly available  belief about the asset value.
A higher fraction of informed flow translates into higher expected loss for the maker per trade.
The maker compensates for this loss by widening her spread, increasing her per-trade revenue.
The spread-monotonicity axiom is indispensable once one accepts that informed traders extract expected information rent and that the maker is rational enough to compensate.
The quantitative form of the dependence of $\delta$ on $\pi$ is determined by (A8); the present axiom fixes the qualitative direction.

\begin{description}[font=\normalfont\bfseries, listparindent=\parindent, itemsep=0.6em]
\item[(A8) Adverse-selection break-even.]
Define the adverse-selection increment
\[
h(\sigma^2, \pi) := \delta(q, \mu, \sigma^2, \lambda, \pi) - \delta(q, \mu, \sigma^2, \lambda, 0)
\]
that is, the  excess of the spread at informed fraction $\pi$ over the spread at zero informed fraction, with the remaining state variables $(q, \mu, \sigma^2, \lambda)$ held fixed.
The axiom imposes two clauses on $h$.
(i) \emph{Dependence.} The increment $h$ depends only on $(\sigma^2, \pi)$,  not on $q$ or $\lambda$.
(ii) \emph{Complementarity.} The increment is additive across the informed-trader population: for  any decomposition of the informed fraction $\pi = \pi_1 + \pi_2$ with $\pi_1, \pi_2 \geq 0$ and $\pi_1 + \pi_2 \leq 1$,
\[
  h(\sigma^2, \pi) = h(\sigma^2, \pi_1) + h(\sigma^2, \pi_2).
\]
\end{description}

The axiom  is the maker's break-even condition for adverse selection, together with a complementarity condition on informed-trader   populations.
The first clause says that the adverse-selection spread depends only on the volatility and the informed fraction, not on the  maker's inventory or on the speed of trade arrivals.
This is the analog of (A5) and (A6) applied to the adverse-selection channel: the information environment  generates a particular adverse-selection cost, and that cost depends only on the variables that describe the  information environment itself, not on the inventory  or liquidity channels.
The second clause is a complementarity condition.
Two nonoverlapping populations of  informed traders, of sizes   $\pi_1$ and $\pi_2$, generate together the same total adverse selection as one combined population of size $\pi_1 + \pi_2$.
Each informed trader contributes a fixed expected information rent, regardless  of how many other informed traders are present.

The break-even content of the axiom is indispensable on the same competitive-equilibrium grounds as (A4): any rule violating it cannot persist in  a competitive market.
If the adverse-selection spread strictly exceeded the expected information rent extracted per informed trade, the maker would be earning strictly positive economic profit on adverse-selection-driven trades, attracting entry; if it strictly fell short, the maker would be losing money and would withdraw.
Steady-state operation requires  equality.
The complementarity content is somewhat more structural.
It rules out  interaction effects among informed  traders, for example informed traders competing with one another and reducing each individual's expected rent.
The empirical regularity of complementarity is supported  by the standard models of adverse selection, e.g., Glosten-Milgrom, in which  informed traders are treated as a homogeneous population whose collective impact on the spread is proportional to their fraction of the flow.
We treat the axiom as indispensable in the baseline framework; the relaxation to diminishing complementarity, which captures competition among informed traders, is a structural extension developed in Section~\ref{sec:diminishing-complementarity} and does not affect the qualitative content of the main characterization.

The combined effect of (A7) and (A8),  together with measurability and (A3), is to make $h$ linear in $\pi$ with $\sigma^2$-coefficient.
The result is $h(\sigma^2, \pi) = 2\kappa_{\mathrm{adv}}\pi\sigma^2$ for  some $\kappa_{\mathrm{adv}} > 0$, as established formally in Lemma~\ref{lem:adverse-spread}.

\subsection{Summary of the Axiom System}

The eight axioms partition into three groups, summarized in Table~\ref{tab:axioms}.
The monotonicity  and linearity axioms (A1) and (A2) impose inventory aversion and the linearity of the skew in inventory.
 The scale and indifference axioms (A3)--(A6) fix the variance scaling, the inventory spread,  and the modularity of the skew.
The adverse selection axioms (A7) and (A8) determine the information part of the spread.
Two natural symmetries, translation equivariance in $\mu$ and oddness of the skew in $q$, together with continuity of the quoting rule in each state variable, emerge as derived theorems of the eight-axiom system; they are stated and proved in  Section~\ref{subsec:derived-symmetries}.

\begin{table}[ht]
\centering
\begin{tabular}{lcp{5.0cm}p{5.1cm}}
\toprule
Axiom & Layer & Content & Role \\
\midrule
 (A1) & I & $\partial m/\partial q < 0$ & Inventory aversion \\
(A2) & III & Skew linear in $q$ & Skew linearity \\
(A3) & II & Homogeneity in $\sigma^2$  & Variance scaling \\
(A4) & I & Inventory indifference & Fixes inventory spread \\
(A5) & III & Skew decoupled from $\lambda$ & Modularity \\
(A6)  & III & Skew decoupled from $\pi$ & Modularity \\
(A7) & I & $\partial\delta/\partial\pi > 0$ & Adverse-selection monotonicity \\
(A8) & I & Adverse-selection break-even with complementarity & Fixes adverse-selection spread \\
\bottomrule
\end{tabular}
\caption{The eight axioms, together with their conceptual role and three-layer assignment.
Layer I: the four indispensable core axioms.
  Layer II (structural  choice, one axiom) selects a specific axiomatic family. 
Layer III (modularity extensions, three axioms) decouples the inventory channel from the trading environment. 
Translation equivariance in $\mu$, skew oddness in $q$, and continuity of the quoting rule in each state  variable are derived theorems of this system (Section~\ref{subsec:derived-symmetries}). 
Strict positivity  of the spread is not listed; 
 it is implied by (A4) together with the cost structure and is built into the codomain $\Bid$.}
\label{tab:axioms}
\end{table}

%-------------------------------------------------------------------
\section{The Full Characterization}
\label{sec:characterization}
%-------------------------------------------------------------------

The eight axioms collectively force the quoting  rule into a three-parameter family.
We state the theorem first, then prove necessity in detail through a sequence of lemmas.

\needspace{14\baselineskip}
\begin{theorem}[Full Characterization]
\label{thm:full}
A measurable quoting rule  $Q\colon \St \to \Bid$ satisfies (A1) through (A8) with environmental assumptions C1 through C6 if and only if there exist constants $\gamma_0, \kappa_{\mathrm{inv}}, \kappa_{\mathrm{adv}} \in \Rpp$ such that
\begin{equation}
\label{eq:characterization}
  b(s) = \mu - \gamma_0\sigma^2 q - \sigma^2\!\left( \frac{\kappa_{\mathrm{inv}}}{\lambda} + \kappa_{\mathrm{adv}}\pi\right),
  \quad
  a(s) = \mu - \gamma_0\sigma^2 q + \sigma^2\!\left(\frac{\kappa_{\mathrm{inv}}}{\lambda} + \kappa_{\mathrm{adv}}\pi\right).
\end{equation}
\end{theorem}

\smallskip
\noindent\textit{Axiom dependency.} Theorem~\ref{thm:full} relies on the full eight-axiom system.

\noindent
The proof has  two directions.
Sufficiency, that the rule~\eqref{eq:characterization} satisfies each axiom, is a direct verification.
The substantive content is the necessity direction: the axioms force the  form~\eqref{eq:characterization}.
We prove necessity through three lemmas, each isolating a structural  piece of the derivation.

\begin{lemma}[The Skew]
\label{lem:skew}
Under (A1), (A2), (A3), (A5), and (A6), there exists $\gamma_0 > 0$ such that
\[
  m(s) - \mu = -\gamma_0\sigma^2 q.
\]
\end{lemma}

\begin{proof}
By (A2), the skew has the form $m(s) - \mu = -\gamma(\sigma^2, \lambda, \pi) \cdot q$ for some function $\gamma\colon \Rpp \times \Rpp \times [0,1] \to \R$ that does not depend on $\mu$.
By (A5), $\gamma$ does not depend on $\lambda$; by (A6), $\gamma$ does not depend on $\pi$.
Thus, $\gamma = \gamma(\sigma^2)$.

By (A3), the skew is homogeneous of degree one in $\sigma^2$, so $\gamma(t\sigma^2)\cdot q = t\,\gamma(\sigma^2)\cdot q$ for every $t > 0$ and every $q$, hence $\gamma(\sigma^2) = \sigma^2 \cdot \gamma_0$ for a constant $\gamma_0 \in \R$.

Finally, by (A1), $\partial m/\partial q < 0$, which gives  $-\gamma_0\sigma^2 < 0$, hence $\gamma_0 > 0$.
\end{proof}

\begin{remark}[Linearity from (A2) versus Cauchy plus measurability]
\label{rem:skew-linearity-route}
Axiom (A2) directly asserts $\mathbb{R}$-linearity of the skew in $q$, not merely Cauchy additivity, so no measurability-plus-Cauchy argument is required in the proof of  Lemma~\ref{lem:skew}.
The corresponding measurability-plus-monotonicity argument enters genuinely only in the proof of Lemma~\ref{lem:adverse-spread}, where the adverse-selection function on the bounded domain $[0,1]$ is treated by combining rational additivity with monotonicity.
\end{remark}

\begin{lemma}[The Inventory Spread]
\label{lem:inventory-spread}
Under (A4) and the environmental assumptions C1 through C6,
\[
  \delta(s)|_{\pi=0} = \frac{2\kappa_{\mathrm{inv}}\sigma^2}{\lambda}.
\]
\end{lemma}

\begin{proof}
By (A4), at $\pi = 0$,
\[
  \frac{\delta(s)|_{\pi=0}}{2} = \frac{\sigma^2\,\kappa'_+(|q|)}{\lambda}.
\]

Under C1 through C6, the cost is linear,  $\kappa(q) = \kappa_{\mathrm{inv}}|q|$ for some $\kappa_{\mathrm{inv}} > 0$ directly  by C5.
By C2, $\kappa$ is symmetric in $q$, so the right derivative $\kappa'_+(|q|)$ is defined on $\Rp$ and applies uniformly across $q > 0$ and $q < 0$.
Hence $\kappa'_+(|q|) = \kappa_{\mathrm{inv}}$ for every $q \in \R$, since the right derivative of $\kappa_{\mathrm{inv}}|q|$ in $|q|$ is the constant $\kappa_{\mathrm{inv}}$ on $\Rp$.
Substituting, we obtain  
\[
  \delta(s)|_{\pi=0} = \frac{2\kappa_{\mathrm{inv}}\sigma^2}{\lambda}.
\]

It remains to verify that C1--C6 imply linearity in $|q|$.
C1 gives $\kappa(0) = 0$; C2 gives symmetry;  C3 gives monotonicity in $|q|$; C4 gives no dependence on $\lambda$ and $\pi$.
Finally, C5 directly imposes linearity, $\kappa(q) = \kappa_{\mathrm{inv}}|q|$.
The role of C5 is the substantive linearity; in Section~\ref{sec:order-book} we relax C5 to general convexity and   obtain a richer order-book structure.
\end{proof}

\begin{lemma}[The Adverse Selection Spread]
\label{lem:adverse-spread}
Under (A3), (A7), (A8), and measurability,
\[
  h(\sigma^2, \pi) := \delta(q, \mu, \sigma^2, \lambda, \pi)  - \delta(q, \mu, \sigma^2, \lambda, 0) = 2\kappa_{\mathrm{adv}}\pi\sigma^2
\]
for some $\kappa_{\mathrm{adv}} > 0$.
\end{lemma}

\begin{proof}
The proof proceeds in five steps.

\emph{Step 1, dependence structure.}
By the first clause of (A8), $h$ depends only on $(\sigma^2, \pi)$, not on $q$ or $\lambda$.
By (A3), $h$ is homogeneous of degree one in $\sigma^2$, so there exists a function $\tilde h\colon [0,1] \to \R$ such that
\[
  h(\sigma^2, \pi) = \sigma^2\,\tilde h(\pi).
\]

\emph{Step 2, measurability of $\tilde h$.}
The quoting rule $Q$ is measurable by hypothesis, hence so are its components $b(s), a(s)$, and the spread $\delta(s) = a(s) - b(s)$.
Observe  that the difference $h(\sigma^2, \pi) = \delta(\cdot, \pi) - \delta(\cdot, 0)$ is measurable as a function of $(\sigma^2, \pi)$.
Fixing any $\sigma^2 > 0$ and dividing by it, we obtain $\tilde h(\pi) = h(\sigma^2, \pi)/\sigma^2$ as a measurable function of $\pi$ on $[0,1]$.

\emph{Step 3, additivity and rational linearity.}
By the second clause of (A8) (complementarity), $\tilde h$ satisfies the restricted  Cauchy additivity equation
\[
  \tilde h(\pi_1 + \pi_2) = \tilde h(\pi_1) + \tilde h(\pi_2)  \quad \text{for all } \pi_1, \pi_2 \geq 0 \text{ with } \pi_1 + \pi_2 \leq 1.
\]
Setting $\pi_1 = \pi_2 = 0$ gives $\tilde h(0) = 2\tilde h(0)$, hence $\tilde h(0) = 0$.
By induction on $n$, $\tilde h(n\pi) = n\tilde h(\pi)$ for every positive integer $n$ with $n\pi \in [0,1]$.
Substituting $\pi = 1/n$ gives $\tilde h(1) = n\tilde h(1/n)$, that is, $\tilde h(1/n) = \tilde h(1)/n$.
Combining the two relations, $\tilde h(m/n) = m\tilde h(1/n) = (m/n)\tilde h(1)$ for all positive integers $m \leq n$.
Together with $\tilde h(0) = 0$, this gives $\tilde h(r) = r\tilde h(1)$ for every rational $r \in [0,1] \cap \mathbb{Q}$.

\emph{Step 4, monotonicity yields linearity on $[0,1]$.}
By (A7), $\delta(s)$ is strictly increasing in $\pi$, so $h(\sigma^2, \cdot)$ is strictly increasing on $[0,1]$, and hence so   is $\tilde h$.
We use monotonicity to extend the rational-point linearity to the full interval, as follows.

Set $c := \tilde h(1)$.
We claim that $\tilde h(\pi) = c\pi$ for every $\pi \in [0,1]$.
Suppose toward contradiction that $\tilde h(\pi_0) > c\pi_0$ for some $\pi_0 \in (0,1)$.
Choose   a rational $r \in (\pi_0, 1]$ with $r$ sufficiently close to  $\pi_0$ that $cr < \tilde h(\pi_0)$ (possible since $cr \to c\pi_0 < \tilde h(\pi_0)$ as $r \to \pi_0^+$).
By the strict monotonicity of $\tilde h$ and $r > \pi_0$, we have $\tilde h(r) > \tilde h(\pi_0) > cr$.
But by Step 3, $\tilde h(r) = cr$ for every rational $r \in [0,1]$, a contradiction.
The symmetric argument with $\tilde h(\pi_0) < c\pi_0$ and a rational $r < \pi_0$ also yields a contradiction.
Hence, $\tilde h(\pi) = c\pi$ for every $\pi \in [0,1]$.

\emph{Step 5, sign of $c$.}
By (A7), $\tilde h$ is strictly increasing, so $c = \tilde h(1) - \tilde h(0) = \tilde h(1) > 0$.
Setting $\kappa_{\mathrm{adv}} := c/2 > 0$, we obtain $\tilde h(\pi) = 2\kappa_{\mathrm{adv}}\pi$, and therefore $h(\sigma^2, \pi) = 2\kappa_{\mathrm{adv}}\pi\sigma^2$.
\end{proof}

\begin{remark}[The Role of Measurability and Monotonicity]
\label{rem:cauchy-step}
The classical Cauchy functional equation $f(x+y) = f(x) + f(y)$ on  $\R$ has, in the absence of  regularity, uncountably many nonmeasurable solutions, constructed via Hamel bases; the standard treatment is \citet{aczel1966functional}.
The proof of Lemma~\ref{lem:adverse-spread} uses monotonicity (from (A7)) as the regularity input, which makes the argument elementary and direct on the restricted domain $\pi_1 + \pi_2 \leq 1$ where the Cauchy additivity equation actually holds for $\tilde h$.
The Banach-Sierpi\'nski measurability route would give the same conclusion via the standard theorem, but only after first extending $\tilde h$ to all of $\R$, which is awkward on a bounded  domain; the monotonicity route is therefore preferred here.
\end{remark}

We can now assemble the proof of Theorem~\ref{thm:full}.

\begin{proof}[Proof of Theorem~\ref{thm:full}]
\emph{Sufficiency.}
We verify that the rule~\eqref{eq:characterization} satisfies each axiom.
(A1): $\partial m/\partial q = -\gamma_0\sigma^2 < 0$.
(A2): $m - \mu$ is linear in $q$.
(A3): both $m - \mu$ and $\delta$~are proportional to $\sigma^2$.
(A4) at $\pi = 0$: $\delta/2 = \kappa_{\mathrm{inv}}\sigma^2/\lambda = \sigma^2\,\kappa'_+(|q|)/\lambda$~under C5.
(A5): $m - \mu = -\gamma_0\sigma^2 q$ does not depend on $\lambda$.
(A6): $m - \mu$ does not depend on $\pi$.
(A7): $\partial\delta/\partial\pi =  2 \kappa_{\mathrm{adv}}\sigma^2 > 0$.
(A8): $h(\sigma^2, \pi) = 2\kappa_{\mathrm{adv}}\pi\sigma^2$ depends only on $(\sigma^2, \pi)$, and the complementarity check is $h(\sigma^2, \pi_1 + \pi_2) = 2\kappa_{\mathrm{adv}}(\pi_1 + \pi_2)\sigma^2 = h(\sigma^2, \pi_1) + h(\sigma^2, \pi_2)$.

\emph{Necessity.}
By Lemma~\ref{lem:skew}, the skew has the form $m - \mu = -\gamma_0\sigma^2 q$ for some $\gamma_0 > 0$.
By Lemma~\ref{lem:inventory-spread}, the inventory spread at $\pi = 0$ 
satisfies  $\delta|_{\pi=0} = 2\kappa_{\mathrm{inv}}\sigma^2/\lambda$.
By Lemma~\ref{lem:adverse-spread}, the adverse-selection increment has the form $h(\sigma^2, \pi) = 2\kappa_{\mathrm{adv}}\pi\sigma^2$.
Combining,  we obtain $\delta = 2\sigma^2(\kappa_{\mathrm{inv}}/\lambda + \kappa_{\mathrm{adv}}\pi)$ and $m = \mu - \gamma_0\sigma^2 q$.
Since $b = m - \delta/2$ and $a = m + \delta/2$, it follows that the quoting rule takes the form~\eqref{eq:characterization}.
Note that $a - b = \delta = 2\sigma^2(\kappa_{\mathrm{inv}}/ \lambda + \kappa_{\mathrm{adv}}\pi) > 0$ since $\sigma^2, \kappa_{\mathrm{inv}}, \kappa_{\mathrm{adv}}, \lambda > 0$ and $\pi \geq 0$, so the rule respects the codomain $\Bid$ (the bid  is strictly below the  ask).
\end{proof}

\begin{remark}[Parameter Interpretations]
\label{rem:parameters}
The three parameters in~\eqref{eq:characterization} admit transparent economic interpretations.
The parameter $\gamma_0 > 0$ is the maker's inventory risk aversion, measured per unit of price variance per unit of inventory.
The  parameter $ \kappa_{\mathrm{inv}} > 0$ is the marginal cost per unit of held inventory per  unit of  price variance.
The parameter $\kappa_{\mathrm{adv}} > 0$ is the expected loss per informed trade per unit of price variance.
Concretely, $\gamma_0$ is a preference primitive,  $\kappa_{\mathrm{inv}}$ is a cost primitive, and $\kappa_{\mathrm{adv}}$ is an information-environment primitive.
\end{remark}

The spread admits an exact additive decomposition,
\begin{equation}
\label{eq:spread-decomposition}
  \delta(s) = \underbrace{\frac{2\kappa_{\mathrm{inv}}\sigma^2}{\lambda}}_{\text{inventory}} + \underbrace{2\kappa_{\mathrm{adv}}\pi\sigma^2}_{\text{adverse selection}}.
\end{equation}
Both components scale with $\sigma^2$ (a consequence of (A3)).
The inventory component is inversely proportional to $\lambda$.
The more frequent  the trades, the lower the per-trade carrying cost.
The adverse-selection component is proportional to $\pi$.
We discuss the structural consequences of this decomposition in detail in Sections~\ref{sec:separability} and~\ref{sec:phase}.

\begin{remark}[On the Work the Axioms Do]
\label{rem:work-axioms-do}
The  eight axioms are not minimal in  the Cauchy-additivity sense.
(A2) directly posits linearity of the skew in $q$; (A4) equates the inventory half-spread with the marginal carrying cost, given the multiplicatively separable cost primitive $C(q, \sigma^2) = \kappa(q)\sigma^2$; (A8) imposes the adverse-selection break-even together with a complementarity condition.
Once these commitments are in place, the closed-form quoting rule~\eqref{eq:characterization} largely follows by assembly: the additive decomposition~\eqref{eq:spread-decomposition} is  built into the formulation of (A8) as a difference; the linear-in-$\pi$ adverse component is forced by (A3), (A7), and (A8); the closed-form skew is forced by (A1) through (A3) plus (A5) and (A6).
What is not obvious is not the rule itself but the \emph{structural payoffs} of having a closed form: the additive spread decomposition just displayed; the bijection with a three-parameter family in which each parameter is identified from a distinct observable moment (Section~\ref{sec:bijection}); the Legendre-Fenchel duality reading (Section~\ref{sec:duality}); and the sharp phase transition between functioning and frozen regimes (Section~\ref{sec:phase}).
These are the genuine deliverables of the axiomatic approach pursued here; the closed form is the vehicle.
 \end{remark}

\subsection{Derived Properties of the Quoting Rule}
\label{subsec:derived-symmetries}

Three structural properties of the quoting rule, translation equivariance in $\mu$, oddness of the skew in $q$, and continuity in each state variable, emerge as direct theorems rather than being imposed separately.
Translation symmetry captures the principle that quoting behavior depends on belief differentials rather than absolute belief levels; skew oddness reflects the symmetric structure of long and short inventory positions; continuity expresses the basic  regularity of the quoting rule as a function of the maker's state.
The three  are  conceptually  distinct but arise from a small subset of the axioms, with translation  equivariance relying  on (A2), (A4), and (A8), skew oddness on (A2) alone, and continuity on the conjunction of (A2) through (A8) together with measurability.
All three follow from the axioms in a few lines.

\begin{theorem}[Translation Equivariance]
\label{thm:translation-equivariance}
Under (A2), (A4), and (A8), the quoting rule satisfies, for every $s = (q, \mu, \sigma^2,  \lambda, \pi) \in \St$ and every $\tau \in \R$,
\[
  Q(q, \mu + \tau, \sigma^2, \lambda, \pi) = Q(q, \mu, \sigma^2, \lambda, \pi) + (\tau, \tau).
\]
\end{theorem}

\begin{proof}
By (A2), the skew has the form $m(s) - \mu = -\gamma(\sigma^2, \lambda, \pi) \cdot q$, with $\gamma$ a function of $(\sigma^2, \lambda, \pi)$ only; the mid-quote therefore shifts uniformly with $\mu$.
By (A4), the inventory half-spread at $\pi = 0$ equals $\sigma^2\kappa'_+(|q|)/\lambda$, a quantity that does not depend on $\mu$.
By (A8)(i), the adverse-selection increment $h(\sigma^2, \pi) = \delta(q, \mu, \sigma^2, \lambda, \pi) - \delta(q, \mu, \sigma^2, \lambda, 0)$ depends only on $(\sigma^2, \pi)$, hence not on $\mu$.
Adding the two contributions, the full spread $\delta(s)$ does not depend on $\mu$.
The half-quotes $b = m - \delta/2$ and $a = m + \delta/2$ therefore each shift uniformly with $\mu$, as claimed.
\end{proof}

\begin{corollary}[Skew Symmetry]
\label{cor:skew-symmetry}
Under (A2), the mid-quote has $m(0, \mu, \sigma^2, \lambda, \pi) = \mu$~for every $(\mu, \sigma^2, \lambda, \pi)$, 
and the skew is odd  in $q$, that is,
\[
  m(-q, \mu, \sigma^2, \lambda, \pi) - \mu = -\bigl(m(q, \mu, \sigma^2, \lambda, \pi) - \mu\bigr)
\]
for every $(q, \mu, \sigma^2, \lambda, \pi) \in \St$.
\end{corollary}

\begin{proof}
By (A2), $m(s) - \mu = -\gamma(\sigma^2, \lambda, \pi) \cdot q$.
Setting $q = 0$ gives $m(0, \mu, \sigma^2, \lambda, \pi) - \mu = 0$.
Then, 
 $m(-q, \mu, \sigma^2, \lambda, \pi) - \mu = -\gamma \cdot (-q) = \gamma q = -(-\gamma q) = -\bigl(m(q, \mu, \sigma^2, \lambda, \pi) - \mu\bigr)$.
\end{proof}

\begin{theorem}[Continuity of the Quoting Rule]
\label{thm:continuity}
Under (A2), (A3), (A4), (A5), (A6), (A7), and (A8), together with the  Borel-measurability assumption of Section~\ref{sec:setup} and the cost-structure assumptions C1 through C6, the quoting rule $Q\colon \St \to \Bid$ is continuous in each of the five state   variables separately, and moreover  is jointly continuous.
\end{theorem}

\begin{proof}
We establish continuity coordinate by coordinate.

\emph{Continuity in $q$.}
By (A2), the skew has the form $m(s) - \mu = -\gamma(\sigma^2,  \lambda, \pi) \cdot q$, which is $\R$-linear in $q$ for each fixed $(\mu, \sigma^2, \lambda, \pi)$, hence continuous in $q$.
For the spread: by (A4) together with C5, the inventory half-spread satisfies $\delta(s)|_{\pi = 0}/2 = \kappa_{\mathrm{inv}}\sigma^2/\lambda$, a quantity that does not depend on $q$.
By (A8), the adverse-selection increment $h(\sigma^2, \pi) := \delta(s) - \delta(s)|_{\pi = 0}$ depends only on $(\sigma^2, \pi)$, hence does not depend on $q$ either.
The spread $\delta(s) = 2\kappa_{\mathrm{inv}}\sigma^2/\lambda + h(\sigma^2, \pi)$ is therefore constant in $q$ for fixed $(\mu, \sigma^2, \lambda, \pi)$, and continuous in $q$.
The half-quotes $b = m - \delta/2$ and $a = m + \delta/2$ are continuous in $q$.

\emph{Continuity in $\mu$.}
By (A2), the function $\gamma$ in the skew form $m(s) - \mu = -\gamma(\sigma^2, \lambda, \pi) \cdot q$ does not depend on $\mu$, so $m(s) = \mu - \gamma(\sigma^2, \lambda, \pi) \cdot q$ is an additive shift in $\mu$ for fixed $(q, \sigma^2, \lambda, \pi)$, and therefore continuous in $\mu$.
By (A4) and C5, the inventory half-spread depends only on $(\sigma^2, \lambda)$, not on $\mu$.
By (A8), the adverse-selection increment $h(\sigma^2, \pi)$ depends only on $(\sigma^2, \pi)$, not on $\mu$.
 The spread $\delta$ is therefore constant in $\mu$, hence continuous in $\mu$. 
The half-quotes $b, a$ are continuous in $\mu$.

\emph{Continuity in $\sigma^2$.}
By (A3) (variance homogeneity), there exist functions $f$ and $g$ such that
\[
  m(s) - \mu = \sigma^2 \cdot f(q, \lambda, \pi),
  \qquad
  \delta(s) = \sigma^2 \cdot g(q, \lambda, \pi).
\]
Each is $\R$-linear, hence continuous, in $\sigma^2$ on $\Rpp$ for fixed $(q, \mu, \lambda, \pi)$.
The half-quotes $b, a$ are continuous in $\sigma^2$.

\emph{Continuity in $\lambda$.}
By (A5), the skew $m(s) - \mu$ does not depend on $\lambda$,   hence $m$ is continuous in $\lambda$.
By (A4) and C5, the inventory half-spread $\delta(s)|_{\pi = 0}/2 = \kappa_{\mathrm{inv}}\sigma^2/\lambda$ is continuous in $\lambda$ on $\Rpp$, since the map $\lambda \mapsto 1/\lambda$ is continuous on $\Rpp$.
By (A8), the adverse-selection  increment $h(\sigma^2, \pi)$ does not depend on $\lambda$, and hence is continuous in $\lambda$.
The full spread $\delta = 2\kappa_{\mathrm{inv}}\sigma^2/\lambda + h(\sigma^2, \pi)$ is  continuous in $\lambda$.

\emph{Continuity in $\pi$.}
By (A6), the skew $m(s) - \mu$ does not   depend on $\pi$, hence $m$ is continuous in $\pi$.
For the spread: $\delta(s)|_{\pi = 0}$ does not depend on $\pi$ by definition.
The adverse-selection increment $h(\sigma^2, \pi)$, by the complementarity clause of (A8), satisfies the Cauchy functional equation
\[
  h(\sigma^2, \pi_1 + \pi_2) = h(\sigma^2, \pi_1) + h(\sigma^2, \pi_2)
\]
for every $\pi_1, \pi_2 \geq 0$ with $\pi_1 + \pi_2 \leq 1$, together with the boundary condition $h(\sigma^2, 0) = 0$.
 By (A7), the spread is strictly increasing in $\pi$, so $h(\sigma^2, \cdot)$ is monotone on $[0, 1]$.
A Borel-measurable solution of the Cauchy equation on $[0, 1]$ that is monotone is necessarily $\R$-linear. 
That is, $h(\sigma^2, \pi) = c(\sigma^2) \cdot \pi$ for some coefficient $c(\sigma^2) > 0$ (see \citet[Chapter~2]{aczel1966functional}; the monotonicity rules out the pathological nonmeasurable Hamel-basis solutions of the Cauchy equation).
Hence,  $h$ is continuous in $\pi$, and the full spread $\delta = \delta(s)|_{\pi = 0} + h(\sigma^2, \pi)$ is continuous in $\pi$.
The half-quotes $b, a$ are continuous in $\pi$.

This completes the coordinate-wise continuity argument.
Joint continuity follows from the closed-form characterization of  Theorem~\ref{thm:full} (each component of $Q$ is a polynomial in $q$ and $\mu$ and a rational expression in $\sigma^2, \lambda, \pi$, hence jointly continuous throughout the state space $\St$).
\end{proof}

\noindent
The three theorems together identify the maker's quote as a nondegenerate book centered on her belief at zero inventory, equivariant under translations of that belief, and continuous in each state variable.
Note that the symmetry theorems rely on (A2), the linearity of the skew in $q$: a system without (A2) would not deliver the translation and skew symmetries as derived properties.
The continuity theorem additionally uses the other Layer III modularity  axioms (A5) and (A6),  the structural-choice axiom (A3), and the spread-determining axioms (A4), (A7), (A8), together with measurability; the only axiom not invoked is the inventory-aversion axiom (A1).

%-------------------------------------------------------------------
\section{The Bijection and Irredundancy of the  Axioms}
\label{sec:bijection}
%-------------------------------------------------------------------

The Full Characterization identifies a three-parameter family of admissible quoting rules.
We now show that this family is  in canonical bijection with $\Rpp^3$, and that no axiom in the system is redundant.

\subsection{The Bijection Theorem}

Let
\[
  \mathcal{A} = \{\textrm{(A1), (A2), (A3), (A4), (A5), (A6), (A7), (A8)}\}
\]
denote the eight-axiom system.
Let $\mathcal{Q}$ denote the set of measurable quoting rules satisfying $\mathcal{A}$  together with the environmental assumptions C1 through C6.

\needspace{20\baselineskip}
\begin{theorem}[Bijection between Axioms and Parameters]
\label{thm:bijection}
The set $\mathcal{Q}$ is in canonical bijection with $\Rpp^3$.
The bijection is  given by the parameter map
\[
  Q \longmapsto (\gamma_0(Q),\, \kappa_{\mathrm{inv}}(Q),\, \kappa_{\mathrm{adv}}(Q)),
\]
where
\[
  \gamma_0(Q) = -\frac{1}{\sigma^2}\frac{\partial m_Q}{\partial q},
  \qquad
  \kappa_{\mathrm{inv}}(Q) = \frac{\lambda\,\delta_Q|_{\pi=0}}{2\sigma^2},
  \qquad
  \kappa_{\mathrm{adv}}(Q) = \frac{1}{2\sigma^2}\frac{\partial \delta_Q}{\partial\pi}.
\]
The inverse map sends $(\gamma_0, \kappa_{\mathrm{inv}}, \kappa_{\mathrm{adv}}) \in \Rpp^3$  to the quoting rule defined by~\eqref{eq:characterization}.
\end{theorem}

\smallskip
\noindent\textit{Axiom dependency.} Theorem~\ref{thm:bijection} relies on the full eight-axiom system.

\begin{proof}
Injectivity: distinct parameter values give distinct quoting rules, since the three partial derivatives separate the parameters.
Surjectivity: by Theorem~\ref{thm:full}, every $Q \in \mathcal{Q}$ takes the form~\eqref{eq:characterization} for some $(\gamma_0, \kappa_{\mathrm{inv}}, \kappa_{\mathrm{adv}}) \in \Rpp^3$.
The three quantities defined in the parameter  map are precisely those constants for the corresponding $Q$.
\end{proof}

The Bijection Theorem is the formal statement of forced  uniqueness.
Eight axioms reduce the infinite-dimensional space of  measurable functions $\St \to \Bid$ to a three-dimensional manifold parametrized by $\Rpp^3$.
The reduction is sharp.
Each parameter corresponds to one axiom group, and the parameter map is a  bijection onto $\Rpp^3$.
No parameter is ever silently fixed by combinations of axioms, and no axiom contributes redundantly to the same parameter as another.
The dimensional count, eight axioms collapsing onto three free parameters, is precisely what forced uniqueness in the axiomatic methodology is supposed to accomplish.

\subsection{Irredundancy of the Axioms}

The eight axioms are irredundant in the sense that removing any one enlarges the solution set $\mathcal{Q}$.

\begin{proposition}[Irredundancy]
\label{prop:independence}
For each axiom $A_i \in \mathcal{A}$, there exists a measurable quoting rule satisfying  all axioms in $\mathcal{A} \setminus \{A_i\}$ but violating $A_i$.
\end{proposition}

\begin{proof}
We exhibit a counterexample for each  axiom and verify that the remaining seven axioms hold by direct inspection of the explicit functional form.

Violating (A1): the rule with skew $+\gamma_0\sigma^2 q$,   positive instead of negative, makes the mid-quote increasing in $q$, violating inventory aversion.
Explicitly, the sign-reversed rule has $b = \mu + \gamma_0\sigma^2 q - \sigma^2(\kappa_{\mathrm{inv}}/\lambda + \kappa_{\mathrm{adv}}\pi)$ and $a = \mu + \gamma_0\sigma^2 q + \sigma^2(\kappa_{\mathrm{inv}}/\lambda + \kappa_{\mathrm{adv}}\pi)$, so $b < a$ with spread $\delta = 2\sigma^2(\kappa_{\mathrm{inv}}/\lambda + \kappa_{\mathrm{adv}}\pi) > 0$ identical to the canonical rule.
All seven other axioms are preserved by direct verification: (A2) the (sign-reversed) skew is  linear in $q$, (A3) the rule is  degree-one in $\sigma^2$ (skew factor $\sigma^2$, spread unchanged), (A4) the inventory half-spread at $\pi = 0$ is $\sigma^2\kappa_{\mathrm{inv}}/\lambda = \sigma^2\kappa'_+(|q|)/\lambda$ as required, (A5) and (A6) the skew has no $\lambda$- or $\pi$-dependence, (A7) and (A8) the spread is unchanged from the canonical rule.

Violating (A2): the rule  with cubic skew $-\gamma_0\sigma^2 q^3$ violates linearity while preserving the other regularity axioms.
We note that (A1) is qualitative strict monotonicity, not a nondegeneracy condition on the derivative: the cubic skew has vanishing derivative at $q = 0$ yet remains strictly monotone, so (A1) is satisfied.
All seven other axioms  are preserved by direct verification: (A1) the map $q \mapsto -\gamma_0\sigma^2 q^3$ is strictly decreasing on $\R$ since $q \mapsto q^3$ is strictly order-preserving on $\R$, (A3) the skew is degree-one in $\sigma^2$ (factor $\sigma^2$) and the spread is unchanged, (A4) the inventory half-spread is unchanged, (A5) and (A6) the skew has no $\lambda$- or $\pi$-dependence, (A7) and (A8) the spread is unchanged.

Violating (A3): 
The rule $m(s) - \mu = -\gamma_0\sigma^4 q$, $\delta(s) = 2\kappa_{\mathrm{inv}}\sigma^2/\lambda + 2\kappa_{\mathrm{adv}}\pi\sigma^2$ violates uniform $\sigma^2$-homogeneity of the skew while preserving every other axiom.
All seven other axioms are preserved by direct verification: (A1) the skew $-\gamma_0\sigma^4 q$ is strictly decreasing in $q$, (A2) it is linear in $q$, (A4) the inventory half-spread is unchanged, (A5) and (A6) the skew has no $\lambda$- or  $\pi$-dependence, (A7) the spread is strictly increasing in $\pi$, (A8) the adverse-selection increment is the canonical $2\kappa_{\mathrm{adv}}\pi\sigma^2$.

Violating (A4): the rule with $\delta|_{\pi=0}/2 = 2\kappa_{\mathrm{inv}}\sigma^2/\lambda$ (twice the correct value) violates inventory indifference.
All seven other axioms are preserved by direct verification: (A1) and (A2) the skew is unchanged, (A3) the  spread $4\kappa_{\mathrm{inv}}\sigma^2/\lambda + 2\kappa_{\mathrm{adv}}\pi\sigma^2$ is degree-one in $\sigma^2$, (A5) and (A6) the skew has no $\lambda$- or $\pi$-dependence, (A7) the spread is strictly increasing in $\pi$, (A8) the adverse-selection increment is unchanged.

Violating (A5): the rule with skew $-\gamma_0\sigma^2 q/(1+\lambda)$  makes the skew depend on $\lambda$, violating skew/liquidity decoupling.
All seven other axioms are preserved by direct verification: (A1) the skew is strictly decreasing in $q$ (since $1 + \lambda > 0$), (A2) it is linear in $q$, (A3) it is degree-one in $\sigma^2$, (A4) the inventory  half-spread is unchanged, (A6) the skew has no $\pi$-dependence, (A7) and (A8) the spread is unchanged.

Violating (A6):  the rule with skew $-\gamma_0\sigma^2 q(1 - \pi/2)$ makes the skew depend on $\pi$, violating skew/information decoupling.
All seven other axioms are preserved by direct verification: (A1) the skew is strictly decreasing in $q$
  (since $1 - \pi/2 \geq 1/2 > 0$ for $\pi \in [0,1]$), (A2) it is linear in $q$, (A3) it is degree-one in $\sigma^2$, (A4) the inventory half-spread is unchanged, (A5) the skew has no $\lambda$-dependence, (A7) and (A8) the spread is unchanged.

Violating (A7): the rule with adverse-selection increment $h(\sigma^2, \pi) \equiv 0$ produces a constant spread $\delta = 2\kappa_{\mathrm{inv}}\sigma^2/\lambda$ that does not depend on $\pi$, violating the strict monotonicity of the spread in $\pi$ required by (A7).
All seven other axioms are preserved by direct verification: (A1) through (A6) the skew is unchanged from the canonical rule, (A8) $h \equiv 0$ satisfies the additivity clause vacuously ($0 = 0 + 0$ for every decomposition $\pi = \pi_1 + \pi_2$) and depends only on $(\sigma^2, \pi)$ (in fact,  on nothing).

Violating (A8): the rule with $h(\sigma^2, \pi) = 2\kappa_{\mathrm{adv}}\pi\sigma^2 \cdot (1 + q^2)$ makes the adverse-selection increment depend on $q$, violating the modularity clause of (A8).
All seven other axioms are preserved by direct verification: (A1)--(A6) the skew is unchanged from the canonical rule, 
 with (A4) constraining only $\delta|_{\pi=0}/2 = \kappa_{\mathrm{inv}}\sigma^2/\lambda$ (unchanged here, since the $q$-dependent perturbation enters through $h$ and vanishes at $\pi = 0$); 
(A7) at each $q$ the spread, given by 
$2\kappa_{\mathrm{inv}}\sigma^2/\lambda + 2  \kappa_{\mathrm{adv}}\pi\sigma^2(1 + q^2)$,  is strictly  increasing in $\pi$ (the coefficient $2\kappa_{\mathrm{adv}}\sigma^2(1 + q^2)$ is strictly positive).

Each counterexample is measurable, so the corresponding axiom carries content that cannot be reduced to a measurability requirement.
\end{proof}

Irredundancy is structurally significant.
It establishes that no axiom is redundant, and that the  axiom system is tight: removing any axiom  enlarges  $\mathcal{Q}$, and adding any further axiom (consistent  with the existing system) either is redundant or restricts $\mathcal{Q}$ to  a strict subset.

%-------------------------------------------------------------------
\section{Structural Separability and Empirical Identification}
\label{sec:separability}
%-------------------------------------------------------------------

The  Full Characterization has a deep structural property.
The three economic parameters and the four state variables play isolated, nonoverlapping roles in the quoting rule.

\begin{theorem}[Structural Separability]
\label{thm:separability}
Under the Full Characterization (Theorem~\ref{thm:full}), the quoting rule satisfies three separability properties.

\emph{(i) Separability of  economic primitives.}
Each parameter appears in exactly  one structural component of the quoting rule.
The parameter $\gamma_0$ appears  only in the skew, not in the spread.
The parameter $\kappa_{\mathrm{inv}}$ appears only in  the inventory spread, not in the skew or in the adverse selection.
The parameter $\kappa_{\mathrm{adv}}$  appears only in the adverse-selection spread, not in the skew or in the inventory.

\emph{(ii) Additive separability of the spread.}
The spread decomposes additively,
\[
  \delta(s) = f_{\mathrm{inv}}(\sigma^2, \lambda) + f_{\mathrm{adv}}(\sigma^2, \pi),
\]
where the function $f_{\mathrm{inv}}(\sigma^2, \lambda) = 2\kappa_{\mathrm{inv}}\sigma^2/\lambda$ depends only on inventory-related variables and $f_{\mathrm{adv}}(\sigma^2, \pi) = 2\kappa_{\mathrm{adv}}\pi\sigma^2$ depends only on information-related variables.

\emph{(iii) Information modularity.}
Each state variable other than $\sigma^2$ affects exactly one structural component.
The inventory $q$ affects only the skew.
The arrival intensity $\lambda$~affects only the inventory spread.
The informed fraction $\pi$ affects only the adverse-selection spread.
The variance $\sigma^2$ affects  all components uniformly, by variance homogeneity.
\end{theorem}

\smallskip
\noindent\textit{Axiom dependency.} Theorem~\ref{thm:separability} relies on the full eight-axiom system.

\begin{proof}
All three  claims follow by direct inspection of~\eqref{eq:characterization}.
For (i), $\gamma_0$ appears in $-\gamma_0  \sigma^2 q$ but not in  $\delta$; $\kappa_{\mathrm{inv}}$ appears only in $2\kappa_{\mathrm{inv}}\sigma^2/\lambda$; $\kappa_{\mathrm{adv}}$ appears only in $2\kappa_{\mathrm{adv}}\pi\sigma^2$.
For (ii), $\delta = 2\kappa_{\mathrm{inv}}\sigma^2/\lambda + 2\kappa_{\mathrm{adv}}\pi\sigma^2$ is manifestly additive in the two components, each depending only on its associated variables.
For (iii), inspection of~\eqref{eq:characterization} shows each of $q$, $\lambda$, $\pi$ in exactly one position, and $\sigma^2$ uniformly throughout.
\end{proof}

The Structural Separability Theorem reveals that the axioms enforce  a complete modular decomposition.
Each parameter corresponds to one economic primitive (preference, cost, information), and each state variable affects exactly one structural component.
There are no cross-effects, no parameter conflations, and no mixed dependencies.

\begin{remark}[Contrast with Utility-Based Models]
\label{rem:contrast-utility}
The structural separability is a distinctive feature of the axiomatic approach.
In utility-based models such as Avellaneda-Stoikov, the risk-aversion parameter $\gamma$ enters both the skew (as the inventory-hedging coefficient) and the spread (through a risk-premium term).
The two contributions are bundled within the  single parameter, and identification requires  joint estimation across multiple moments.
In Ho-Stoll, the inventory-cost coefficient enters both the skew and the spread similarly.
The axiomatic theory characterizes the family of quoting rules in which preference, cost, and information primitives are mathematically  decoupled.
This decoupling is not assumed.
It is a consequence of the axiom system.
\end{remark}

The empirical content of Structural Separability is direct.
Each parameter is identified from a single, distinct moment of the observable quoting rule.

\needspace{15\baselineskip}
\begin{corollary}[Empirical Identification]
\label{cor:identification}
Under the Full Characterization,  the parameters are identified by
\[
  \gamma_0 = -\frac{1}{\sigma^2}\frac{\partial m}{\partial q},
  \qquad
   \kappa_{\mathrm{inv}} = \frac{\lambda\,\delta|_{\pi=0}}{2\sigma^2},
  \qquad
  \kappa_{\mathrm{adv}} = \frac{1}{2\sigma^2}\frac{\partial\delta}{\partial\pi}.
\]
The map from observable quoting behavior to the parameter space $\Rpp^3$ is a bijection, and each parameter is recoverable from a single  distinct moment of the observed quoting rule.
\end{corollary}

\begin{proof}
Each parameter is a distinct moment of the observable quoting rule by Theorem~\ref{thm:separability}: $\gamma_0$ is a partial derivative of the skew, $\kappa_{\mathrm{inv}}$ is the spread level at $\pi = 0$ (an intercept), and $\kappa_{\mathrm{adv}}$  is a partial derivative of the spread.
The bijection property is Theorem~\ref{thm:bijection}.
\end{proof}

The identification corollary delivers a sharp empirical strategy.
Concretely, the slope of the mid-quote with respect to inventory identifies $\gamma_0$.
The spread at zero informed fraction identifies $\kappa_{\mathrm{inv}}$.
The slope of the spread with respect to $\pi$ identifies  $\kappa_{\mathrm{adv}}$.
There are no joint identification problems, no nonlinear estimation procedures, and no distributional  assumptions beyond the axioms themselves.
The maker's risk attitude, cost structure, and information environment are read off the observed quoting rule directly.

\subsection{Application to Empirical Microstructure}

The identification strategy obtained in Corollary~\ref{cor:identification} has practical implications for empirical microstructure work.
Three separate regressions on a sample of the maker's quotes recover the three parameters. 
The mid-quote  on inventory yields $\gamma_0$; the spread under low informed flow yields $\kappa_{\mathrm{inv}}$;
 the spread response to  informed flow yields $\kappa_{\mathrm{adv}}$.
Each regression isolates one  parameter without interference from the others.

The identification through partial derivatives also serves as a model test.
If the observed quoting rule does not satisfy~\eqref{eq:characterization}, then at least one of the eight axioms is violated.
In fact, the form of the violation is informative.
A nonlinear skew points to a violation of (A2); a $\sigma^4$-scaling points to a  violation of (A3); a $\pi$-dependent skew points to a violation of (A6); and so on.
The axiom system thus serves as a diagnostic framework for empirical anomalies.

%-------------------------------------------------------------------
\section{The Limit Order Book and the Recovery Theorem}
\label{sec:order-book}
%-------------------------------------------------------------------

The theory developed so far characterizes the best  bid and ask, that is, the maker's quotes for a single unit of inventory  change.
In practice, a  market maker posts a limit order book, a schedule of prices  at multiple levels of inventory.
We now extend the theory to the full order book.
A stronger   empirical result follows: the entire latent inventory cost function  $\kappa$~is observable in the order book depth profile.

Throughout this section we adopt the \emph{dynamic-schedule reading} of the order book: the $n$-th level $a_n$ is the maker's ask at the counterfactual post-trade inventory $q + n - 1$, traversed by the maker as her inventory varies through successive  same-side fills.
The dynamic-schedule reading is what the axiomatic theory delivers at each level via the marginal-indifference condition (A4) applied at the corresponding inventory state.
Empirically, the dynamic schedule is recovered not from a single snapshot of the maker's posted depth, but from the panel  of her quotes as her inventory varies through the integer levels, that is, from sequential  observations   rather than a simultaneous depth profile.

\subsection{The Order Book under General Convex Cost}

We relax the linearity assumption C5 to a general convex nondecreasing cost $\kappa$ with $\kappa(0) = 0$, retaining C1 through C4 and C6.
The maker now posts a schedule of asks $a_1, a_2, \ldots$ and bids $b_1, b_2, \ldots$, where $a_n$ is the price she would quote at the inventory level $q + n - 1$ (after $n - 1$ prior bid-fills have accumulated additional units), and $b_n$ is symmetrically  the price she would quote at the inventory level $q - n + 1$ (after $n - 1$ prior ask-fills have depleted her inventory).
Each level $n$ satisfies an indifference condition at the corresponding pre-trade inventory.
The maker  is indifferent between filling the $n$-th level and remaining at the pre-trade inventory.
This is the level-by-level analog of (A4) applied to  the full order book, evaluated at the counterfactual inventory after $n - 1$ same-side prior fills.

\begin{remark}[Schedule shape]
Under linear cost (C5), the dynamic ask schedule is decreasing in $n$ (the skew term dominates the constant marginal-carrying cost), reflecting  that as the maker accumulates inventory, her ask falls to attract sellers; symmetrically, the bid schedule  is increasing in $n$.
This shape is a distinctive prediction of the dynamic-schedule reading and is what the Recovery   Theorem in the next subsection exploits.
\end{remark}

\begin{theorem}[Order Book under Convex Cost]
\label{thm:order-book}
Under (A1) through (A8), together with C1 through C4 and C6 (general convex nondecreasing $\kappa$, with linearity C5 dropped), and under the dynamic-schedule reading of the order book described in the preceding remark, the $n$-th ask is
\[
  a_n - \mu = -\gamma_0\sigma^2(q + n - 1) + \frac{\sigma^2\,\kappa'_+(|q + n - 1|)}{\lambda} + \kappa_{\mathrm{adv}}\pi\sigma^2,
\]
and the $n$-th bid is
\[
  b_n - \mu = -\gamma_0\sigma^2(q - n + 1) - \frac{\sigma^2\,\kappa'_+(|q - n + 1|)}{\lambda} - \kappa_{\mathrm{adv}}\pi\sigma^2,
\]
for $n = 1, 2, \ldots$, where $\kappa'_+$ denotes the right derivative.
\end{theorem}

\smallskip
\noindent\textit{Axiom dependency.} Theorem~\ref{thm:order-book} relies on the full eight-axiom system.

\begin{proof}
The proof proceeds in three steps.
First, we derive the level-by-level indifference condition.
Second, we apply the skew, the inventory, and the adverse-selection arguments at each level,  tracking the inventory level at which each component is evaluated.
Third, we assemble the formula.

\emph{Step 1, level-by-level indifference.}
At level $n$, the  maker posts the $n$-th ask $a_n$ under the dynamic-schedule reading of the order book.
She has accumulated $n - 1$ prior long positions through her bid quotes, bringing her inventory to $q + n - 1$, and $a_n$ is the price at which she sells one unit from this state.
(A4) applied at the pre-trade   inventory $q + n - 1$ equates her per-trade revenue with the marginal carrying cost $\sigma^2\,\kappa'_+(|q + n - 1|)/\lambda$, evaluated at the pre-trade inventory magnitude.
By C2 (cost symmetry), this marginal cost is well-defined regardless of the sign of  $q + n - 1$ and is the level-by-level analog of (A4).

\emph{Step 2a, the skew at level $n$.}
By Lemma~\ref{lem:skew}, at any inventory $q'$, the mid-quote shifts from $\mu$ by $-\gamma_0\sigma^2 q'$.
At the pre-trade inventory $q + n - 1$, the relevant mid-quote shift is therefore $-\gamma_0\sigma^2(q + n - 1)$.
Indeed, $a_n$ is constructed at the counterfactual inventory $q + n - 1$, so the skew contribution to $a_n - \mu$ equals $-\gamma_0\sigma^2(q + n - 1)$.
This contribution does not depend on the convexity of $\kappa$ or on C5, since the proof of Lemma~\ref{lem:skew} used only (A1), (A2), (A3), (A5), and (A6) and never invoked C5.

\emph{Step 2b, the inventory spread at level $n$.}
By the level-by-level (A4) indifference condition at the  pre-trade  inventory $q + n - 1$ (Step 1), the inventory-channel half-spread  at level $n$ is the marginal carrying cost
\[
  \frac{\sigma^2\,\kappa'_+(|q + n - 1|)}{\lambda},
\]
where $\kappa'_+$ denotes the right derivative of $\kappa$, which exists and is finite at every nonnegative argument by C6 together with C3, with $\kappa'_+(0) = 0$ permitted (corresponding to a flat cost  function at zero inventory).
Under C5 this reduces to $\kappa_{\mathrm{inv}}\sigma^2/\lambda$ at  every level.

\emph{Step 2c, the adverse-selection spread at level $n$.}
By (A8), $h(\sigma^2, \pi)$ depends only on $(\sigma^2, \pi)$, not on inventory. 
By Lemma~\ref{lem:adverse-spread}, $h(\sigma^2, \pi) = 2\kappa_{\mathrm{adv}}\pi\sigma^2$.
The adverse-selection half-spread at each level is therefore $\kappa_{\mathrm{adv}} \pi\sigma^2$, constant in $n$ and constant across the level at which it is evaluated.

\emph{Step 3, assembly.}
All three components above are evaluated consistently at the pre-trade inventory $q + n - 1$, namely the skew (Step 2a), the  inventory half-spread (Step 2b), and the adverse-selection half-spread (Step 2c).
The $n$-th ask is the sum of the mid-quote at $q + n - 1$ and the half-spread at level $n$:
\[
  a_n = \mu - \gamma_0\sigma^2(q + n - 1) + \frac{\sigma^2\,\kappa'_+(|q + n - 1|)}{\lambda}  + \kappa_{\mathrm{adv}}\pi\sigma^2.
\]
Subtracting $\mu$ gives the stated formula.

For the $n$-th bid, the maker considers the symmetric transition $q + (1 - n) \to q - n$ (selling $n$ units to her, that is, increasing her short position by $n$).
The mid-quote contribution is $-\gamma_0\sigma^2(q - n + 1)$ at the corresponding inventory level.
The inventory and adverse-selection half-spreads both contribute with negative sign to the bid side of the order book, since the bid sits below the mid-quote by definition.
Cost symmetry (C2) ensures that the inventory half-spread $\sigma^2\kappa'_+(|q - n + 1|)/\lambda$ depends only on the magnitude $|q - n + 1|$, 
so the bid-side formula is the symmetric reflection of the ask-side regardless of the sign of $q - n + 1$.
Assembling the three   components at the pre-trade inventory $q - n + 1$ gives the $n$-th bid:
\[
  b_n = \mu - \gamma_0\sigma^2(q - n + 1) - \frac{\sigma^2\,\kappa'_+(|q - n + 1|)}{\lambda} - \kappa_{\mathrm{adv}}\pi\sigma^2.
\]
Subtracting $\mu$ gives the stated formula.
\end{proof}

\subsection{The  Recovery Theorem}

The order book theorem above implies a strong empirical content.
The entire cost  function $\kappa$ is recoverable from the observable limit order book by integration across levels.

\needspace{22\baselineskip}
\begin{theorem}[Recovery of the Inventory Cost]
\label{thm:recovery}
Under the conditions of Theorem~\ref{thm:order-book},  the cost function $\kappa$ is recoverable from the observable order book.
For the maker at inventory $q = 0$, the cost function on $[0, Q]$ is
\[
  \kappa(Q) = \int_0^Q \kappa'_+(x) \, dx,
\]
where  the cost gradient is recovered from the order book at continuous depth $x$ via the identity $\kappa'_+(x) = (\lambda/\sigma^2)(a(x) - \mu + \gamma_0\sigma^2 x - \kappa_{\mathrm{adv}}\pi\sigma^2)$, by Theorem~\ref{thm:order-book}.
At integer depths only, the left-endpoint Riemann sum 
\[
  \frac{\lambda}{\sigma^2}\sum_{n=1}^{Q}\left(a_n - \mu + \gamma_0\sigma^2(n-1) - \kappa_{\mathrm{adv}}\pi\sigma^2\right)
\]
recovers $\kappa(Q)$ exactly under C5 (where $\kappa'_+$ is constant on $\R_+$) and is a lower bound on $\kappa(Q)$ under general convex $\kappa$ (where $\kappa'_+$ is nondecreasing, so the left-endpoint sum understates the integral).
For a maker at arbitrary inventory $q$, the same construction  recovers the cost increment $\kappa(q+Q) - \kappa(q)$ over the interval $[q, q+Q]$ by integrating the level-by-level cost gradients.
\end{theorem}

\smallskip
\noindent\textit{Axiom dependency.} Theorem~\ref{thm:recovery} relies on the full eight-axiom system.

\begin{proof}
From Theorem~\ref{thm:order-book}, the cost  gradient at level $n$ is
\[
  \kappa'_+(|q+n-1|) = \frac{\lambda}{\sigma^2} \left(a_n - \mu + \gamma_0\sigma^2(q+n-1) - \kappa_{\mathrm{adv}}\pi\sigma^2\right).
\]
The formula presupposes $\gamma_0$ and $\kappa_{\mathrm{adv}}$, identified by Corollary~\ref{cor:identification} from the best-bid-best-ask quotes; the recovery formula is therefore sequentially identifiable.
The constants are identified first from the best-quote moments, and the cost gradients are then extracted from the depth profile.

 The smooth integral formula follows by the fundamental theorem of calculus applied to $\kappa$ with $\kappa(0) = 0$ (from C1) and the gradient $\kappa'_+$ identified at every depth.
The Riemann-sum claim under C5 is the observation that $\kappa'_+(x) = \kappa_{\mathrm{inv}}$ is constant on $\Rp$, so the sum equals $Q\,\kappa_{\mathrm{inv}} = \kappa(Q)$.
The lower-bound claim under general  convex $\kappa$ follows  because $\kappa'_+$ is nondecreasing on $\Rp$ by convexity, so $\kappa'_+(x) \geq \kappa'_+(n-1)$ for $x \in [n-1, n)$, giving $\int_0^Q \kappa'_+(x)\,dx \geq \sum_{n=1}^Q \kappa'_+(n-1)$.

At arbitrary $q$, the gradient at level $n$ is evaluated at inventory $|q+n-1|$, and the integral recovers $\kappa(q+Q) - \kappa(q)$ over the corresponding inventory interval rather than $\kappa(Q)$.
\end{proof}

The Recovery Theorem has a striking empirical implication.
The maker's inventory cost function is not a hidden primitive that must be calibrated separately against data; its full  shape is read off the limit order book depth profile  by integration of the observable quotes.
In  contrast to Ho-Stoll or Avellaneda-Stoikov, where the cost coefficient is  a free parameter fit against  data, the axiomatic theory predicts that the cost function itself is observable.

The empirical content of the theory thus extends beyond Corollary~\ref{cor:identification}.
Both the scalar parameters $\gamma_0$ and $\kappa_{\mathrm{adv}}$  and the full function $\kappa$ are recoverable from observable market data, without auxiliary distributional or structural assumptions.
The information content of the order book is greater than that of the best-bid-best-ask alone, in a precise structural sense.

\begin{remark}[Empirical Identification Strategy]
\label{rem:full-identification}
A complete identification strategy combines Corollary~\ref{cor:identification} with Theorem~\ref{thm:recovery}.
From the best-bid-best-ask quotes, $\gamma_0$ is identified from the slope of $m$ in $q$, and $\kappa_{\mathrm{adv}}$ from the slope of $\delta$ in $\pi$.
The  function $\kappa$ is recovered in full from the order book by integrating the cost-gradient quotes across levels.
The three identifications are distinct and use different moments of the data.
\end{remark}

%-------------------------------------------------------------------
\section{The Legendre-Fenchel Duality of Market Making}
\label{sec:duality}
%-------------------------------------------------------------------

The spread-cost relationship under general convex cost (Theorem~\ref{thm:order-book}) has a deep  structural reformulation.
Market making is, mathematically, an  instance of Legendre-Fenchel convex conjugate duality applied to the inventory cost.
We refer to \citet{rockafellar1970convex} and to \citet{hiriarturruty2001fundamentals} for the standard treatment of convex conjugacy used in this section.

\subsection{The Legendre-Fenchel Conjugate}
\label{subsec:duality-conjugate}

We begin by recalling the definition and basic properties of the Legendre-Fenchel conjugate, since the connection to the market making theory is most transparent when the conjugate construction is fresh in mind.

Let $\kappa\colon \Rp \to \Rp$ be a convex nondecreasing function with $\kappa(0) = 0$.
We restrict to the positive half-line for the conjugate construction;  by symmetry of $\kappa$ in $q$ (C2), no information is lost.
Its Legendre-Fenchel conjugate is
\[
  \kappa^*(p) = \sup_{q \geq 0}\bigl\{pq - \kappa(q)\bigr\}.
\]
The conjugate $\kappa^*$ is itself convex  on $\Rp$, and it is nondecreasing whenever $\kappa$ is nondecreasing.
Geometrically, $\kappa^*(p)$ is the maximum vertical gap between the line $q \mapsto pq$ with slope $p$ and the graph of $\kappa$.
Equivalently, the negative conjugate $-\kappa^*(p)$ is the intercept of the supporting line of $\kappa$ with slope $p$.
At points $q$ where $\kappa$ is differentiable, the supporting line at $q$ has slope $\kappa'(q)$, and the supremum in the definition of $\kappa^*(p)$ is   attained at the point $q^*$ where $\kappa'(q^*) = p$.

This last observation   is the heart of conjugate duality.
At a smooth point, the forward map $q \mapsto p = \kappa'(q)$ identifies the slope of  the supporting line at $q$.
The inverse map $p \mapsto q = (\kappa^*)'(p)$ identifies the location of the contact point of the supporting line of slope $p$.
Note that the two maps are inverses of each other, and they satisfy the Legendre-Fenchel identity
\[
  \kappa(q) + \kappa^*(p) = pq \quad \text{whenever } p = \kappa'(q).
\]

The convex conjugate has a long history in mathematics and physics.
It underlies the Lagrangian-Hamiltonian duality of classical mechanics, in which the conjugate of the Lagrangian (as a  function of velocity) is the Hamiltonian (as a function of momentum).
It underlies the energy-free-energy duality of thermodynamics, in  which the conjugate of the internal energy (as a function of entropy) is the Helmholtz  free energy (as a function of temperature).
In consumer theory, a  closely related duality construction connects the expenditure function and the indirect utility.
In optimization theory, the dual of a convex minimization problem is built from the convex conjugates of the objective and constraints.

We illustrate the  conjugate construction with three examples, chosen to match the special cases we shall  encounter in the market making theory.

\paragraph{Example: linear cost.}
Suppose $\kappa(q) = \kappa_{\mathrm{inv}} q$ for $q \geq 0$.
This is the canonical case under our environmental assumption C5.
The conjugate of $\kappa$ is
given by the function $\kappa^*(p) = \sup_{q \geq 0}\{pq - \kappa_{\mathrm{inv}} q\} =  \sup_{q \geq 0}\{(p - \kappa_{\mathrm{inv}})q\}$.
If $p \leq \kappa_{\mathrm{inv}}$, the supremum is attained at $q = 0$ and  equals $0$.
If $p > \kappa_{\mathrm{inv}}$, the expression $(p - \kappa_{\mathrm{inv}})q$ is unbounded above on $\Rp$,  so the supremum is $+\infty$.
Hence, 
\[
  \kappa^*(p) = \begin{cases} 0 & \text{if } p \leq \kappa_{\mathrm{inv}}, \\ +\infty & \text{if } p > \kappa_{\mathrm{inv}}. \end{cases}
\]
The conjugate  is the indicator function of the interval $[0, \kappa_{\mathrm{inv}}]$ shifted by zero.
Economically, the maximum marginal price is $\kappa_{\mathrm{inv}}$, beyond which the maker is unwilling to trade at any quantity.

\paragraph{Example: quadratic cost.}
Suppose $\kappa(q) = \kappa_{\mathrm{inv}} q^2 / 2$.
The conjugate is computed by maximizing $pq - \kappa_{\mathrm{inv}} q^2/2$ over $q$, which gives the first-order condition $p = \kappa_{\mathrm{inv}} q$, hence $q = p/\kappa_{\mathrm{inv}}$.
Substituting,
\[
  \kappa^*(p) = p \cdot \frac{p}{\kappa_{\mathrm{inv}}}  - \frac{\kappa_{\mathrm{inv}}}{2}\left(\frac{p}{\kappa_{\mathrm{inv}}}\right)^2 = \frac{p^2}{2\kappa_{\mathrm{inv}}}.
\]
The conjugate is also quadratic, with the reciprocal coefficient.
The quadratic family is self-dual under Legendre-Fenchel, up to coefficient reciprocation.
 At the special point $\kappa_{\mathrm{inv}} = 1$ (the natural  unit), the cost function and its conjugate coincide identically.

\paragraph{Example: power cost.}
Suppose $\kappa(q) = q^p / p$ for some $p > 1$.
The first-order condition $\tilde p = q^{p-1}$ gives $q = \tilde p^{1/(p-1)}$.
Substituting,
\[
  \kappa^*(\tilde p) = \tilde p \cdot \tilde  p^{1/(p-1)} - \frac{1}{p}\tilde p^{p/(p-1)} = \tilde p^{p/(p-1)}\left(1 - \frac{1}{p}\right) = \frac{\tilde p^{p^*}}{p^*},
\]
where $p^* = p/(p-1)$ is the conjugate exponent of $p$.
The duality swaps $p$ and $p^*$:
 a maker with cost exponent $p$ faces a dual marginal-price function with the conjugate exponent.

These three examples cover the linear case (canonical under C5), the quadratic case (self-dual in natural units), and the general power case (illustrating the conjugate-exponent interchange).
We now turn to the main duality theorem for market making.

\subsection{The Duality Theorem}

We now show how the quoting rule, under (A4) and a smooth convex inventory cost, encodes both members of the Legendre-Fenchel conjugate pair $(\kappa, \kappa^*)$.
The maker's indifference condition relates the inventory level  $q$ to the half-spread, which (under appropriate rescaling) plays the role of the marginal price conjugate to $q$.
The standard differential identities of convex conjugation then give the inverse map and the Legendre-Fenchel equality.

\needspace{14\baselineskip}
\begin{theorem}[Legendre-Fenchel Duality]
\label{thm:duality}
Under (A4) with smooth strictly convex inventory cost $\kappa$, the quoting rule encodes the Legendre-Fenchel conjugate pair $(\kappa, \kappa^*)$ as follows.
\begin{enumerate}
\item   The forward map $q \mapsto \tilde p = \kappa'(q)$ is the derivative  of $\kappa$.
By (A4), the rescaled half-spread is $\tilde p = (\lambda/\sigma^2)(\delta(q)/2 - \kappa_{\mathrm{adv}}\pi\sigma^2)$.
\item  The inverse map $\tilde p \mapsto q = (\kappa^*)'(\tilde p)$ is the derivative of the conjugate $\kappa^*$.
\item The Legendre-Fenchel identity holds at the indifference point, $\kappa(q) + \kappa^*(\tilde p) = q\tilde p$.
\end{enumerate}
\end{theorem}

\smallskip
\noindent\textit{Axiom dependency.} Theorem~\ref{thm:duality} relies on (A4) together with the smooth strictly convex cost structure, and is therefore a Layer~I result.

\begin{proof}
Part 1: (A4) rearranges to $\sigma^2 \kappa'(q)/\lambda  = \delta(q)/2 - \kappa_{\mathrm{adv}}\pi\sigma^2$, equivalently $\kappa'(q) = \tilde p$ with the indicated rescaling.

Part 2: this is Fenchel's identity for smooth strictly convex conjugates.
Under strict convexity, $\kappa'$ is strictly monotone and therefore globally invertible on its range; the inverse map is precisely $(\kappa^*)'$.
 A convex but not strictly convex $\kappa$ would admit flat segments where $\kappa'$ is constant and not invertible, so strict convexity is needed for Part 2.
The derivation is otherwise standard, and the equivalence $q = (\kappa^*)'(\tilde p) \iff \tilde p = \kappa'(q)$ holds for all $q$ in the relevant domain.

Part 3: the Young-Fenchel inequality $\kappa(q) + \kappa^*(\tilde p) \geq q\tilde p$ holds for all $(q, \tilde p)$ and is an equality at the indifference point, where $\tilde p = \kappa'(q)$.
The equality follows from the definition of $\kappa^*$ as a supremum, which is attained at $q$ when $\tilde p = \kappa'(q)$. 
Figure~\ref{fig:lf-geometry} illustrates the geometry.
\end{proof}

\begin{figure}[t]
\centering
\begin{tikzpicture}[scale=1.5,>=Stealth]
    % --- Axes
  \draw[->, gray!70] (-0.3, 0) -- (3.4, 0) node[right, black] {$q$};
  \draw[->, gray!70] (0, -2) -- (0, 3.4) node[above, black] {};

   % ---  Convex curve kappa(q) = q^2 / 2
  \draw[thick, NavyBlue, domain=0:2.7, samples=80, smooth]
        plot ({\x}, {0.5*\x*\x}) node[right, NavyBlue] {$\kappa(q)$};

  % --- Marked point at q0 = 1.8, kappa(q0) = 1.62
  \node[circle, fill=black, inner sep=1.1pt] at (1.8, 1.62) {};
  \draw[dashed, gray!70] (1.8, 0) -- (1.8, 1.62);
  \node[below] at (1.8, 0) {$q_0$};

  % --- Tangent line at q0: slope = q0 = 1.8
  \draw[thick, BrickRed] (0, -1.62) -- (2.8, 3.42);

  % --- y-intercept of tangent = -kappa*(p0)
  \node[circle, fill=BrickRed, inner sep=1.1pt] at (0, -1.62) {};
%  \draw[thick, BrickRed] (-0.15, -1.62) -- (-0.95, -1.62)
%        node[left, BrickRed] {$-\kappa^*(p_0)$};
  \node[left, BrickRed] at (-0.15, -1.62) {$-\kappa^*(p_0)$};
%   \draw[BrickRed, thick, dashed] (-0.10, -1.62) -- (-0.25, -1.62);
% \draw[BrickRed, thick] (-0.05, -1.62) -- (-0.25, -1.62);
\draw[
  BrickRed,
  line width=0.1pt,
  dash pattern=on 0.01pt off 0.5pt,
  line cap=round
] (-0.05,-1.62) -- (-0.25,-1.62);
  % --- Slope label on  tangent
  \node[BrickRed, right] at (2.70, 3.25) {slope $p_0 = \kappa'(q_0)$};

  % --- Brace indicating kappa(q0)
  \draw[decorate, decoration={brace,amplitude=4pt,mirror}, gray!70]
        (1.95, 0.02) -- (1.95, 1.6) node[midway, right=4pt, black] {$\kappa(q_0)$};
\end{tikzpicture}
\caption{Legendre-Fenchel geometry for the quadratic cost $\kappa(q) = q^2/2$.
The tangent to $\kappa$ at inventory level $q_0$ has slope $p_0 = \kappa'(q_0)$, which is the marginal price, and vertical intercept $-\kappa^*(p_0)$.
Sweeping $q_0$ across all inventory levels traces out both $\kappa$ and its conjugate $\kappa^*$ simultaneously: 
the slopes recover the forward map $q_0 \mapsto p_0$, and the intercepts recover the dual values $p_0 \mapsto \kappa^*(p_0)$.
The quadratic is self-dual: $\kappa^*(p) = p^2/2$.}
\label{fig:lf-geometry}
\end{figure}
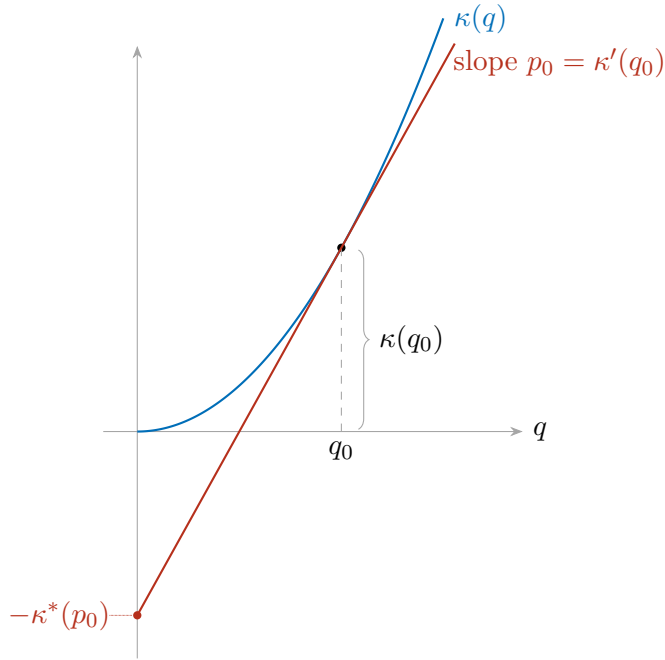

\subsection{The Market  Making Interpretation of the Special Cases}

The conjugate examples  in Section~\ref{subsec:duality-conjugate} are now reread as statements about market making.

\paragraph{Linear cost (C5).}
Under C5, the maker has constant marginal cost.
The cost function and its conjugate satisfy the canonical pair:  $\kappa(q) = \kappa_{\mathrm{inv}} q$ and $\kappa^*(\tilde p) = 0$ for $\tilde p \leq \kappa_{\mathrm{inv}}$, with $\kappa^*(\tilde p) = +\infty$ for $\tilde p > \kappa_{\mathrm{inv}}$.
The market making interpretation is that the maker is willing to fill at any half-spread up to $\kappa_{\mathrm{inv}}\sigma^2/\lambda$ (the marginal cost per trade); beyond this half-spread, she would prefer to scale up her position without bound, but the indifference condition fixes the unique half-spread at exactly $\kappa_{\mathrm{inv}}\sigma^2/\lambda$.

\paragraph{Quadratic cost.}
Under a quadratic cost $\kappa(q) = \kappa_{\mathrm{inv}} q^2/2$, the order book at level $n$ posts a half-spread proportional to $|q + n - 1|$, by Theorem~\ref{thm:order-book}.
The depth profile is linear in $n$, with slope $2\kappa_{\mathrm{inv}}\sigma^2 / \lambda$.
The dual structure says that, equivalently, the maker's inventory at half-spread $\tilde p$ is  proportional to $\tilde p$,  with the inverse proportionality.
The two pictures are equivalent under Legendre-Fenchel: 
 one looks at inventory as a  function of half-spread, the other looks at half-spread as a function of inventory.

\paragraph{Power cost.}
Under a power cost $\kappa(q) = q^p / p$ for $p > 1$, the order book depth profile  is concave in $n$  for $p < 2$ (sublinear cost growth) and convex in $n$ for $p > 2$ (superlinear cost growth).
The dual picture has the inverse curvature, characterized by the conjugate exponent $p^*$.
The maker's structural cost exponent and the dual exponent are connected by the conjugate-exponent  relation $1/p + 1/p^* = 1$.

\begin{remark}[Recovery and Duality]
\label{rem:recovery-duality}
Theorem~\ref{thm:duality} unifies the recovery result (Theorem~\ref{thm:recovery}) with the duality structure.
The spread-cost relation $\tilde p = \kappa'(q)$ is the first-order  expression of Legendre-Fenchel duality;
  the half-spread reveals the derivative of $\kappa$.
The Recovery Theorem is the integral expression of the same duality; $\kappa$ is recovered by integrating $\kappa'$ across the depth profile.
Both $\kappa$ and $\kappa^*$ are recoverable from observable market data, since the depth profile gives $\kappa'$ and the conjugate $\kappa^*$ is determined by $\kappa$ via the conjugacy operation.
\end{remark}

The duality structure places market making within the broader mathematics of convex analysis.
The maker's inventory  cost and her marginal-price function are conjugate variables under the Legendre-Fenchel transform.
The limit order book reveals both.
The reformulation has  aesthetic and mathematical content: it identifies market making as a particular instance of a general mathematical structure that recurs across physics (thermodynamic duality), economics (consumer theory), and optimization theory (dual programs).

%-------------------------------------------------------------------
\section{Phase Transitions and the Three Liquidity Regimes}
\label{sec:phase}
%-------------------------------------------------------------------

We now turn to a sharp empirical implication of the characterization.
Combining the spread formula~\eqref{eq:spread-decomposition} with a disagreement bound between maker and counterparties yields a condition for market viability.
The condition takes the form of a phase transition.
A sharp boundary separates a functioning regime from a frozen regime, with the functioning regime partitioning into substructures by which  spread channel dominates.
The phase boundary  and the within-region substructure together provide a structural account of liquidity at the level of the quoting rule itself.

\subsection{The Disagreement  Bound}

Consider a counterparty with belief $\mu_C$ about the asset value.
The counterparty buys at the maker's ask $a$ only if buying is profitable, that is, $\mu_C > a$.
This requires the disagreement $\mu_C - \mu$ to exceed the maker's  markup $a -  \mu = \delta/2$.
Symmetrically, the counterparty sells at the maker's bid $b$ only if $\mu_C < b$, requiring the symmetric disagreement.
The maker, with belief $\mu$, posts a strictly positive spread (by definition of the codomain  $\mathcal{B}$), so a trade can occur only when some counterparty's disagreement exceeds the half-spread.

Recall from Section~\ref{sec:setup} that $\delta_{\max} > 0$ denotes the maximum admissible spread, defined as $2\max_C|\mu_C - \mu|$, twice the maximum belief disagreement between the maker and any counterparty in the trading environment.
The factor of two is precisely  the half-spread on either side, so the condition that some counterparty's disagreement exceed the half-spread is equivalent to $\delta < \delta_{\max}$.
For the market to be functioning, that is, to generate positive trade volume,  the equilibrium spread must therefore satisfy $\delta < \delta_{\max}$.
When $\delta \geq \delta_{\max}$, no counterparty has sufficient disagreement to find trading profitable, and the market freezes.

\subsection{The Phase Transition}

Substituting the equilibrium spread from~\eqref{eq:spread-decomposition} into the viability condition $\delta < \delta_{\max}$ gives a sharp characterization.

\begin{theorem}[Liquidity Phase Transition]
\label{thm:phase}
Given the maximum admissible spread $\delta_{\max} > 0$, the structural parameters $\kappa_{\mathrm{inv}}, \kappa_{\mathrm{adv}} > 0$, the volatility $\sigma^2 > 0$, and the informed-trader fraction $\pi \in [0,1]$, the market is functioning if and only if
\begin{equation}
\label{eq:phase-condition}
  \frac{\kappa_{\mathrm{inv}}}{\lambda} + \kappa_{\mathrm{adv}}\pi  < \frac{\delta_{\max}}{2\sigma^2}.
\end{equation}
\end{theorem}

\smallskip
\noindent\textit{Axiom dependency.} Theorem~\ref{thm:phase} relies on the core (Layer~I) together with the structural choice (A3).

\begin{proof}
The market functions if and   only if the equilibrium spread satisfies $\delta < \delta_{\max}$.
By~\eqref{eq:spread-decomposition}, $\delta = 2\sigma^2(\kappa_{\mathrm{inv}}/\lambda + \kappa_{\mathrm{adv}}\pi)$.
Substituting and dividing by $2\sigma^2 > 0$ gives the stated inequality.
\end{proof}

\subsection{Critical Thresholds and the Phase Diagram}

The phase transition condition~\eqref{eq:phase-condition} admits three natural extreme cases.

\paragraph{Critical liquidity.}
At fixed $\pi$, the critical arrival intensity is
\[
  \lambda^*(\pi) = \frac{\kappa_{\mathrm{inv}}}{\delta_{\max}/(2\sigma^2) - \kappa_{\mathrm{adv}}\pi}.
\]
For $\lambda > \lambda^*$ the market functions; for  $\lambda < \lambda^*$ the market freezes.
The critical liquidity depends on the informed fraction; more informed flow raises the liquidity threshold.

\paragraph{Critical informed fraction.}
At $\lambda \to \infty$ (perfectly liquid), the inventory term in~\eqref{eq:phase-condition} vanishes, and the boundary depends on adverse  selection alone:
\[
  \pi^* = \frac{\delta_{\max}}{2\kappa_{\mathrm{adv}}\sigma^2}.
\]
For $\pi > \pi^*$ the market is frozen regardless of liquidity.
The adverse-selection premium alone exceeds $\delta_{\max}$;  no level of trade frequency can sustain trade.
If $\delta_{\max} > 2\kappa_{\mathrm{adv}}\sigma^2$, then $\pi^* > 1$, and adverse selection alone cannot freeze the market regardless of the informed fraction; the frozen regime then requires sufficiently low $\lambda$.

\paragraph{Critical volatility.}
At fixed $(\lambda, \pi)$, the maximum volatility consistent with a functioning market is
\[
  \sigma^2_{\max}(\lambda, \pi) = \frac{\delta_{\max}}{2(\kappa_{\mathrm{inv}}/\lambda + \kappa_{\mathrm{adv}}\pi)}.
\]
Beyond $\sigma^2_{\max}$ the market freezes.
In particular, a spike in uncertainty can  push the market  across the phase boundary even with constant $(\lambda, \pi)$.

\subsection{The Three Liquidity Regimes}

Within the functioning  region, the spread admits a substructure governed by which of its two components dominates.

\needspace{14\baselineskip}
\begin{theorem}[Three Liquidity Regimes]
\label{thm:three-regimes}
Define the crossover liquidity $\lambda_\times = \kappa_{\mathrm{inv}}/(\kappa_{\mathrm{adv}}\pi)$.
Assume the ordering condition
\[
  \kappa_{\mathrm{adv}}\pi < \frac{\delta_{\max}}{4\sigma^2},
\]
equivalently $\lambda^*(\pi) < \lambda_\times$.
The parameter space then partitions into three liquidity regimes.

\emph{Frozen regime} ($\lambda < \lambda^*(\pi)$): the equilibrium spread exceeds $\delta_{\max}$, so no trade occurs.

\emph{Inventory-dominated regime} ($\lambda^*(\pi) < \lambda  < \lambda_\times$):  the inventory spread $2\kappa_{\mathrm{inv}}\sigma^2/\lambda$ exceeds the adverse-selection spread $2\kappa_{\mathrm{adv}}\pi\sigma^2$.
The market is functioning, but inventory cost is the binding consideration in the maker's quoting.

\emph{Adverse-selection-dominated regime} ($\lambda > \lambda_\times$): the adverse-selection spread exceeds the inventory spread.
As $\lambda \to \infty$, the spread converges to $2\kappa_{\mathrm{adv}}\pi\sigma^2$, the pure adverse-selection spread.
\end{theorem}

\smallskip
\noindent\textit{Axiom dependency.} Theorem~\ref{thm:three-regimes} relies on the core (Layer~I) together with the structural choice (A3).

\begin{proof}
 The frozen-functioning boundary is given by Theorem~\ref{thm:phase}.
The inventory-vs-adverse-selection boundary is the equation $\kappa_{\mathrm{inv}}/\lambda = \kappa_{\mathrm{adv}}\pi$, which solves to $\lambda = \kappa_{\mathrm{inv}}/(\kappa_{\mathrm{adv}}\pi) = \lambda_\times$.
The ordering condition $\kappa_{\mathrm{adv}}\pi < \delta_{\max}/(4\sigma^2)$ is equivalent to $\lambda^*(\pi) < \lambda_\times$, since
\begin{align*}
  \lambda^*(\pi) < \lambda_\times
    &\iff \frac{\kappa_{\mathrm{inv}}}{\delta_{\max}/(2\sigma^2) - \kappa_{\mathrm{adv}}\pi} < \frac{\kappa_{\mathrm{inv}}}{\kappa_{\mathrm{adv}}\pi} \\
    &\iff \kappa_{\mathrm{adv}}\pi < \frac{\delta_{\max}}{2\sigma^2} - \kappa_{\mathrm{adv}}\pi \\
    &\iff \kappa_{\mathrm{adv}}\pi < \frac{\delta_{\max}}{4\sigma^2}.
\end{align*}
Under this ordering condition, the two thresholds partition $\Rpp$ into the three intervals stated.
The asymptotic identity $\lim_{\lambda\to\infty}\delta = 2\kappa_{\mathrm{adv}}\pi\sigma^2$ follows by direct substitution into~\eqref{eq:spread-decomposition}.
\end{proof}

\begin{remark}[The intermediate adverse-selection regime]
\label{rem:two-regime}
When the ordering condition $\kappa_{\mathrm{adv}}\pi < \delta_{\max}/(4\sigma^2)$  fails but the phase condition $\kappa_{\mathrm{adv}}\pi < \delta_{\max}/(2\sigma^2)$ still holds, the market functions for sufficiently large $\lambda$ but the crossover liquidity $\lambda_\times$ lies in the frozen region.
At $\lambda = \lambda_\times$ the spread equals $4\kappa_{\mathrm{adv}}\pi\sigma^2 \geq \delta_{\max}$.
The inventory-dominated functioning interval is empty in  this case, and every functioning $\lambda$ has spread dominated by the adverse-selection component.
The three-regime decomposition then degenerates to a two-regime decomposition: frozen ($\lambda < \lambda^*(\pi)$) and adverse-dominated ($\lambda > \lambda^*(\pi)$).
\end{remark}

\noindent\textit{Edge case.} 
The crossover liquidity $\lambda_\times = \kappa_{\mathrm{inv}}/(\kappa_{\mathrm{adv}}\pi)$ is well-defined only when $\pi > 0$;
at $\pi = 0$ the adverse-selection spread vanishes and the regime decomposition collapses to a single functioning regime, with $\lambda_\times = +\infty$ formally and the phase condition reducing to $\kappa_{\mathrm{inv}}/\lambda < \delta_{\max}/(2\sigma^2)$.

The three-regime structure has direct empirical content.
Empirical markets  can be classified by which regime they inhabit; the boundaries are sharp and the regime identification is a straightforward exercise.
The structural  drivers are clear: increases in volatility ($\sigma^2 \uparrow$),  informed flow ($\pi \uparrow$), or carrying cost ($\kappa_{\mathrm{inv}} \uparrow$), or decreases in arrival intensity ($\lambda \downarrow$), all push a  market   toward the frozen boundary. 
Figure~\ref{fig:phase-diagram} shows the partition.

\begin{figure}[t]
\centering
\begin{tikzpicture}[scale=3.3,>=Stealth,font=\small]
    % --- Frozen region (left of lambda*(pi)): pi from 0 to 1, lambda from 0 to 1/(2-pi)
  \fill[Salmon!40]
    (0, 0) --
    plot[domain=0:1, samples=50, smooth] ({1/(2-\x)}, \x) --
    (0, 1) -- cycle;

  % --- Inventory-dominated region (between lambda*(pi) and lambda_x(pi))
  \fill[ProcessBlue!25]
    plot[domain=0:1, samples=50, smooth] ({1/(2-\x)}, \x) --
    (1, 1) --
    plot[domain=1:0.25, samples=50, smooth] ({1/\x}, \x) --
    (4, 0) -- (0.5, 0) -- cycle;

  % --- Adverse-selection-dominated  region (right of lambda_x(pi) for pi >= 0.25)
  \fill[ForestGreen!25]
    plot[domain=0.25:1, samples=50, smooth] ({1/\x}, \x) --
    (4, 1) -- (4, 0.25) -- cycle;

  % --- Axes
  \draw[->, thick] (0, 0) -- (4.3, 0) node[right] {$\lambda$};
  \draw[->, thick] (0, 0) -- (0, 1.15) node[above] {$\pi$};
  \draw[dashed, gray] (0, 1) -- (4.2, 1);
  \node[left] at (0, 1) {$1$};
  \node[below left] at (0, 0) {$0$};

  % ---  Boundary curves
  \draw[very thick, BrickRed, domain=0:1, samples=50, smooth]
    plot ({1/(2-\x)}, \x);
  \draw[very thick, NavyBlue, domain=0.25:1, samples=50, smooth]
    plot ({1/\x}, \x);

  % --- Region labels
  \node at (0.32, 0.5) {Frozen};
  \node[align=center] at (1.55, 0.32) {Inventory-\\dominated};
  \node[align=center] at (3, 0.65) {Adverse-selection-\\dominated};

  % --- Curve labels
  \node[BrickRed, right] at (0.465, 0.76) {$\lambda^*(\pi)$};
  \node[NavyBlue, right] at (2.65, 0.265) {$\lambda_\times(\pi)$};
\end{tikzpicture}
\caption{Phase diagram in $(\lambda, \pi)$ space. 
 On the left, the frozen region $\{\lambda < \lambda^*(\pi)\}$ (salmon) lies behind the phase boundary $\lambda^*(\pi) = \kappa_{\mathrm{inv}}/(\delta_{\max}/(2\sigma^2) - \kappa_{\mathrm{adv}}\pi)$; the functioning region lies to its right.
The functioning region partitions further into an inventory-dominated regime $\{\lambda^*(\pi) < \lambda < \lambda_\times(\pi)\}$ (blue) and an adverse-selection-dominated regime $\{\lambda > \lambda_\times(\pi)\}$ (green), with crossover curve 
 $\lambda_\times(\pi) = \kappa_{\mathrm{inv}}/(\kappa_{\mathrm{adv}}\pi)$. 
The depicted curves use illustrative parameter values $\kappa_{\mathrm{inv}} = \kappa_{\mathrm{adv}} = \sigma^2 = 1$ and $\delta_{\max} = 4$, under which the ordering condition $\kappa_{\mathrm{adv}}\pi < \delta_{\max}/(4\sigma^2)$ of Theorem~\ref{thm:three-regimes} holds throughout $\pi \in (0, 1)$; under failure of that condition (Remark~\ref{rem:two-regime}), the inventory-dominated region is empty and the diagram collapses to a frozen/adverse-dominated dichotomy.}
\label{fig:phase-diagram}
\end{figure}
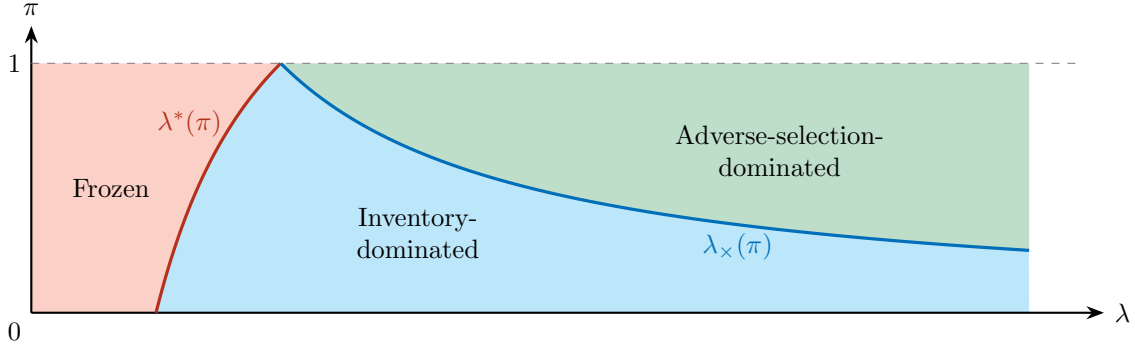

\begin{remark}[Market Freezes as Phase Transitions]
\label{rem:freezes}
The phase transition theorem provides a structural account of empirically observed market freezes.
During the 2008 financial crisis, several markets exhibited episodes in which trading ceased despite the apparent presence  of willing makers.
Such freezes are difficult to rationalize within  standard optimization models, since each maker individually has positive expected profit from quoting.
Within the axiomatic framework, freezes correspond to crossings of the phase boundary~\eqref{eq:phase-condition}: a spike in volatility, in adverse selection, or  in inventory cost (or a collapse in trade arrival intensity) can push the implied spread beyond $\delta_{\max}$, eliminating the willing-counterparty population.
The maker's individual profitability is not the issue.
The aggregate market clears only  when the spread is below the maximum disagreement, and that  condition can fail even when individual makers are profitable in expectation.
\end{remark}

\subsection{Information-Minimizing Liquidity}
\label{subsec:info-max}

The crossover liquidity $\lambda_\times = \kappa_{\mathrm{inv}}/(\kappa_{\mathrm{adv}}\pi)$ defined as the inventory-versus-adverse-selection boundary in Theorem~\ref{thm:three-regimes} admits a second, distinct characterization in information-theoretic terms.

Each trade conveys information about the asset value $v$.
We derive the per-trade information content under an explicit illustrative model of Bayesian update, which we now make precise.

\paragraph{Gaussian-update model.}
We adopt the following modeling premise, used solely to derive a concrete form for   the per-trade information content and the resulting information rate.
\begin{enumerate}[label=(M\arabic*),itemsep=0pt]
\item The maker's prior over $v$ is $\mathcal{N}(\mu, \sigma^2)$.
\item A counterparty trade at  the ask $a$ (respectively the bid $b$) is interpreted as evidence that the counterparty's reservation value exceeds $a$
  (respectively falls below $b$), and induces a posterior update on $v$.
\item Under any reasonable Gaussian update, the Kullback-Leibler divergence between the posterior and the prior is proportional to $(a - \mu)^2/\sigma^2$ to leading order in $(a-\mu)/\sigma$, that is, the per-trade information  content scales as $(a-\mu)^2/\sigma^2$ with a constant of proportionality $\kappa_{\mathrm{info}} > 0$ that depends on the update model but not on the trading environment $(\lambda, \sigma^2)$.
\end{enumerate}

Under (M1)--(M3), the information content of one trade is, to leading order,
\[
  \kappa_{\mathrm{info}}\,\frac{(a - \mu)^2}{\sigma^2}.
\]
At zero inventory and zero skew, $a - \mu = \delta/2$, giving information content  $\kappa_{\mathrm{info}}\,(\delta/2)^2/\sigma^2$ per trade.
Multiplying by the arrival intensity $\lambda$ gives the information rate
\[
  R_{\mathrm{info}}(\lambda) = \kappa_{\mathrm{info}}\,\frac{\lambda\,(\delta/2)^2}{\sigma^2}.
\]
The constant $\kappa_{\mathrm{info}}$ does not depend on $\lambda$ or on the structural parameters $(\kappa_{\mathrm{inv}}, \kappa_{\mathrm{adv}})$, so the location of any stationary point of $R_{\mathrm{info}}$ in $\lambda$ is unaffected by $\kappa_{\mathrm{info}}$.
Substituting the equilibrium spread $\delta = 2\sigma^2(\kappa_{\mathrm{inv}}/\lambda + \kappa_{\mathrm{adv}}\pi)$ from~\eqref{eq:characterization}, $(\delta/2)^2 = \sigma^4(\kappa_{\mathrm{inv}}/\lambda + \kappa_{\mathrm{adv}}\pi)^2$, and
\[
  R_{\mathrm{info}}(\lambda) = \kappa_{\mathrm{info}}\,\lambda\sigma^2\left(\frac{\kappa_{\mathrm{inv}}}{\lambda} + \kappa_{\mathrm{adv}}\pi\right)^2.
\]

\begin{theorem}[Minimum Information Rate at the Crossover]
\label{thm:info-max}
Assume (M1)--(M3); these are modeling premises external to the axiom system.
Under (M1)--(M3) the information rate $R_{\mathrm{info}}(\lambda)$  attains a unique stationary point at the crossover liquidity $\lambda = \lambda_\times = \kappa_{\mathrm{inv}}/(\kappa_{\mathrm{adv}}\pi)$ of Theorem~\ref{thm:three-regimes}.
At $\lambda_\times$ the two spread components are equal:
the inventory component $2\kappa_{\mathrm{inv}}\sigma^2/\lambda_\times$ coincides with the adverse-selection component $2\kappa_{\mathrm{adv}}\pi\sigma^2$.
The stationary  point is a minimum, and $R_{\mathrm{info}}$ is U-shaped in $\lambda$ on $\Rpp$.
The location of the minimum is  independent of the proportionality constant $\kappa_{\mathrm{info}}$ in (M3).
\end{theorem}

\smallskip
\noindent\textit{Axiom dependency.} 
The proof rests on the axiomatic spread formula (Theorem~\ref{thm:full}) together with the Gaussian-update modeling premises (M1)--(M3); 
 the axioms supply the spread, the modeling premises supply the information-rate functional.
The information-rate coincidence with the crossover liquidity is a feature of the combined axiom-plus-information-model setup, not a free-standing consequence of the axiom system; a different information model would yield a different functional form for $R_{\mathrm{info}}$ and, in general, a different location of its stationary point.

\begin{proof}
Write $R_{\mathrm{info}}(\lambda) = \kappa_{\mathrm{info}}\sigma^2\,f(\lambda)$, where $f(\lambda) = \lambda(\kappa_{\mathrm{inv}}/\lambda + \kappa_{\mathrm{adv}}\pi)^2$.
We then substitute 
 $u = \sqrt{\lambda}$, which gives $f = (\kappa_{\mathrm{inv}}/u + \kappa_{\mathrm{adv}}\pi\,u)^2$.
The bracket $\kappa_{\mathrm{inv}}/u + \kappa_{\mathrm{adv}}\pi\,u$ is minimized by the arithmetic-geometric mean inequality   at $u^* = \sqrt{\kappa_{\mathrm{inv}}/(\kappa_{\mathrm{adv}}\pi)}$, where it equals $2\sqrt{\kappa_{\mathrm{inv}}\kappa_{\mathrm{adv}}\pi}$.
The bracket diverges as $u \to 0^+$ and as $u \to \infty$, since $1/u$ and $u$ each diverge at the respective ends.
Squaring is monotone on $\Rp$, so $f$ has the same critical structure as the bracket, minimized at $u^*$ and diverging at both ends.
In $\lambda$ coordinates, the minimum is at $\lambda_\times = (u^*)^2 = \kappa_{\mathrm{inv}}/(\kappa_{\mathrm{adv}}\pi)$.
\end{proof}

\begin{remark}[Crossover as Information Minimum]
\label{rem:info-min}
A structural reading of Theorem~\ref{thm:info-max} is the following.
The crossover liquidity marks the boundary between the inventory and adverse-selection regimes (Section~\ref{sec:phase}); under the leading-order Gaussian information model (M1)--(M3), it is also the \emph{least informative} liquidity.
In both limits, the information rate exceeds the balanced rate and in fact diverges.

That the regime boundary and the information minimum coincide does not follow from either  definition on its own.
Two distinct criteria, namely structural regime classification and information rate, pick out the same  liquidity under (M1)--(M3).
The coincidence suggests a connection between the spread decomposition and the information content of trades that the axiomatic theory makes precise, but the prediction is a joint feature of the axioms together with the leading-order Gaussian information model.
A different information model would generically place the minimum elsewhere; in particular, the leading-order approximation in (M3) breaks down where the spread approaches $\delta_{\max}$, which is the regime in which the predicted U-shape would be most testable.
Empirical or policy implications of the location of the minimum should accordingly be read as conditional on (M1)--(M3) rather than as direct consequences of the axiom system.
\end{remark}

%-------------------------------------------------------------------
\section{Comparative Statics}
\label{sec:comparative}
%-------------------------------------------------------------------

The closed-form characterization~\eqref{eq:characterization} makes the comparative statics of the quoting rule immediate.
We collect the nontrivial ones here.
Each comparative-statics result   follows by direct differentiation of~\eqref{eq:characterization}; the interpretive content is in the  structural reading of the  signs and magnitudes.

\subsection{The Mid-Quote}

The mid-quote depends on four of the five state variables; the dependence  on each follows directly from the characterization formula.
We collect the four partial derivatives below.

\paragraph{Belief.}
The derivative $\partial m/\partial \mu = 1$, since $m(s) = \mu - \gamma_0\sigma^2 q$ by (A2).
The mid-quote tracks the maker's belief one-for-one.

\paragraph{Inventory.}
The derivative $\partial m/\partial q = -\gamma_0\sigma^2 < 0$, by (A1) together with the characterization.
The slope identifies $\gamma_0$.
Empirically, this is the slope of the maker's quote-skew curve as her inventory varies.

\paragraph{Variance.}
The derivative $\partial m/\partial \sigma^2 =  -\gamma_0 q$.
The mid-quote shifts more aggressively under high volatility when the maker holds  inventory.
The sign depends on the sign of $q$.

\paragraph{Liquidity and information.}
The derivatives $\partial m/\partial \lambda = 0$ and $\partial m/\partial \pi = 0$, by (A5) and (A6).
The mid-quote does not depend on the  trading environment, conditional on the maker's belief and inventory.

\subsection{The Spread}

The spread depends on four of the five state variables, with $q$ acting only indirectly through the cost-function curvature, which vanishes under the linear (C5).
We collect the four partial derivatives below.

\paragraph{Inventory.}
The derivative $\partial \delta/\partial q = 0$.
The spread is constant across inventory levels in the linear-cost  case (C5): the maker's spread does not widen when she holds large inventory.
Under general convex cost (Section~\ref{sec:order-book}), the spread does widen with $|q|$ at the rate $\kappa''$, recovering an inventory-dependent spread within the general framework.

\paragraph{Variance.}
The derivative $\partial \delta/\partial \sigma^2 = 2(\kappa_{\mathrm{inv}}/\lambda + \kappa_{\mathrm{adv}}\pi) > 0$.
Higher volatility widens the spread uniformly.

\paragraph{Liquidity.}
The derivative $\partial \delta/\partial \lambda = -2 \kappa_{\mathrm{inv}}\sigma^2/\lambda^2 < 0$.
More frequent trades reduce the per-trade carrying cost and compress the  spread.
The functional form is the canonical $1/\lambda$~shape.

\paragraph{Information.}
The derivative $\partial \delta/\partial \pi = 2\kappa_{\mathrm{adv}}\sigma^2 > 0$.
More informed flow widens the spread  linearly.
This  identifies $\kappa_{\mathrm{adv}}$.

 \subsection{Cross-Partial Derivatives}

The first-order partial derivatives of $m$ and $\delta$ describe the marginal response of the quoting rule to each state variable  in isolation.
Cross-partial derivatives describe how  the marginal response to one variable depends on another, capturing  interaction effects.
We collect the nonzero cross-partials.

\paragraph{Inventory-variance interaction in the mid-quote.}
The cross-partial is equal to 
$\partial^2 m/\partial q \partial \sigma^2 = -\gamma_0$.
The slope of the mid-quote in inventory becomes more negative in absolute value as volatility rises.
Equivalently,  the maker's inventory-hedging response is more aggressive under high volatility.
The cross-partial identifies $\gamma_0$ from a volatility-interacted regression of $m$ on $q$.

\paragraph{Volatility-liquidity interaction in the spread.}
The cross-partial here is
\[
  \frac{\partial^2 \delta}{\partial \sigma^2 \, \partial \lambda} = -\frac{2\kappa_{\mathrm{inv}}}{\lambda^2} < 0.
\]
The sensitivity of the spread to volatility decreases with liquidity.
In liquid markets, the same volatility shock widens the spread by less than in  illiquid markets, because the inventory cost is diluted across more trades.

\paragraph{Volatility-information interaction in the spread.}
The cross-partial 
is equal to 
$\partial^2 \delta/\partial \sigma^2 \partial \pi = 2\kappa_{\mathrm{adv}} > 0$.
The sensitivity of the spread to volatility increases with the informed-trader fraction.
The cross-partial is constant in the state, identifying $\kappa_{\mathrm{adv}}$ from a volatility-interacted regression of $\delta$ on $\pi$.

\paragraph{Liquidity-information interaction in the spread.}
The cross-partial $\partial^2 \delta/\partial \lambda \partial \pi = 0$.
The sensitivities to $\lambda$ and to $\pi$ are structurally decoupled.
This is a consequence of additive separability (Theorem~\ref{thm:separability}).
Empirically, this cross-restriction is a testable implication of the axiom system.

\subsection{Elasticities}

The partial derivatives above quantify how the spread and mid-quote respond,  in absolute terms, to each parameter.
For empirical applications,  the corresponding elasticities are often more  directly interpretable, since they are dimensionless and unit-free.

\paragraph{Spread elasticity in  liquidity.}
\[
  \epsilon_{\delta, \lambda} := \frac{\partial \log \delta}{\partial \log \lambda} = -\frac{\kappa_{\mathrm{inv}}/\lambda}{\kappa_{\mathrm{inv}}/\lambda + \kappa_{\mathrm{adv}}\pi}.
\]
The elasticity equals $-1$ in the inventory-dominated regime ($\kappa_{\mathrm{inv}}/\lambda \gg \kappa_{\mathrm{adv}}\pi$) and approaches $0$ in the adverse-selection-dominated regime ($\kappa_{\mathrm{inv}}/\lambda \ll \kappa_{\mathrm{adv}}\pi$).
The  crossover happens at  the critical liquidity $\lambda_\times$ where the elasticity equals $-1/2$.
Empirically, $\epsilon_{\delta, \lambda}$ is identified from a  log-log regression of $\delta$ on  $\lambda$; the value of the elasticity localizes the market in the three-regime structure.

\paragraph{Spread elasticity in informed fraction.}
\[
  \epsilon_{\delta, \pi} := \frac{\partial \log \delta}{\partial \log \pi} = \frac{\kappa_{\mathrm{adv}}\pi}{\kappa_{\mathrm{inv}}/\lambda + \kappa_{\mathrm{adv}}\pi}.
\]
Mirror image of the previous case:  equals $0$ in the inventory-dominated regime and approaches $1$ in the adverse-selection-dominated regime.

\paragraph{Spread elasticity in volatility.}
\[
  \epsilon_{\delta, \sigma^2} := \frac{\partial \log \delta}{\partial \log \sigma^2} = 1.
\]
The unit elasticity in $\sigma^2$ is a direct consequence of (A3) (variance homogeneity).
Empirically, a deviation from unit elasticity is a violation of (A3).

\paragraph{Skew elasticity in inventory.}
\[
  \epsilon_{m - \mu, q} := \frac{\partial \log|m - \mu|}{\partial \log|q|} = 1.
\]
The skew is linear in $q$, so its elasticity in $q$ is unit.
A deviation from unit elasticity is a violation of (A2).

\subsection{Comparison with Existing Models}

The comparative statics enable direct comparison with the optimization-based models in the literature.

Ho-Stoll and Avellaneda-Stoikov each predict a spread of the form $\delta = A \sigma^2 + B$ for some structural  constants $A, B$;  Glosten-Milgrom predicts $\delta \propto \pi\sigma$, with $\sigma$-scaling rather  than $\sigma^2$-scaling (and is the $\beta = 1/2$ structural neighbor of our family, characterized by (A3s); see Section~\ref{sec:universal}).
Our characterization predicts $\delta = (2\kappa_{\mathrm{inv}}/\lambda + 2\kappa_{\mathrm{adv}}\pi)\sigma^2$, separating the variance   coefficient into an  inventory part and an adverse-selection part.
The decomposition is structural.
Each component scales differently with the structural variables.

In Ho-Stoll, the inventory cost coefficient enters both the spread and the  skew with the same parameter, so the two contributions cannot be separated empirically from a single regression.
In Avellaneda-Stoikov, the risk-aversion  parameter $\gamma$ enters both components similarly.
The axiomatic decomposition resolves this conflation.
Our $\gamma_0$ is the skew coefficient alone, identified from $\partial m/\partial q$.
Our $\kappa_{\mathrm{inv}}$ is the spread coefficient at $\pi = 0$, identified from a regression of $\delta$ on $1/\lambda$ holding $\pi$ at zero.
Our $\kappa_{\mathrm{adv}}$ is the slope of the spread in $\pi$,  identified from $\partial \delta/\partial \pi$.

The economic content is that what optimization-based models compress into one or two parameters, our axiomatic theory decomposes into three.
The reason is structural: the axiom system enforces separability (Theorem~\ref{thm:separability}), so different optimization-based parameters that mix preference, cost, and  information can be decomposed into the three orthogonal axiomatic parameters.

\subsection{An Empirical  Implementation Roadmap}
\label{subsec:empirical-roadmap}

The identification corollary (Corollary~\ref{cor:identification}) gives the formal mapping from observable quoting behavior to the three parameters.
We now describe a step-by-step implementation roadmap for an empirical researcher.
The roadmap assumes that the researcher has access to quote-level  data, that is, the bid and ask of one or more market makers, together with the inventory position and the relevant state variables.

\paragraph{Data requirements.}
The minimal data are time-stamped observations of:
(i) the maker's bid and ask, 
(ii) her inventory $q$, 
(iii) the public belief $\mu$ (typically the mid-quote of the consolidated order book or a Bayesian estimate), 
(iv) the price variance $\sigma^2$ (typically   an exponentially-weighted realized variance), 
 (v) the trade arrival intensity $\lambda$,   and 
(vi) a proxy for the informed-trader fraction $\pi$ (typically a volume-weighted order imbalance or PIN-style estimator).
Items (i)--(iii) are read off the quote stream and the maker's inventory record; (iv)--(vi) are constructed from market-level data.

\paragraph{Step 1, identification of $\gamma_0$.}
Run a regression of the mid-quote on inventory:
\[
  m_t - \mu_t = -\gamma_0 \sigma_t^2 q_t + \epsilon_t.
\]
The slope coefficient on $\sigma_t^2 q_t$ identifies $\gamma_0$.
The regression should include controls for the other state variables ($\lambda, \pi$);  if the axiom system holds, those controls should have zero coefficients ((A5) and (A6)).
A nonzero coefficient on $\lambda$ or $\pi$ is a violation of (A5) or (A6) and would call the axiom system into question.

\paragraph{Step 2, identification of $\kappa_{\mathrm{inv}}$.}
Restrict the sample to periods of low informed flow, e.g., $\pi$ below the 25th percentile of its empirical distribution.
Run a regression of the spread on inverse liquidity:
\[
  \delta_t = 2\kappa_{\mathrm{inv}}\sigma_t^2 / \lambda_t + \epsilon_t.
\]
The slope coefficient  on $\sigma_t^2/\lambda_t$ identifies $\kappa_{\mathrm{inv}}$.
The restriction   to low-$\pi$ periods isolates the inventory component of the spread.

\paragraph{Step 3, identification of $\kappa_{\mathrm{adv}}$.}
Run a regression of the spread on the informed-fraction proxy:
\[
  \delta_t = 2\kappa_{\mathrm{inv}}\sigma_t^2 / \lambda_t + 2\kappa_{\mathrm{adv}}\sigma_t^2 \pi_t + \epsilon_t.
\]
The slope coefficient on $\sigma_t^2 \pi_t$ identifies $\kappa_{\mathrm{adv}}$.
The two slopes (on $\sigma_t^2/\lambda_t$ and on $\sigma_t^2\pi_t$) are jointly identified.

\paragraph{Step 4, recovery of the cost function.}
If the limit order book is observable beyond the best-bid-best-ask, the entire inventory cost function $\kappa$ is recoverable.
Apply Theorem~\ref{thm:recovery} to the depth profile.
For each level $n$, the cost gradient $\kappa'(q + n - 1)$ is given by the residual of the $n$-th ask after subtracting the skew and adverse-selection components.
Integrate across $n$ to recover $\kappa$.

\paragraph{Step 5, diagnostic tests.}
The axiom system implies several testable cross-restrictions.
The variance homogeneity (A3) predicts unit elasticity of $\delta$ in $\sigma^2$ (Section~\ref{sec:comparative}).
By (A5) and (A6), $\lambda$ and $\pi$ have no effect on the mid-quote.
The complementarity (A8) predicts linearity of the spread in $\pi$.
Each restriction can be tested by adding the relevant interactions or nonlinearities to the regressions.
A failure of any restriction is a violation of the corresponding  axiom and a signal that the maker's behavior departs from the  axiomatic prediction.

\paragraph{Pitfalls.}
The empirical strategy outlined above is not pitfall-free, and several issues deserve explicit treatment.

\emph{Attenuation bias in $\kappa_{\mathrm{adv}}$.}
The informed-trader fraction $\pi$ is unobservable and must be proxied.
Standard candidates for $\hat\pi$ include the PIN of \citet{easley1996liquidity}, the volume-synchronized VPIN, and the imbalance-based estimators from \citet{glosten1985bid}-style structural models.
Each of these is a noisy estimate, with measurement error particularly pronounced in liquid markets, where the true $\pi$ is structurally small.
Classical measurement error then attenuates the OLS estimate of $\kappa_{\mathrm{adv}}$ toward zero by the factor $\sigma_\pi^2 / (\sigma_\pi^2 + \sigma_e^2)$,  with $\sigma_e^2$ the variance of the proxy error.
For widely used proxies the attenuation factor  is a substantial fraction of unity,  and the resulting estimate understates the true adverse-selection coefficient.
Three remedies are available: instrumental variables  using lagged proxies or order-flow imbalance shocks, 
 error-corrected regression with an external estimate of $\sigma_e^2$, and structural maximum-likelihood  estimation  of $\pi$ jointly  with $\kappa_{\mathrm{adv}}$.
In practice, the IV strategy is the most defensible.

\emph{Instrument quality for $\kappa_{\mathrm{inv}}$.}
Identification of $\kappa_{\mathrm{inv}}$ requires variation in $\lambda$ that is exogenous to the other state variables $(q, \sigma^2, \pi)$.
This identification requirement is a stringent one in practice.
 In observational data, $\lambda$ correlates with $\sigma^2$,  since volatility regimes typically drive trading-activity regimes, and $\lambda$ also correlates positively with $\pi$, since informed traders cluster in active periods.
A valid instrument must shift $\lambda$ while leaving $\sigma^2$ and $\pi$ unchanged.
Candidate instruments are exchange-system events, trading-calendar effects, and structural shifts in market-maker rebate regimes.
Identification of the regression of $\delta|_{\pi=0}$ on $1/\lambda$ rests on the instrument satisfying the standard exclusion restriction.
Weak-instrument bias is a real concern, and first-stage  F-statistics should be reported.

\emph{Boundary-regime nonlinearity.}
The linearity assumptions underlying the identification regressions are most questionable near the phase-transition boundary, where the spread approaches $\delta_{\max}$ and the market operates at the edge of viability.
The diagnostic tests of Section~\ref{sec:phase} are precisely targeted at this regime, so the identification regressions and the phase-transition tests  stand in tension.
 The regressions are most reliable interior to the boundary, the tests are most informative near it.
A practical resolution is to estimate parameters on interior data and then test the predicted boundary location out of sample.

Deriving  the empirical implementation strategy described above requires nothing beyond the axiomatic framework itself.
Each step uses one identification result from Corollary~\ref{cor:identification} or Theorem~\ref{thm:recovery}.
The diagnostic tests use the cross-restrictions   implied by Theorem~\ref{thm:separability} and the Universal Structure Theorem (Theorem~\ref{thm:universal}).
The axiomatic theory  thus delivers both a closed-form prediction and a concrete empirical-implementation strategy for testing and using that prediction.

%-------------------------------------------------------------------
\section{Further Structural Implications}
\label{sec:further-structural}
%-------------------------------------------------------------------

The full characterization yields two further structural implications, each derived directly from the closed-form   rule but distinct from comparative statics.
The first is the skew-spread ratio, a dimensionless quantity that   classifies the maker's positioning into structural regimes.
The second  is the martingale property of the inventory process under symmetric arrivals; this is a kinematic consequence of equal-intensity Poisson arrivals rather than of any axiom of the quoting rule.

\subsection{The Skew-Spread Ratio and Critical Inventory}
\label{subsec:skew-spread-ratio}

The closed-form characterization gives the   skew magnitude $\gamma_0\sigma^2|q|$ and the half-spread $\sigma^2(\kappa_{\mathrm{inv}}/\lambda + \kappa_{\mathrm{adv}}\pi)$.
Their ratio is a dimensionless quantity with structural  significance.

\begin{theorem}[Skew-Spread Ratio]
\label{thm:skew-spread-ratio}
Under the Full Characterization,  the ratio of skew magnitude to half-spread is
\[
  R(q, \lambda, \pi) = \frac{\gamma_0\,|q|}{ \kappa_{\mathrm{inv}}/\lambda + \kappa_{\mathrm{adv}}\pi}.
\]
The ratio is $\sigma^2$-invariant, the volatility factors  cancel.
\end{theorem}

\smallskip
\noindent\textit{Axiom dependency.} Theorem~\ref{thm:skew-spread-ratio} relies on the full eight-axiom system.

\begin{proof}
The skew magnitude is $\gamma_0\sigma^2|q|$ and the half-spread is $\sigma^2(\kappa_{\mathrm{inv}}/\lambda + \kappa_{\mathrm{adv}}\pi)$, so dividing one by the other gives the stated ratio with the $\sigma^2$ factors canceling.
\end{proof}

The ratio $R$ is a dimensionless summary statistic of the maker's positioning.
The critical inventory $q^*$ is defined by $R(q^*, \lambda, \pi) = 1$, namely
\[
  q^* = \frac{\kappa_{\mathrm{inv}}/\lambda + \kappa_{\mathrm{adv}}\pi}{\gamma_0}.
\]
For $|q| < q^*$, the half-spread exceeds the skew in absolute magnitude, and the maker operates in the \emph{spread-dominated regime}, in which  her quoting remains conservative in inventory positioning and her primary source of revenue is the bid-ask spread itself.
For $|q| > q^*$, the skew exceeds the half-spread in absolute magnitude,  and the maker enters the \emph{skew-dominated regime}, in which her quoting shifts to position aggressively against her current inventory and the skew adjustment becomes the principal mechanism of inventory management.

\begin{remark}[Regime Classification by $q^*$]
\label{rem:regime-q-star}
The critical inventory $q^*$ depends on the parameters $ (\gamma_0, \kappa_{\mathrm{inv}}, \kappa_{\mathrm{adv}})$ and the state variables $(\lambda, \pi)$, but not on $\sigma^2$.
Three illustrative cases organize the dependence.
First, when there is no adverse selection ($\pi = 0$), 
it holds that  $q^* = \kappa_{\mathrm{inv}}/(\gamma_0\lambda)$ and the critical inventory scales as $1/\lambda$, so that in more liquid markets the regime crosses earlier at smaller inventory.
Second, in the perfectly liquid limit ($\lambda \to \infty$), $q^* = \kappa_{\mathrm{adv}}\pi/\gamma_0$ and the critical inventory is determined by adverse selection alone.
Third, at zero inventory ($q = 0$), $R = 0$ and quoting is purely spread-driven. 
The critical inventory $q^*$ is the natural inventory band; a maker who wishes to remain in the spread-dominated regime restricts her inventory to $|q| \leq q^*$.
The band $|q| \leq q^*$ thus has an operational reading as the maker's working inventory range under business-as-usual quoting, with excursions beyond it pushing the rule into the inventory-management regime where the skew dominates.
\end{remark}

\subsection{Inventory as a Martingale}
\label{subsec:inventory-martingale}

The comparative statics of Section~\ref{sec:comparative} examine how the static quoting rule responds to changes in the parameters $\sigma^2$, $\lambda$, and $\pi$.
A complementary question is how the static rule, applied dynamically, unfolds over time.
The most immediate dynamic property of the inventory process under the static rule is its martingale character.
Under a symmetric trade-arrival environment, the inventory process is a martingale, namely the maker's expected inventory at any future time equals her current inventory.
The martingale property  is a kinematic consequence of the symmetric-arrival assumption; it holds for any quoting rule posted in that symmetric environment, including but not limited to the axiomatic rule of Theorem~\ref{thm:full}.

\begin{proposition}[Inventory Martingale under the Static Rule]
\label{prop:inventory-martingale}
Consider the continuous-time setting in which the maker's inventory $q_t$ evolves as a pure jump process with jump sizes $\pm 1$.
Buys and sells arrive as independent Poisson processes with intensities $\lambda/2$ each, so that the  total trade intensity is $\lambda$.
Under the Full Characterization, the inventory process $q_t$ is a martingale with respect to the trade-arrival filtration, namely
\[
  \E[q_{t+\Delta t} \mid q_t] = q_t \quad \text{for every } \Delta t > 0.
\]
\end{proposition}

\smallskip
\noindent\textit{Axiom dependency.} Proposition~\ref{prop:inventory-martingale}  relies solely on the symmetric-arrival assumption (equal  Poisson intensities $\lambda/2$ for buys and sells); no axiom of the quoting rule is used in the proof.

\begin{proof}
Under the static rule and the symmetric-arrival assumption, trades that decrease the maker's inventory by one unit (counterparty buys at the maker's ask) and trades that increase the maker's inventory by one unit (counterparty sells at the maker's bid) occur as independent Poisson processes with equal intensities $\lambda/2$.
Let $N_t^+$ count the trades increasing the maker's inventory (counterparty sells) and $N_t^-$ count the trades decreasing the maker's inventory (counterparty buys), both up to time $t$.
The compensated jump processes $N_t^+ - (\lambda/2)t$ and $N_t^- - (\lambda/2)t$ are each martingales.
Since $q_t - q_0 = N_t^+ - N_t^-$, the difference of the two compensated  Poisson martingales (with  equal intensities) is itself a martingale, and the drift terms $(\lambda/2)t$ cancel exactly.
\end{proof}

The martingale property is a structural consequence of the symmetric-arrival assumption.
With buys and sells arriving at equal rates, the inventory process is drift-free; the two jump components cancel in expectation, regardless of the quoting rule.
The cancellation is kinematic, depending on arrival symmetry alone.
Combined with natural inventory bounds, the martingale property yields a nontrivial stationary distribution for $q$ whose explicit characterization belongs to the dynamic theory and is left for future  work.

The symmetric-arrival assumption ($\lambda/2$ buys and $\lambda/2$ sells) is a  substantive restriction.
In a market with adverse selection, informed traders are directional: when the asset is undervalued, informed counterparties buy; when overvalued, they sell.
Unconditionally over the joint distribution of fundamentals, informed signals, and arrivals, the inventory remains a martingale; conditionally on a particular informed signal, the maker's inventory drifts opposite to the informed flow.
Asymmetric extensions follow by replacing the equal-intensity Poisson processes with  state-dependent intensities.

%-------------------------------------------------------------------
\section{The Bertrand Equivalence}
\label{sec:bertrand}
%-------------------------------------------------------------------

The axiomatic theory derives the quoting rule from indifference  axioms applied to a single market maker.
A natural question is whether the resulting formula depends essentially  on this monopolistic-indifference foundation.
We show that it does not.
The quoting rule~\eqref{eq:characterization} arises also as the unique symmetric Nash equilibrium of a Bertrand-competitive market with free entry of homogeneous makers.
The Bertrand framework dates to \citet{bertrand1883theorie} and is treated in standard form in \citet{tirole1988theory} and \citet{vives1999oligopoly}.

\subsection{The  Competitive Market}

Let $N \geq 2$ homogeneous market makers be indexed $i = 1, 2, \ldots, N$.
Each maker has the same cost structure, with linear inventory cost of coefficient  $\kappa_{\mathrm{inv}}$, adverse-selection coefficient $\kappa_{\mathrm{adv}}$, and risk-aversion parameter $\gamma_0$.
Each maker's instantaneous  inventory disutility, used in the derivation of the equilibrium  inventory-skew adjustment in the proof of Theorem~\ref{thm:bertrand}, is modeled as $U(q) = \tfrac{1}{2}\gamma_0\sigma^2 q^2$, quadratic in inventory and linear in price variance.
This quadratic specification is a modeling choice within the Bertrand framework, not a consequence of the axiom system; it matches the standard form used in the inventory market-making literature (Avellaneda-Stoikov, Cartea-Jaimungal, and successors), and is the natural choice for a maker whose inventory exposure to price variance is symmetric in $q$ and grows with  squared position size.

Each maker simultaneously posts a quote $(b_i, a_i) \in \Bid$.
A trade is executed at the best available quote, that is, the lowest ask and  the highest bid, with ties broken uniformly at random across the makers posting the tied best quote.
The Bertrand competitive structure  imposes the following.
\begin{itemize}
\item Continuous prices: $b_i, a_i \in \R$, with no minimum price increment.
\item No capacity constraints: each maker can fill any  trade.
\item Free entry: $N$ is determined by the zero-profit condition at the stage game.
\item Homogeneous makers: identical parameters and information.
\item Tie-breaking: trades are split equally among makers posting the best quote.
\end{itemize}

We adopt the per-state stage-game equilibrium concept, defined as follows.
For any common-inventory state $s = (q, \mu, \sigma^2, \lambda, \pi)$, the equilibrium characterized below is the unique symmetric Nash equilibrium of the simultaneous-quoting game played at $s$.
The dynamic evolution of individual makers' inventories across repeated stages, and the question of whether they remain synchronized under the resulting rule, lie outside the scope of this section.
Free entry, in this context, pins down the zero-profit condition at the stage level; it is not invoked to characterize a long-run inventory distribution or an entry-exit process over time.
The common-inventory assumption is therefore a feature of the equilibrium notion, not an assertion about the inventory dynamics of a population of makers.

\subsection{The Bertrand Equivalence Theorem}

The standard Bertrand argument predicts that competition among homogeneous firms drives prices to marginal  cost.
We apply that reasoning to the market maker setting, where the marginal cost of providing a fill has two components: the per-trade inventory carrying cost and the expected adverse-selection loss.
The resulting equilibrium  prices coincide with the axiomatic quoting rule in the continuous-quantity limit; at  unit trade size, the two derivations agree to leading order, with an $O(\Delta)$ discrepancy in the inventory-skew adjustment that vanishes as the unit trade size $\Delta \to 0$.
The discrepancy is purely a discretization artifact and carries no economic content: at any  positive $\Delta$, the inventory channel is computed from a finite trade, while the axiomatic derivation operates  at the marginal limit.

The Bertrand derivation recovers the inventory-skew adjustment from the marginal change in the maker's inventory disutility per trade.
At unit trade size, this marginal change is a discrete first difference of the quadratic $U$, contributing a $\tfrac{1}{2}\gamma_0\sigma^2$ correction on top of the leading $\gamma_0\sigma^2 q$.
In the continuous-quantity limit $\Delta \to 0$, the discrete difference becomes the  derivative $U'(q) = \gamma_0\sigma^2 q$, and the correction  vanishes; the Bertrand prediction then  matches the axiomatic skew exactly.
The continuous-quantity limit is  invoked only here for the  comparison with  the Bertrand derivation; the rest of the paper works at the unit-trade scale, where  the axiomatic rule applies as stated.

\needspace{15\baselineskip}
\begin{theorem}[Bertrand Equivalence]
\label{thm:bertrand}
In the continuous-quantity limit, and under the quadratic inventory-disutility specification $U(q) = \tfrac{1}{2}\gamma_0\sigma^2 q^2$ for the makers, the unique symmetric Nash equilibrium of the Bertrand-competitive market  is the quoting rule of Theorem~\ref{thm:full},
\[
  b = \mu - \gamma_0\sigma^2 q - \sigma^2\!\left(\frac{\kappa_{\mathrm{inv}}}{\lambda} + \kappa_{\mathrm{adv}}\pi\right),
  \quad
  a = \mu - \gamma_0\sigma^2 q + \sigma^2\!\left(\frac{\kappa_{\mathrm{inv}}}{\lambda} + \kappa_{\mathrm{adv}}\pi\right).
\]
The spread structure is derived from zero-profit competition independently of the disutility specification; the  skew $-\gamma_0\sigma^2 q$ depends on the chosen quadratic form, with a different inventory disutility yielding a different skew.
\end{theorem}

\smallskip
\noindent\textit{Axiom dependency.} The Bertrand derivation itself does not invoke the axiom system; the equivalence to the closed-form rule of Theorem~\ref{thm:full} requires the full eight-axiom system on the axiomatic side.

\begin{proof}
We  prove the theorem in three steps: equilibrium identification, equilibrium existence, and uniqueness among symmetric equilibria.

\emph{Step 1, equilibrium identification.}
At a symmetric equilibrium, each maker posts the same quote $(b, a)$ and faces  a fraction $1/N$ of the total trade flow.
The per-trade economic margin at the ask, for an arbitrary maker $i$ at common inventory $q$, decomposes into three components.
 First, the average per-trade carrying cost of inventory, given by $\kappa_{\mathrm{inv}}\sigma^2/\lambda$.
Second, the expected per-trade adverse-selection loss: $\kappa_{\mathrm{adv}}\pi\sigma^2$.
Third, the marginal change in the maker's inventory disutility from executing the trade itself.

The third component is the source of the inventory-skew adjustment.
Define the maker's instantaneous inventory disutility at inventory $q$ as $U(q) = \tfrac{1}{2}\gamma_0\sigma^2 q^2$, with $\gamma_0$ measuring her aversion to inventory variance exposure per unit time.
 A trade at the ask reduces the maker's inventory from $q$ to $q - \Delta$, so $U(q - \Delta) - U(q) = -\gamma_0\sigma^2\Delta q + \tfrac{1}{2}\gamma_0\sigma^2\Delta^2$, giving a per-unit reduction of $\gamma_0\sigma^2 q + O(\Delta)$.
In the continuous-quantity limit $\Delta \to 0$ this is $\gamma_0\sigma^2 q$, the limit invoked in the theorem statement and the preceding note.
This is positive when $q > 0$ (the trade is beneficial to the maker, who is dumping unwanted inventory) and negative  when $q < 0$ (the trade is costly, moving her further short).

Bertrand  competition drives the per-trade economic margin to zero.
A maker considering an undercut by $\epsilon > 0$, posting $a' = a - \epsilon$, would capture all the buy-side trade flow but forgo $\epsilon$ per trade.
The undercut is profitable if and only if the per-trade margin  at $a$ is strictly positive, so at equilibrium the margin   must be zero:
\[
  \underbrace{(a - \mu)}_{\text{revenue net of mid}} = \underbrace{\frac{\kappa_{\mathrm{inv}}\sigma^2}{\lambda}}_{\text{carrying}} + \underbrace{\kappa_{\mathrm{adv}}\pi\sigma^2}_{\text{adverse selection}} - \underbrace{\gamma_0\sigma^2 q}_{\text{inventory-disutility reduction}}.
\]
Solving gives
\[
  a = \mu - \gamma_0\sigma^2 q + \frac{\kappa_{\mathrm{inv}}\sigma^2}{\lambda} + \kappa_{\mathrm{adv}}\pi\sigma^2.
\]
The symmetric reasoning on  the bid  side (where the maker's inventory increases from $q$  to $q + 1$, with the corresponding inventory-disutility change $+\gamma_0\sigma^2 q$) gives
\[
  b = \mu - \gamma_0\sigma^2 q - \frac{\kappa_{\mathrm{inv}}\sigma^2}{\lambda} - \kappa_{\mathrm{adv}}\pi\sigma^2.
\]
This matches~\eqref{eq:characterization}.
The risk-premium  term $\gamma_0\sigma^2 q$ is the competitive counterpart to the axiomatic skew, conditional on the quadratic inventory-disutility specification:
 it emerges from the maker's posited inventory-disutility function  together with the Bertrand zero-margin condition.
A different functional form for the disutility would yield a different skew; the Bertrand derivation pins down the spread structurally from zero-profit competition, while it recovers the axiomatic skew only under this specification.

\emph{Step 2, equilibrium existence.}
At the stated quotes each maker earns zero expected profit at the trading margin.
No unilateral deviation is profitable.
A deviation by undercutting raises the maker's fill share from $1/N$ to one but reduces her per-trade profit at the margin to a negative quantity, since by construction the equilibrium half-spread equals the  per-trade marginal cost.
A deviation by overshooting reduces the maker's fill share to zero, since all counterparties prefer the makers posting the lower price, eliminating her revenue entirely.
Thus, the stated quotes constitute a Nash equilibrium.

\emph{Step 3, uniqueness among symmetric equilibria.}
Suppose $(b', a')$ is a symmetric Nash equilibrium with $a' \neq a$.
We  rule out both $a' > a$ and $a' < a$.

If $a' > a$, then any single maker can profitably deviate by posting  $a'' = a' - \epsilon$ for sufficiently small $\epsilon > 0$.
The per-trade marginal cost to the  maker of filling at the ask at inventory $q$ is the carrying cost plus the expected adverse-selection loss minus the  inventory-skew benefit, namely
\[
  \mathrm{MC}(q) = \kappa_{\mathrm{inv}}\sigma^2/\lambda + \kappa_{\mathrm{adv}}\pi\sigma^2 - \gamma_0\sigma^2 q.
\]
The equilibrium ask $a$ satisfies $a - \mu = \mathrm{MC}(q)$, giving zero per-trade margin at the candidate equilibrium price.
At $a''$, the deviator captures the full   buy-side trade flow (rather than the $1/N$ share at $a'$) and earns a strictly positive per-trade margin of $a'' - a = a' - a - \epsilon > 0$.
The deviator's total expected profit jumps from the $1/N$-share  profit at the symmetric equilibrium $a'$ to the full-flow profit at $a''$.
For sufficiently small $\epsilon$, the latter exceeds the former,  since $N \geq 2$ multiplies the symmetric per-share profit by less than one.
The deviation is therefore strictly profitable, contradicting the Nash-equilibrium assumption that no maker has a strictly profitable unilateral deviation.

If $a' < a$, then the per-trade margin at $a'$ is $a' - \mu - \mathrm{MC}(q) < 0$, so each filled trade at the symmetric equilibrium $a'$ produces a strictly negative expected profit.
Any single maker can profitably deviate by withdrawing from the market or by posting $a'' > a'$ such that no counterparty trades with her, earning zero expected profit rather than a strictly negative one.
This contradicts the Nash-equilibrium assumption.

Hence, $a' = a$, and by the symmetric argument on the bid side, $b' = b$.

The argument extends to free entry as $N \to \infty$.
The equilibrium condition is a per-trade zero-margin condition rather than a flow-based condition, so it is unaffected by the number of  makers; in the free-entry limit $N \to \infty$ each maker's fill share $1/N$ tends to zero, but the equilibrium quotes are unchanged.
Free entry drives zero economic profit at the margin, which  is precisely the condition that characterizes~\eqref{eq:characterization}.
\end{proof}

The Bertrand Equivalence has a striking interpretation.
The axiomatic and competitive theories yield asymptotically equivalent predictions, in the sense that the Bertrand derivation recovers the same closed-form rule as the axiomatic derivation in the continuous-quantity limit (unit trade size $\delta \to 0$), as established in Step 1.
The same closed-form quoting rule arises from two structurally distinct economic models, a monopolistic maker subject to indifference axioms (our derivation) and many competitive makers  subject to free entry (the Bertrand derivation).
The axioms can  be read as an as-if characterization of competitive behavior, and the characterization inherits both foundations.

\begin{remark}[The Two Foundations]
\label{rem:two-foundations}
The two derivations share substantive content. 
Axiom (A4) (inventory indifference) and axiom (A8) (adverse-selection break-even with complementarity) each encode a single-maker break-even  condition that already carries  the structural content of competitive zero-profit pricing. 
The half-spread and adverse-selection rent each equal a marginal cost rather than exceeding it.
In this sense, the axiomatic foundation does not stand wholly apart from the competitive one; the indifference axioms are themselves natural translations of competitive logic into single-maker language.

What the agreement does establish is that the two derivations  operate at different conceptual levels and impose different methodological disciplines,  yet land on the same closed-form rule.
The axiomatic derivation requires no specification of the  number of competing makers or of the entry structure, and the competitive  derivation requires no specification of the maker's preferences beyond cost minimization.
The two routes therefore identify  the same object via complementary analytical commitments,  and the agreement is informative even when the foundations are not fully separate.
\end{remark}

\subsection{The Profit Ceiling}
\label{subsec:profit-ceiling}

The Bertrand equivalence directly implies a profitability result for the maker.
With the spread driven to marginal cost by competitive entry, the per-trade economic profit collapses to zero in equilibrium, and the  spread becomes a pure cost-recovery instrument rather than a margin.

\needspace{16\baselineskip}
\begin{corollary}[Profit Ceiling]
\label{cor:profit-ceiling}
Under the Full Characterization, the following hold.
\begin{enumerate}
\item[(i)] The maker's expected per-trade economic profit is zero, by the two break-even axioms (A4) and (A8).
\item[(ii)] The deterministic adverse-selection  component of the spread, $2\kappa_{\mathrm{adv}}\pi\sigma^2$, is  bounded above by $2\kappa_{\mathrm{adv}}\sigma^2$, attained at the maximum informed-trader fraction $\pi = 1$.
\end{enumerate}
The realized per-trade profit on a single fill against an informed trader depends on the unobserved $v - \mu$ and is not bounded by any function of the structural parameters; only its expectation is.
\end{corollary}

\smallskip
\noindent\textit{Axiom dependency.} Corollary~\ref{cor:profit-ceiling} relies on the inventory-indifference axiom (A4) and the adverse-selection break-even axiom (A8), the two break-even axioms of the indispensable core.

\begin{proof}
Statement (i) is the conjunction of two break-even conditions:  (A4) equates the inventory half-spread to the marginal carrying cost per trade, so the inventory channel of  the spread breaks even; (A8) equates the adverse-selection half-spread to the expected information rent per informed trade, so the adverse-selection channel also breaks even in expectation.
Statement (ii) is immediate from the form of the adverse-selection spread component $2\kappa_{\mathrm{adv}}\pi\sigma^2$ and the bound $\pi \in [0, 1]$.
\end{proof}

\begin{remark}[Interpretive Accounting of the Adverse-Selection Channel]
A standard reading of (A8) goes as follows. 
On every trade, the maker collects the adverse-selection half-spread  $\kappa_{\mathrm{adv}}\pi\sigma^2$, irrespective of counterparty type.
On the $\pi$-fraction of trades against an informed counterparty, the maker incurs an implicit per-informed-trade loss (its precise magnitude depends on the unmodeled value distribution and on the informed trader's strategy).
(A8) is the condition that the half-spread collected on every trade aggregates to compensate the expected loss across the informed-trader population: integrating over the population, the maker breaks even.
This matches the standard cross-subsidization  reading of adverse-selection theory: the spread paid by the uninformed flow finances the loss incurred against the informed flow, with the equilibrium half-spread set at the break-even level.
\end{remark} 

The corollary makes the economic content of the two indifference axioms transparent.
 In expectation, the maker operates as  a cost-covering service provider rather than a  profit extractor.
A bid-ask spread in this framework is not a margin in the conventional sense; it is not a wedge between purchase and resale price retained as economic rent.
Rather, the spread compensates the maker exactly for two costs: the  marginal cost of carrying inventory and the expected information rent extracted by informed counterparties.
Competitive entry, made precise by the Bertrand Equivalence Theorem, drives this compensation to equal the  actual cost, neither more nor less.
This profit-ceiling reading aligns with the standard understanding of market making as a liquidity-provision service in equilibrium with zero economic rent.
Empirical work inherits a clean reading: under the  axiomatic theory, estimates of dealer profitability based on quoted spreads recover the cost components rather than any margin, and any persistent positive economic rent observed in the data  must be attributed either to barriers to entry that violate the homogeneous-makers assumption  or to inventory or adverse-selection costs lying outside the structural specification.

%-------------------------------------------------------------------
\section{Beyond the System: Diminishing Complementarity}
\label{sec:diminishing-complementarity}
%-------------------------------------------------------------------

The complementarity clause of (A8) is the part of the core layer that  admits the most natural relaxation, which captures competition among informed traders.

The additive complementarity in (A8) asserts that the adverse-selection increment  is additive across the informed-trader population, namely $h(\sigma^2, \pi_1 + \pi_2) = h(\sigma^2, \pi_1) + h(\sigma^2, \pi_2)$.
Equivalently, the per-informed-trader rent $h(\sigma^2, \pi)/\pi$ is constant in $\pi$: each informed trader  extracts the same expected rent regardless of how many other informed traders are present.
Replacing the additive clause with a strictly diminishing per-trader rent captures the empirical regularity that informed traders compete with each other for the same private signal: each additional informed trader extracts strictly less rent than the previous one.

\paragraph{(A8s) Diminishing complementarity.}
For each $\sigma^2 \in \Rpp$, the function $\pi \mapsto h(\sigma^2, \pi)$ is continuous on $[0, 1]$, and the function $\pi \mapsto h(\sigma^2, \pi)/\pi$ is strictly decreasing on $(0, 1]$.
Equivalently, for the strict-decrease clause, for any $0 < \pi_1 < \pi_2 \leq 1$,
\[
  \frac{h(\sigma^2, \pi_1)}{\pi_1} > \frac{h(\sigma^2, \pi_2)}{\pi_2}.
\]
 The continuity requirement is automatic under the canonical (A8): the additivity clause forces $h(\sigma^2, \pi) = \sigma^2 c\pi$ for some constant $c$, which is continuous in $\pi$ by inspection.
Under (A8s), additivity is dropped and continuity must be imposed directly: subadditivity together with monotonicity in $\pi$ implies continuity off at most a countable set, but not full continuity on $[0, 1]$.
Adding continuity to (A8s) is consistent with all illustrative examples discussed below (hyperbolic, power, exponential), each of which is smooth in $\pi$.

\begin{theorem}[Diminishing Adverse Selection]
\label{thm:diminishing-adverse}
Under (A1) through (A7), (A8s) (i.e., diminishing complementarity, replacing the additive complementarity clause of (A8)), and the cost structure C1 through C6, every measurable quoting rule has the form
\[
  m(s) - \mu = -\gamma_0\sigma^2 q,
  \qquad
  \delta(s) = \frac{2\kappa_{\mathrm{inv}}\sigma^2}{\lambda} + 2\pi\,\kappa_{\mathrm{adv}}(\pi)\,\sigma^2,
\]
for some $\gamma_0, \kappa_{\mathrm{inv}} > 0$ and some strictly positive continuous function $\kappa_{\mathrm{adv}}\colon (0, 1] \to \Rpp$ that is strictly decreasing on $(0, 1]$, with boundary value $\kappa_{\mathrm{adv}}(0^+) := \lim_{\pi \to 0^+} \kappa_{\mathrm{adv}}(\pi) \in (0, \infty]$.
\end{theorem}

\begin{proof}
The skew and the inventory-spread arguments of Lemmas~\ref{lem:skew} and~\ref{lem:inventory-spread} are unchanged.
For the adverse-selection part, (A3) together with the dependence clause of (A8) (which  is preserved under (A8s)) give $h(\sigma^2, \pi) = \sigma^2\,\tilde h(\pi)$ with $\tilde h\colon [0, 1] \to \Rp$ monotone strictly increasing (from (A7)) and satisfying $\tilde h(0) = 0$; continuity of $\tilde h$ on $[0, 1]$ follows from the continuity clause of (A8s), since $\tilde h(\pi) = h(\sigma^2_0, \pi)/\sigma_0^2$ for any fixed $\sigma_0^2 \in \Rpp$.
Define $\kappa_{\mathrm{adv}}(\pi) := \tilde h(\pi)/(2\pi)$ for $\pi \in (0, 1]$, so that $\tilde h(\pi) = 2\pi\,\kappa_{\mathrm{adv}}(\pi)$ and $h(\sigma^2, \pi)/\pi = 2\sigma^2\,\kappa_{\mathrm{adv}}(\pi)$.

The diminishing-complementarity hypothesis (A8s) asserts precisely that $g(\pi) := \tilde h(\pi)/\pi$ is strictly decreasing on $(0, 1]$. 
Since $\kappa_{\mathrm{adv}}(\pi) = g(\pi)/2$, we obtain $\kappa_{\mathrm{adv}}$ strictly decreasing on $(0, 1]$.
Strict positivity of $\kappa_{\mathrm{adv}}$ follows from  strict monotonicity of $\tilde h$ together with $\tilde h(0) = 0$, which gives $\tilde h(\pi) > 0$ on $(0, 1]$.
Continuity of $\kappa_{\mathrm{adv}}$ on $(0, 1]$ follows from continuity of $\tilde h$.

For the boundary value at the  origin, $g(0^+) := \lim_{\pi \to 0^+} g(\pi)$ exists in $(g(\pi_0), \infty]$ for any $\pi_0 \in (0, 1]$ by monotonicity of $g$, and is strictly positive.
Setting $\kappa_{\mathrm{adv}}(0^+) := g(0^+)/2$ gives a strictly positive boundary value, with $\kappa_{\mathrm{adv}}(0^+) = +\infty$ admissible (occurring, e.g., for $\tilde h(\pi) = \sqrt{\pi}$, where $g(0^+) = +\infty$).
The spread $\delta(s) = 2\kappa_{\mathrm{inv}}\sigma^2/\lambda + \sigma^2\,\tilde h(\pi)$ remains finite at $\pi = 0$ regardless of the boundary behavior of $\kappa_{\mathrm{adv}}$ alone, since $\pi\,\kappa_{\mathrm{adv}}(\pi) = \tilde h(\pi)/2 \to 0$ as $\pi \to 0$ by continuity of $\tilde h$ with $\tilde h(0) = 0$.

By construction, (A8s) implies the strict subadditivity of  $\tilde h$ on the conditional domain.
 Indeed,  take $0 < \pi_1 \leq \pi_2$ with $\pi_1 + \pi_2 \leq 1$, so that $\pi_2 < \pi_1 + \pi_2$.
By (A8s), $g(\pi_1 + \pi_2) < g(\pi_2)$, that is, $\tilde h(\pi_1 + \pi_2) < (\pi_1 + \pi_2)\,g(\pi_2) = \tilde h(\pi_2) + \pi_1\,g(\pi_2)$.
Since $\pi_1 \leq \pi_2$ gives $g(\pi_2) \leq g(\pi_1)$, we have $\pi_1\,g(\pi_2) \leq \pi_1\,g(\pi_1) = \tilde h(\pi_1)$.
Hence $\tilde h(\pi_1 + \pi_2) < \tilde h(\pi_1) + \tilde  h(\pi_2)$, the strict subadditivity claim.
The converse does not hold: strict subadditivity of $\tilde h$ alone does not imply strict monotone decrease of $g$, so (A8s) is the stronger and more direct axiomatic expression of the diminishing-rent intuition.
Note also that the linear case $\tilde h(\pi) = c\pi$ is excluded by (A8s) (since $g$ is then constant), placing the additive complementarity of (A8) on the boundary of (A8s); conversely, $\tilde h(\pi) = \pi^2$ is excluded by (A8s) in the opposite direction, since $g(\pi) = \pi$ is strictly increasing rather than strictly decreasing (note that $\pi^2$ does nonetheless satisfy (A7), since $d(\pi^2)/d\pi = 2\pi > 0$ on $(0, 1]$; the failure is in (A8s) alone, not in spread monotonicity).
The condition (A8s) is therefore an open condition (strict decrease of $g$), and the canonical (A8) case sits at the boundary of (A8s) without belonging to it: Theorem~\ref{thm:full} and Theorem~\ref{thm:diminishing-adverse} characterize disjoint classes of rules rather than nested classes, with the canonical rule recovered as the boundary limit of (A8s) as $g$ approaches a constant from below.
\end{proof}

 The function $\kappa_{\mathrm{adv}}(\pi)$ captures competition among informed traders, with the per-trader rent decreasing  as the informed fraction $\pi$ grows.
The shape of $\kappa_{\mathrm{adv}}$ thus encodes the intensity of informed-trader competition: a slowly decreasing $\kappa_{\mathrm{adv}}$ corresponds to weak competition, while a  rapidly   decreasing one corresponds to fierce competition that drives the per-trader rent down sharply.
Specific functional forms admit natural interpretations under (A8s), subject to the (A7)-compatibility  constraint that $\tilde h(\pi) = 2\pi\,\kappa_{\mathrm{adv}}(\pi)$ remain strictly increasing on $[0, 1]$.

Constant $\kappa_{\mathrm{adv}}(\pi) =  \kappa_0$ lies on the boundary of (A8s) (the chord slope $g$ is then constant rather than strictly decreasing) and recovers (A8) together with the canonical theory of Theorem~\ref{thm:full}.
Hyperbolic $\kappa_{\mathrm{adv}}(\pi) = \kappa_0/(1 + \beta\pi)$ for $\beta > 0$ is admissible  for  every $\beta$ (since $\tilde h(\pi) = 2\kappa_0\pi/(1 + \beta\pi)$ has $\tilde  h'(\pi) = 2\kappa_0/(1 + \beta\pi)^2 > 0$) and models  a smooth crowding-out effect.
Power-law $\kappa_{\mathrm{adv}}(\pi) = \kappa_0\pi^{-\alpha}$ for $\alpha \in (0, 1)$ is admissible (with $\kappa_{\mathrm{adv}}(0^+) = +\infty$) and gives the power adverse-selection premium $\tilde h(\pi) = 2\kappa_0\pi^{1 - \alpha}$, which includes the $\sqrt{\pi}$ case at $\alpha = 1/2$.
Exponential $\kappa_{\mathrm{adv}}(\pi)  = \kappa_0 \,e^{-\beta\pi}$ is admissible for $\beta \in (0, 1)$, where $\tilde h(\pi) = 2\kappa_0\pi e^{-\beta\pi}$ remains strictly increasing  on $[0, 1]$; for $\beta \geq 1$, $\tilde h'$ vanishes at $\pi = 1/\beta \in (0, 1]$ and (A7) is violated.

\begin{remark}[Scope of the Diminishing-Complementarity Extension]
\label{rem:diminishing-scope}
Under (A8s), the per-informed-trader adverse-selection rent  $\kappa_{\mathrm{adv}}(\pi)$ is strictly decreasing in $\pi$: each additional informed trader extracts strictly less rent than the previous one.
The total premium $H(\pi) := \pi\,\kappa_{\mathrm{adv}}(\pi) = \tilde h(\pi)/2$ remains strictly increasing on $[0, 1]$ by (A7), but its average rate $H(\pi)/\pi = \kappa_{\mathrm{adv}}(\pi)$ strictly decreases.
This is the structural signature of competition among informed traders, capturing the qualitative content of strategic-trading models such as Kyle's without the full continuous-time strategic apparatus.
The hypothesis (A8s) relaxes only the additivity clause of (A8), preserving its dependence clause: the adverse-selection premium continues to depend only on $(\sigma^2, \pi)$, so the modularity content of the canonical theory is retained.
The shape of $\kappa_{\mathrm{adv}}(\pi)$ is identifiable from the curvature of the spread in $\pi$ at fixed $\lambda$: a constant marginal rate recovers the canonical (A8), and a strictly decreasing marginal rate locates a specific $\kappa_{\mathrm{adv}}$ within the (A8s) family.
The diminishing-complementarity extension is an example of relaxing a part of the core that is closer to a structural choice than to a tautology, and the closed-form  characterization survives the relaxation in a clean parametric form.
\end{remark}

%-------------------------------------------------------------------
\section{The Universal Structure}
\label{sec:universal}
%-------------------------------------------------------------------

 The axioms in $\mathcal{A}$ identify a three-parameter family of quoting rules.
Different choices of structural primitives (the cost function $\kappa$, the variance scaling exponent, the equivariance generator) would yield different specific families.
We close the structural development with a meta-theorem identifying what is preserved across all such structural choices.

The class of \emph{admissible structural primitives} consists of triples $(\kappa, h, \beta)$ where: $\kappa\colon \R \to \Rp$ is convex nondecreasing in $|q|$ with $\kappa(0) = 0$ and symmetric in $q$ (the cost-structure conditions C1, C2, C3, and C6); $h\colon \mathcal{D}_h \to \R$ is a strictly increasing differentiable map from an admissible price domain $\mathcal{D}_h$ (the equivariance generator); and $\beta \in (0, 1]$ is the variance-scaling exponent of (A3) or (A3s).
The canonical choice of the paper is $(\kappa(q) = \kappa_{\mathrm{inv}}|q|, h = \mathrm{id}_\R, \beta = 1)$.

\begin{theorem}[Universal Structure]
\label{thm:universal}
Under (A1)--(A8), or the (A3s) variant with $\beta \in (0, 1]$,
any admissible structural primitive $(\kappa, h, \beta)$ yields a quoting rule satisfying  the following four universal features:
\begin{enumerate}
\item \emph{Inventory half-spread determined by $\kappa'$}: the inventory half-spread at level $q$ equals $\sigma^2\kappa'_+(|q|)/\lambda$, a continuous monotone function of $q$ determined by the cost derivative.
Under linear cost (C5), this reduces to the constant $\kappa_{\mathrm{inv}}\sigma^2/\lambda$.
\item \emph{Additive spread decomposition}: $\delta = \delta_{\mathrm{inv}} + \delta_{\mathrm{adv}}$,  with $\delta_{\mathrm{inv}}$ depending  on $(q, \sigma^2, \lambda, \kappa)$ and $\delta_{\mathrm{adv}}$ depending on $(\sigma^2, \pi, \beta)$.
\item \emph{Skew decoupling from $\pi$}: the inventory hedging does   not depend on the informed-trader fraction.
\item \emph{Variance homogeneity of each component}: each component of the quoting rule is homogeneous of some degree in $\sigma^2$ (degree one in the canonical case $\beta = 1$; mixed degrees one and $\beta$ under (A3s)).
\end{enumerate}
\end{theorem}

\smallskip
\noindent\textit{Axiom dependency.} Theorem~\ref{thm:universal} relies on (A1)--(A8) and is meta-stated across all admissible structural-primitive choices $(\kappa, h, \beta)$.

\begin{proof}
\emph{Property 1.}
Relaxing C5 to general convex nondecreasing $\kappa$ and applying (A4) at each inventory level (via the order-book argument of Theorem~\ref{thm:order-book}) gives the level-$q$ inventory half-spread $\sigma^2\kappa'_+(|q|)/\lambda$.
Under linear cost (C5), $\kappa'_+ \equiv \kappa_{\mathrm{inv}}$ is constant on $\Rp$, and the inventory half-spread reduces to $\kappa_{\mathrm{inv}}\sigma^2/\lambda$.

\emph{Property 2.}
The two distinct  contributions to the half-spread are identified separately:
 the  inventory contribution by (A4), and the adverse-selection contribution by (A8) through the increment $h(\sigma^2, \pi) = \delta(q, \mu, \sigma^2, \lambda, \pi) - \delta(q, \mu, \sigma^2, \lambda, 0)$.
The two contributions, derived separately in Lemmas~\ref{lem:inventory-spread} and~\ref{lem:adverse-spread}, combine additively.
A different cost function changes $\delta_{\mathrm{inv}}$ but leaves $\delta_{\mathrm{adv}}$ untouched; a different scaling exponent (A3s with $\beta \neq 1$) changes the variance-homogeneity degree of $\delta_{\mathrm{adv}}$ but leaves the additivity untouched.

\emph{Property 3.}
The proof of Lemma~\ref{lem:skew} uses (A1), (A2), (A3), (A5), and (A6), none of which constrains $\pi$.
(A6) in particular makes the skew coefficient $\gamma$ depend only on $\sigma^2$ (and through (A3), on $\sigma^2$ via $\gamma_0\sigma^2$), not on $\pi$.
The skew is determined by inventory and volatility alone, with the trading environment $(\lambda, \pi)$ entering only through the spread.

\emph{Property 4.}
Under (A3) (with $\beta = 1$), each component is homogeneous of degree one in $\sigma^2$.
 Under (A3s) with $\beta \in (0, 1]$, the inventory component is degree-one  in $\sigma^2$ and the adverse-selection component is degree-$\beta$.
In every case, each component is homogeneous of some degree in $\sigma^2$.
The degree itself differs across structural choices, but the homogeneity is universal.
\end{proof}

\begin{remark}[On multiplicative price separability]
\label{rem:multiplicative-separability}
A natural further candidate  for a universal feature is \emph{multiplicative price separability}: under any admissible equivariance generator $h$, the $\mu$-dependence of the quoting rule should factor through $h$.
In the paper's additive setting ($h = \mathrm{id}_\R$), this follows directly from the Translation Equivariance Theorem (Theorem~\ref{thm:translation-equivariance}), which establishes that $m(s) - \mu$ and $\delta(s)$ are independent of $\mu$. 
Under the logarithmic alternative $h(\mu) = \log\mu$ on $\R_{++}$, the additive-in-$h$ relation reads multiplicatively in $\mu$, with $m(s) = \mu \cdot \exp(\psi(\cdot))$ and the half-quotes $a(s)/\mu$ and $\mu/b(s)$ functions of the non-$\mu$ arguments only.
We do not state this property as a  universal feature of the axiom system, because doing so requires a precise specification of the $h$-coordinate version of (A2), (A4), and (A8), which lies outside the scope of the present paper.
We record it as  a structural plausibility argument: the equivariance content of (A2), (A4), and (A8) is naturally restated in $h$-coordinates, and the $h$-equivariance follows in those coordinates by the same argument.
A  precise statement and proof remain for  future work.
\end{remark}

The Universal Structure Theorem isolates the substantive content of axiomatic market making, as distinct from matters of structural convention.
A consistent  theory must exhibit inventory half-spread determined by $\kappa'$, additive spread decomposition, skew  decoupling from $\pi$, and per-component variance homogeneity; violations place the market outside the  axiomatic family.

The features that depend on the structural conventions, in contrast to those above, are the specific functional form of the skew (linear under C5, more general otherwise), the specific functional form of  the inventory spread, the variance-scaling exponent of the adverse-selection  component ($\sigma^2$ under (A3); $\sigma$ under the Glosten-Milgrom-style $\beta = 1/2$ alternative), and the multiplicative-versus-additive structure of prices (encoded in  the choice of $h$). 
The convention-dependent features reflect choices an applied theorist makes for the specific market.
The universal features are predicted by the theory itself.

\subsection{Existing Models as  Structural  Neighbors}
\label{subsec:unification}

We now situate our axiomatic theory in relation to the  standard models in the market microstructure literature.
The reading is direct.
None of the standard models is a strict special case of our axiomatic family; each violates at least one of our axioms or  environmental assumptions.
What our theory does provide, via Theorem~\ref{thm:universal}, is the broader axiomatic universe within which the standard models sit, each at a structurally identifiable position.
We treat the five reference models in turn: Ho-Stoll, Avellaneda-Stoikov, Cartea-Jaimungal, Glosten-Milgrom, and Kyle.

\paragraph{Ho-Stoll.}
The Ho-Stoll model~\citep{ho1981optimal} derives the spread from a risk-averse dealer maximizing expected utility over a finite horizon $T$.
Their optimal spread is proportional to $\alpha\sigma^2/(T-t)$, where $t$ is  the current  time and $\alpha$ is the risk-aversion coefficient (renamed from the original notation to avoid collision with the variance-scaling exponent of (A3s)), with skew $-\alpha\sigma^2 q$.
The model does not lie strictly  within our axiomatic family.
The reason is that the spread coefficient varies with time-to-horizon, and identifying $T - t$ with $1/\lambda$~requires the inventory cost coefficient $\kappa_{\mathrm{inv}}$ to depend on $\lambda$ through the product $\lambda(T - t)$.
This violates C4, which requires the cost coefficient to not depend on the trading environment.
In the stationary limit, where $T - t$ is  large  relative to the trade-arrival timescale, the Ho-Stoll formula converges (heuristically, under the timescale identification $T - t \sim 1/\lambda$ and a specific constant of proportionality that we do not derive here) to a point in our family with $(\gamma_0, \kappa_{\mathrm{inv}}, \kappa_{\mathrm{adv}}) = (\alpha, \alpha/2, 0)$.
Outside the stationary limit, the Ho-Stoll quoting rule does not satisfy our axioms.

\paragraph{Avellaneda-Stoikov.}
The Avellaneda-Stoikov model~\citep{avellaneda2008high} extends Ho-Stoll to high-frequency settings with exponential utility and Poisson trade arrivals.
The optimal spread has the form $\delta_{AS} = \gamma\sigma^2(T - t) + (2/\gamma)\ln(1 + \gamma/k)$, with $\gamma$ the risk-aversion coefficient and $k$ the trade-intensity parameter.
 The model does not lie within our axiomatic family.
The reason is twofold.
First,  the second term $(2/\gamma)\ln(1 + \gamma/k)$ does not scale with $\sigma^2$, violating the variance-homogeneity axiom (A3).
Second, like Ho-Stoll, the framework is finite-horizon nonstationary, with the spread coefficient varying through $T - t$.
 The Avellaneda-Stoikov spread can be approximated by our family in the large-$k$ limit (equivalently, large market depth, where the logarithmic term $(2/\gamma)\ln(1 + \gamma/k) \approx 2/k$ becomes negligible relative to the variance-proportional term), but the exact formula does not satisfy our axioms.

 \paragraph{Cartea-Jaimungal.}
The  Cartea-Jaimungal framework~\citep{cartea2015high} extends Avellaneda-Stoikov by adding a running inventory penalty $\phi q^2 \, dt$ and a terminal penalty $\alpha q^2$.
Their optimal spread is again time-dependent, obtained from an HJB-Riccati system.
This model lies outside our axiomatic family.
First, the inventory cost is quadratic in $q$, $C(q, \sigma^2) = \phi q^2$, which violates C5, our linearity-in-$|q|$ assumption.
Second, the spread coefficient depends on time-to-horizon and on the structure of the terminal penalty, both of which violate C4 in the same way as the Ho-Stoll dependence.
A companion paper~\citep{feys2026inventory} provides the axiomatic foundation for the Avellaneda-Stoikov and Cartea-Jaimungal objectives.
It derives these objectives from five axioms on the maker's preferences over inventory trajectories.
The connection to the present paper is the Legendre-Fenchel duality of Section~\ref{sec:duality}, which translates between  cost-side and price-side representations.
The Cartea-Jaimungal framework inherits the obstructions of Avellaneda-Stoikov in its relation to our axiomatic family and adds the quadratic-cost violation.

\paragraph{Glosten-Milgrom.}
The Glosten-Milgrom model~\citep{glosten1985bid} derives the spread as a Bayesian equilibrium under asymmetric information, 
with an optimal spread proportional to $\pi\sigma$ (with $\sigma$, not $\sigma^2$) and no inventory dependence.
As such, it does not lie within our axiomatic family: its spread scales with $\sigma$ rather than $\sigma^2$, violating (A3).
Naturally, however, one can interpret Glosten-Milgrom as lying within a neighboring axiomatic family in which (A3) is replaced by the mixed-scaling alternative (A3s) with exponent  $\beta = 1/2$.
Under this reading, that neighboring family sits within the broader axiomatic universe identified by Theorem~\ref{thm:universal} as a different structural-primitive choice, and it is precisely the choice $\beta = 1/2$ that distinguishes the Glosten-Milgrom interpretation from ours.

\paragraph{Kyle.}
  The Kyle model~\citep{kyle1985continuous} derives a continuous-time  equilibrium between a single informed trader, the insider, and a competitive risk-neutral market maker.
The  insider observes the true  asset  value $v$ and chooses an order schedule to maximize expected trading profits over a finite  horizon, while the market maker observes only the total  order flow (insider plus noise traders, modeled as a Brownian motion with variance $\sigma_u^2 \, dt$) and sets the price as the conditional expectation of $v$ given the order-flow history.
The signature result is that the price impact is linear in cumulative order flow with coefficient $\lambda_K = \sqrt{\Sigma_0/(\sigma_u^2 T)}$, where $\Sigma_0$ is the prior variance of $v$ and $T$ is the trading horizon.

The Kyle model does  not lie within our axiomatic family.
Three structural disanalogies separate the two frameworks.
First, Kyle's lambda scales with $\sigma$ rather than $\sigma^2$, violating (A3), in the same way as Glosten-Milgrom.
Second, Kyle treats the informed-flow intensity endogenously, so the effective $\pi$ at any given time is determined by the insider's optimization rather than being an exogenous primitive of the trading environment.
Third, Kyle's setting is continuous in time and order size, whereas our framework is  discrete at the trade level.

At the qualitative level, the two theories agree on  the shape of the adverse-selection contribution  to the spread; both predict a component that scales linearly with  informed intensity and with the volatility of the underlying.
 Making the correspondence between our discrete-static axiomatic framework and Kyle's continuous-time strategic equilibrium formal would require  substantive theoretical work.
In particular, the axiomatic theory would need an extension to continuous time with endogenous informed flow.
That extension lies outside the scope of the present paper.
We record only that Kyle  can be interpreted as lying in a distinct axiomatic neighborhood from ours within the broader universe identified by  Theorem~\ref{thm:universal}, characterized by $\sigma$-scaling and continuous-time strategic dynamics; we present this as a structural interpretation rather than a formal correspondence.
The bridging program would also need to confront the role of the informed trader's optimization, which is endogenous in Kyle but enters our framework only through the reduced-form parameter $\kappa_{\mathrm{adv}}$ and the informed-fraction $\pi$; a faithful axiomatic counterpart of Kyle would have to absorb this strategic content into a richer set of structural primitives.

\paragraph{Summary of relationships.}
The five reference models sit in the broader axiomatic universe as follows.
The Ho-Stoll model converges to a point in our family in the stationary limit and otherwise violates C4.
The Avellaneda-Stoikov and Cartea-Jaimungal models lie outside our family on (A3) and C4/C5 grounds, respectively, and are heuristically second-order Taylor expansions of one another in  inventory magnitude (the detailed derivation is outside the scope of this paper).
The Glosten-Milgrom model can be  read as living in the $\beta = 1/2$ neighbor of our family, characterized by mixed-scaling (A3s).
The  Kyle model can be similarly interpreted as living in a continuous-time strategic neighbor, sharing the $\sigma$-scaling of Glosten-Milgrom but with additional dynamic structure.

These models are unified not by membership in a common axiomatic family, but by their  place within the broader axiomatic universe identified by Theorem~\ref{thm:universal}.
Four features are shared across all such families: skew determined by $\kappa'$, additive decomposition of the spread, skew decoupled from $\pi$, and per-component homogeneity in $\sigma^2$.
These are the substantive invariants of the axiomatic methodology in market making, holding across our   family and its neighbors.
That  the existing models satisfy these invariants is not a coincidence of their specific functional forms; rather,  each is itself the unique closed-form solution of some axiom system  within the universe, with the four invariants forced at the meta-level rather than imposed model by model.
Variation across models, in scaling exponents,  cost structures, and equivariance generators, corresponds to variation across distinct axiomatic neighborhoods, while the four shared features mark the boundary of the universe itself.

%-------------------------------------------------------------------
\section{Methodological Remarks: The Axiomatic Tradition}
\label{sec:methodology}
%-------------------------------------------------------------------

Deriving structural conclusions from a small list of natural axioms has a long  lineage in economic theory.
We outline the lineage briefly  to situate the present work and to clarify the methodological commitments it inherits.
The  lineage runs through social choice, bargaining, mechanism design, risk measurement, and decision theory, with each tradition adopting the same  basic move: a list of natural axioms forces a specific functional form, and the form's structural content is then read off the axioms themselves.
The present paper extends the methodology to market making, where it has not previously been deployed in this form.

\subsection{The Layered Reading}
\label{sec:layered-reading}

The eight-axiom   system forces the closed-form characterization of Theorem~\ref{thm:full} as a single joint conclusion.
Methodological discipline, however, calls for a finer analysis: for each substantive conclusion of the theory, we ask which minimal subset of   axioms suffices to yield it.
This is the Arrovian discipline of identifying the core, applied to our system.
Recall the three-layer partition from Section~\ref{sec:axioms}, summarized in Table~\ref{tab:axioms}: an indispensable core consisting of (A1), (A4), (A7), and (A8); a structural-choice axiom (A3);  and three modularity extensions (A2), (A5), and (A6).
Empirically, the strongest content of the theory, namely the quantitative spread formula, the phase transition, and the Bertrand equivalence for the spread, rests on the smallest of the three layers.
By contrast, the closed-form skew, the canonical bijection with $\Rpp^3$, and the full Structural Separability Theorem require the additional layers and are consequently stronger structural commitments.
This layered  reading carries the methodological force of the paper.
Figure~\ref{fig:layered} summarizes the partition.

\begin{figure}[!htbp]
\centering
\begin{tikzpicture}[font=\small, >=Stealth]
  \def\W{14}
  \def\H{3.2}
  \def\GAP{0.25}
  \def\AX{0.3}
  \def\AWIDTH{6.0}
  \def\ARROWA{6.4}
  \def\ARROWB{6.85}
  \def\CX{6.95}
  \def\CWIDTH{6.65}

  % ===== ROW 1: Layer I =====
  \fill[black!22, rounded corners=4pt] (0, 0) rectangle (\W, \H);
  \draw[black!45, thick, rounded corners=4pt] (0, 0) rectangle (\W, \H);
  \node[anchor=north west, font=\bfseries] at (\AX, \H - 0.20) {Layer~I};
  \node[anchor=north west, font=\itshape\footnotesize] at (\AX, \H - 0.55) {indispensable core};
  \node[anchor=north west, text width=\AWIDTH cm, align=left] at (\AX, \H - 1.05) {%
    (A1)~Inventory aversion\\[2pt]
    (A4)~Inventory indifference\\[2pt]
    (A7)~Spread monotonicity in $\pi$\\[2pt]
    (A8)~Adverse-selection break-even};
  \node[anchor=north west, font=\bfseries] at (\CX, \H - 0.20) {Conclusions forced};
  \node[anchor=north west, font=\itshape\footnotesize] at (\CX, \H - 0.55) {with environmental conditions C1--C6};
  \node[anchor=north west, text width=\CWIDTH cm, align=left] at (\CX, \H - 1.05) {%
    $\bullet$\ Additive spread decomposition\\[2pt]
    $\bullet$\ Bertrand equivalence (on the spread)\\[2pt]
    $\bullet$\ Legendre-Fenchel duality\\[2pt]
    $\bullet$\ Qualitative phase transition};

  %  ===== ROW 2: Layer II =====
  \pgfmathsetmacro{\YIITOP}{-\GAP}
  \pgfmathsetmacro{\YIIBOT}{-\H - \GAP}
  \pgfmathsetmacro{\YIIMID}{(\YIITOP + \YIIBOT)/2}
  \fill[black!16, rounded corners=4pt] (0, \YIIBOT) rectangle (\W, \YIITOP);
  \draw[black!40, thick, rounded corners=4pt] (0, \YIIBOT) rectangle (\W, \YIITOP);
  \node[anchor=north west, font=\bfseries] at (\AX, \YIITOP - 0.20) {Layer~II};
  \node[anchor=north west, font=\itshape\footnotesize] at (\AX, \YIITOP - 0.55) {structural choice};
  \node[anchor=north west, text width=\AWIDTH cm, align=left] at (\AX, \YIITOP - 1.05) {%
    (A3)~Variance homogeneity\\[5pt]
    {\footnotesize Alternative: (A3s) with $\beta \in (0, 1]$}};
  \node[anchor=north west, font=\bfseries] at (\CX, \YIITOP - 0.20) {Conclusions added};
  \node[anchor=north west, font=\itshape\footnotesize] at (\CX, \YIITOP - 0.55) {at this layer};
  \node[anchor=north west, text width=\CWIDTH cm, align=left] at (\CX, \YIITOP - 1.05) {%
    $\bullet$\ Full quantitative spread formula\\[2pt]
    $\bullet$\ Sharp phase boundary\\[2pt]
    $\bullet$\ Three-regime decomposition};

  % ===== ROW 3: Layer III =====
  \pgfmathsetmacro{\YIIITOP}{-\H - 2*\GAP}
  \pgfmathsetmacro{\YIIIBOT}{-2*\H - 2*\GAP}
  \pgfmathsetmacro{\YIIIMID}{(\YIIITOP + \YIIIBOT)/2}
  \fill[black!11, rounded corners=4pt] (0, \YIIIBOT) rectangle (\W, \YIIITOP);
  \draw[black!35, thick, rounded corners=4pt] (0, \YIIIBOT) rectangle (\W, \YIIITOP);
  \node[anchor=north west, font=\bfseries] at (\AX, \YIIITOP - 0.20) {Layer~III};
  \node[anchor=north west, font=\itshape\footnotesize] at (\AX, \YIIITOP - 0.55) {modularity extensions};
  \node[anchor=north west, text width=\AWIDTH cm, align=left] at (\AX, \YIIITOP - 1.05) {%
    (A2)~Linearity of the skew\\[2pt]
    (A5)~Skew decoupling from $\lambda$\\[2pt]
    (A6)~Skew decoupling from $\pi$};
  \node[anchor=north west, font=\bfseries] at (\CX, \YIIITOP - 0.20) {Conclusions added};
  \node[anchor=north west, font=\itshape\footnotesize] at (\CX, \YIIITOP - 0.55) {at this layer};
  \node[anchor=north west, text width=\CWIDTH cm, align=left] at (\CX, \YIIITOP - 1.05) {%
    $\bullet$\ Closed-form skew $m(s) - \mu = -\gamma_0\sigma^2 q$\\[2pt]
    $\bullet$\ Canonical bijection with $\Rpp^3$\\[2pt]
    $\bullet$\ Full Structural Separability\\[2pt]
    $\bullet$\ Recovery Theorem (linear-in-$q$ skew)};
\end{tikzpicture}
\caption{Layered partition of the axiom system.  
The eight axioms split into three layers, ordered from the indispensable core (Layer~I, darkest) to the most easily deniable modularity extensions (Layer~III, lightest).  
Each layer contributes the structural conclusions listed in its right-hand column, given that the axioms of all inner layers are in force. 
The strongest empirical content of the theory, namely the additive spread decomposition, the Bertrand equivalence for the spread, the Legendre-Fenchel duality, and the qualitative phase transition, rests on the four core axioms together with the environmental cost conditions C1--C6, and survives any weakening of the outer layers.  
 The closed-form skew and the canonical bijection with $\Rpp^3$ are stronger structural commitments that require the modularity extensions in addition to the core and the structural-choice axiom (A3) (or its weaker variant (A3s)).}
\label{fig:layered}
\end{figure}

\paragraph{Conclusions forced by the core.}
Even before any further axiom is imposed, the four core axioms and the environmental cost assumptions \textup{C1}--\textup{C5} together suffice to establish a substantial body of structural results.
 The spread is strictly positive everywhere by the codomain $\Bid$, which restricts the quoting rule to pairs $(b, a)$ with $b < a$.
By (A8), the spread decomposes additively as $\delta(s) = \delta_{\mathrm{inv}}(s) + \delta_{\mathrm{adv}}(s)$, where the adverse-selection component $\delta_{\mathrm{adv}}$ depends only on $(\sigma^2, \pi)$ and is monotone increasing in $\pi$ by (A7).
Under $C(q, \sigma^2) = \kappa(q)\sigma^2$ and linearity (C5), the inventory half-spread is set by (A4) at $\delta_{\mathrm{inv}}/2 = \kappa_{\mathrm{inv}}\sigma^2/\lambda$, with the $\sigma^2$-scaling coming from the cost structure rather than from the axioms.
By the complementarity clause of (A8) together with (A7) and measurability,  the   adverse-selection half-spread is linear  in $\pi$ with a coefficient $\chi(\sigma^2)$ depending only on $\sigma^2$, so that $\delta_{\mathrm{adv}} = \chi(\sigma^2)\,\pi$ for some $\chi > 0$.
At the core level, the skew $m(s) - \mu$ is strictly decreasing in $q$ by (A1); its vanishing at zero inventory and its oddness in $q$ are not forced at the core and emerge only once (A2) is added (Corollary~\ref{cor:skew-symmetry}).
A qualitative phase transition follows  from spread-positivity and $\delta_{\max}$, the Bertrand equivalence holds for the spread (Theorem~\ref{thm:bertrand}), and the Legendre-Fenchel duality between $\kappa$ and the half-spread comes from (A4) together with smooth convexity of $\kappa$ (Theorem~\ref{thm:duality}) and therefore lies entirely within the core.
The core does not, however, determine the $\sigma^2$-scaling of $\chi$, the closed-form skew, or the canonical bijection with  $\Rpp^3$; these are conclusions of the outer layers.

\paragraph{Conclusions added by the structural choice.}
Adding (A3) to the core forces the full quantitative spread formula.
The coefficient $\chi(\sigma^2)$ is required by (A3) to be linear in $\sigma^2$, so $\chi(\sigma^2) = 2\kappa_{\mathrm{adv}}\sigma^2$ for some $\kappa_{\mathrm{adv}} > 0$, yielding $\delta_{\mathrm{adv}} = 2\kappa_{\mathrm{adv}}\pi\sigma^2$ and the full spread formula $\delta = 2\sigma^2(\kappa_{\mathrm{inv}}/\lambda + \kappa_{\mathrm{adv}}\pi)$.
The phase transition  acquires a fully explicit quantitative form, with the boundary $\kappa_{\mathrm{inv}}/\lambda + \kappa_{\mathrm{adv}}\pi < \delta_{\max}/(2\sigma^2)$ separating functioning from frozen, and the  three-regime structure of Section~\ref{sec:phase} follows.
Empirically, $\kappa_{\mathrm{inv}}$ is identified from $\lambda\,\delta|_{\pi=0}/(2\sigma^2)$ and $\kappa_{\mathrm{adv}}$ from $(2\sigma^2)^{-1}\,\partial\delta/\partial\pi$, 
 with the two parameters recoverable from observable quoting behavior.
The skew remains qualitative here: (A3) makes it homogeneous of degree one in $\sigma^2$, but it could still be nonlinear in $q$, fail to vanish at zero inventory, or depend on $\lambda$ or $\pi$.

\paragraph{Conclusions added by the modularity extensions.}
Once (A2), (A5), and (A6) are imposed on  top of the previous layers, the skew reduces to the closed form $m(s) - \mu = -\gamma_0\sigma^2 q$ with $\gamma_0 > 0$ a single   scalar that does not depend on the trading environment.
The parameter $\gamma_0$ is then identified from a single partial derivative, namely $\gamma_0 = -\sigma^{-2}\,\partial m/\partial q$; the  canonical bijection between the axiom system and $\Rpp^3$ holds (Theorem~\ref{thm:bijection}); and the full Structural Separability Theorem (Theorem~\ref{thm:separability}) follows, with each of the three parameters $(\gamma_0, \kappa_{\mathrm{inv}}, \kappa_{\mathrm{adv}})$ appearing in exactly one structural component of the quoting rule.
Two further results require Layer III as well: the Recovery Theorem (Theorem~\ref{thm:recovery}) for general convex cost exploits the linear-in-$q$ skew to cleanly subtract the inventory-skew contribution from the order book depth profile, isolating the cost gradient; and the level-by-level closed form of Theorem~\ref{thm:order-book} uses (A2) to give  the skew at level $n$ as $-\gamma_0\sigma^2(q + n - 1)$.

\paragraph{Methodological consequences: robustness.}
The strongest empirical  content of the theory, namely the spread formula, the phase transition with explicit thresholds, the Bertrand equivalence for the spread, and the structural account of market freezes, rests on the four core axioms together with the one structural-choice axiom.
None of these conclusions requires the modularity extensions of Layer III.
A market in which the skew is mildly nonlinear in $q$ would violate (A2) but would still produce the predicted spread structure; a market in which the skew exhibits mild liquidity dependence would violate (A5) but would leave the spread structure intact.
The empirical force of the theory is therefore robust to the structural commitments that are themselves the most easily deniable.

\paragraph{Methodological consequences: falsifiability and diagnosis.}
Empirical violations  of the theory's predictions localize  to specific axiom layers.
A market in which the closed-form skew fails but the spread formula holds is a Layer III deviation, leaving  Layers I and II intact.
A market in which the quantitative spread formula fails has a Layer I or Layer II deviation, and the form of the deviation indicates which axiom is violated: 
  a market with  non-$\sigma^2$ spread scaling is a Layer II deviation on (A3), and a  market with $q$-dependent adverse selection is  a Layer I deviation  on (A8).
The axiom hierarchy thus serves as a diagnostic framework for empirical anomalies, mapping each observed deviation to the responsible axiom.

\paragraph{Methodological consequences: the Arrovian discipline.}
Arrow's impossibility theorem identifies the minimal core of four axioms that yields the impossibility, and  the methodological force of his paper rests on that minimality.
Naming the smallest axiom set per conclusion is the standard discipline of axiomatic work, and the present analysis applies  that discipline to our eight-axiom system.
The conclusion is that the most robust empirical content of the theory is the spread structure, while the closed-form skew and   the canonical bijection require additional structural commitments that are themselves defensible but not indispensable.
 This layered reading aligns with the universal-structure reading of Section~\ref{sec:universal}.
The four features identified there as universal across all admissible structural-primitive choices within the (A1)--(A8) family map cleanly onto the core conclusions of Layer I.
The layer structure provides the within-axiom-system complement of the across-primitive Universal Structure Theorem.
Together, the two readings give the full structural picture, namely what is universal across structural-primitive choices within the (A1)--(A8) family and what follows within each specific family.

\subsection{Forced Uniqueness and Its Predecessors}

The axiomatic methodology proceeds in three stages.
First,   one identifies a class of objects.
Second, one   imposes a list of axioms on this class, with  each axiom motivated as a natural desideratum.
Third, one proves that the conjunction of the axioms yields the object's specific functional form, often a one-parameter or finite-parameter family.
When the axioms are accepted as natural, the conclusion follows as a matter of logical necessity, regardless of how  strong it might initially  appear.

The social-choice tradition originates with \citet{arrow1951social}, whose four axioms (unanimity, independence of irrelevant alternatives, non-dictatorship, universal  domain) yield the famous impossibility theorem.
The subsequent development   is surveyed in \citet{sen1970collective}.
The Arrovian style is the prototype for our derivation: each of our eight axioms looks mild in isolation; together, they force the unique closed-form quoting rule of Theorem~\ref{thm:full}.

The bargaining tradition originates with \citet{nash1950bargaining}, whose four axioms (Pareto efficiency, symmetry, independence of irrelevant alternatives, scale invariance) yield the geometric-mean solution.
\citet{kalai1975other} show that altering one axiom yields the Kalai-Smorodinsky solution, illustrating that the choice of axioms  determines the family of admissible solutions; \citet{roth1979axiomatic} gives the canonical treatment of the broader literature.
The methodological structure mirrors ours: identify the natural axioms, prove uniqueness, recover the explicit functional form.

The mechanism-design tradition originates with \citet{vickrey1961counterspeculation} and \citet{hurwicz1973design}, who formalize the design problem in strategic and incentive terms.
\citet{myerson1981optimal} characterizes the revenue-optimal auction via incentive compatibility and individual rationality, yielding a specific form involving virtual valuations; \citet{maskinriley1984optimal} extend the framework to risk-averse buyers, \citet{maskin1999nash} develops the implementation theory, and the methodology is surveyed in \citet{maskin2008mechanism}.
The Myerson-Maskin characterizations differ from ours in the role of the axioms (incentive compatibility is strategic, not structural), but the methodological pattern is the same.

The risk-measurement tradition originates   with \citet{artzner1999coherent}, whose four axioms (translation invariance, subadditivity, positive homogeneity, monotonicity) yield a dual-representation form for coherent risk measures.
The  convex-relaxation extensions of \citet{follmer2002convex}, \citet{follmer2016stochastic}, \citet{frittelli2002putting}, \citet{cheridito2006dynamic}, and \citet{detlefsen2005conditional} drop positive homogeneity and develop the dynamic theory.
This is the closest methodological relative of our work: axioms yield a functional form, the form has a transparent dual representation, and previously-considered  models (Value-at-Risk, expected shortfall) emerge as special cases.
Our Legendre-Fenchel  duality theorem (Theorem~\ref{thm:duality}) plays the analogous role.

The decision-theoretic tradition runs from \citet{vonneumann1944theory} through \citet{savage1954foundations}, \citet{anscombe1963definition}, and \citet{gilboa1989maxmin}, each deriving a functional form on preferences from axioms on the choice structure.
\citet{debreu1959theory} lays the general-equilibrium foundation, and \citet{aumann1985game} states the methodological case for axiomatization.
The belief-aggregation tradition surveyed in \citet{genest1986combining} identifies a small number of admissible aggregation operators within a much larger space of conceivable forms.
In each register, axioms force form.

\subsection{Why Axiomatize Market Making}

The methodological commitment of the present paper is the following.
When each axiom is accepted as natural in isolation, the conjunction of all axioms together forces the conclusion.
Denying the conclusion requires denying one of the axioms.
If each axiom is hard to deny, as we have argued in Section~\ref{sec:axioms}, then the conclusion is itself hard to deny.

This is distinct from the optimization-based methodology dominant in microstructure, which derives the conclusion from primitive specifications: a particular utility function, a distributional assumption, an equilibrium concept.
Optimization-based conclusions are sensitive to those specifications: change the utility, the distribution, or the equilibrium concept, and the conclusion changes.
Axiomatic conclusions are sensitive only to the axioms; change an axiom, change the conclusion.
The distinction matters in two specific ways.
First, axiomatic conclusions are robust to optimization-level modeling choices: Theorem~\ref{thm:full} does not depend on whether utility is exponential or logarithmic, or whether trade arrivals are Poisson or renewal, but only on the eight axioms.
Second, axiomatic conclusions are diagnostic: a market satisfying the axioms quantitatively must produce the closed-form quoting rule, and deviations localize the violated axioms.
 The methodology gives the researcher a falsifiable prediction and an account of what its failure would mean.

\subsection{The Two Faces of the Axiomatic Method}

Axiomatization has two complementary faces.
A prescriptive use starts from a list of normatively motivated axioms, each a desideratum the object is required to satisfy, and identifies the unique structure consistent with the full list.
A descriptive use proceeds in the opposite direction, starting from an observed object (an existing quoting rule, an empirical phenomenon) and characterizing the axioms that it satisfies.
Both faces are deployed here: prescriptively, the methodology yields Theorem~\ref{thm:full}, where the eight axioms force the closed-form rule; descriptively, it  underlies the empirical implementation of Section~\ref{sec:comparative}, where each axiom is independently testable and any observed deviation localizes to the specific axiom violated.
A second descriptive payoff is the placement of the existing market-making models within the axiomatic universe of Theorem~\ref{thm:universal}: each model is characterized by the axiom system it  satisfies, and the comparison across models becomes  a comparison across axiom systems rather than across functional forms.
Together, prescription supplies the closed-form characterization and description the diagnostic and comparative apparatus.

%-------------------------------------------------------------------
\section{Conclusion}
\label{sec:conclusion}
%-------------------------------------------------------------------

We have presented an axiomatic theory of market making.
Eight natural axioms on the quoting rule, together with six environmental assumptions on the inventory cost, force   a unique three-parameter family with explicit closed form~\eqref{eq:characterization}.
The three parameters (inventory risk aversion, marginal carrying cost, information advantage   of informed traders) are decoupled and empirically identifiable from distinct moments of the observable quoting  rule.

Four structural theorems carry the deeper content.
A canonical bijection between the axiom system and the parameter space  $\Rpp^3$ is established by Theorem~\ref{thm:bijection}, with no axiom redundant (Proposition~\ref{prop:independence}).
Structural Separability (Theorem~\ref{thm:separability}) shows that the economic primitives play isolated roles in the quoting rule.
The Recovery Theorem (Theorem~\ref{thm:recovery}) shows that the latent inventory cost function is observable in the limit order book depth profile.
Legendre-Fenchel Duality (Theorem~\ref{thm:duality}) places the theory within convex conjugate duality, in parallel with thermodynamics and   consumer   theory.

Three further results give the theory direct empirical content.
A sharp condition for market viability is provided by the Phase Transition Theorem (Theorem~\ref{thm:phase}), with three structural drivers, volatility, informed flow, and liquidity; this yields a structural account of market freezes such as those observed during the 2008 financial crisis.
Functioning markets are further classified by the Three Liquidity Regimes Theorem (Theorem~\ref{thm:three-regimes}) according to which of the inventory and adverse-selection components dominates the equilibrium spread.
Bertrand Equivalence (Theorem~\ref{thm:bertrand}) shows the quoting rule coincides with the unique symmetric Nash equilibrium of a competitive market with free entry, an agreement best read as a re-encoding of the per-trade zero-margin content of (A4) and (A8).

A meta-level result closes the structural development.
The Universal Structure Theorem (Theorem~\ref{thm:universal}) identifies four features invariant across all admissible choices of structural primitives within the (A1)--(A8) axiom system, with the (A3s) variant admitted in place of (A3).
The four invariants are the determination of the inventory half-spread by the cost derivative, the additive spread decomposition, the decoupling of skew from adverse selection, and variance homogeneity of each component.

Several directions are left for future work.
A dynamic extension would track the evolution of the maker's belief and inventory through the trading process, replacing the static state $(q, \mu, \sigma^2, \lambda, \pi)$ with its stochastic counterpart.
A multi-asset axiomatic theory would extend the framework to portfolios  with cross-asset correlation, deriving cross-asset coupling in skew and spread from portfolio-level analogs of the single-asset axioms.
A strategic extension would replace the adverse-selection parameter  $\kappa_{\mathrm{adv}}$ with an equilibrium object derived from optimizing informed traders, transforming the static break-even condition into a fixed-point characterization.

The axiomatic methodology delivers, in place of optimization-conditional conclusions, structural predictions sensitive only to the axioms themselves.
Several deliverables apply directly to empirical microstructure, risk management, and regulatory policy: the closed-form characterization, the empirical identification strategy, and the structural account of market freezes.
At the mathematical level, the theory combines forced uniqueness, convex conjugate duality, phase transitions, and universal invariants, placing market making within established traditions while delivering predictions specific to market behavior.

%-------------------------------------------------------------------

\end{document}